\documentclass[11pt]{article}

\usepackage{graphicx,amsmath,amsfonts,amssymb,bm,hyperref,url,breakurl,epsfig,epsf,color,fullpage,MnSymbol,mathbbol,fmtcount,algorithmic,algorithm,semtrans,cite,caption,subcaption,multirow,xcolor}
  
  \usepackage{mathtools}

\usepackage{pbox}

\usepackage{titlesec}
\usepackage{booktabs}
\usepackage{pifont}

\usepackage{multirow}

\usepackage{tikz}
\usepackage{pgfplots}
\usetikzlibrary{pgfplots.groupplots}

\titleformat{\paragraph}
{\normalfont\normalsize\bfseries}{\theparagraph}{1em}{}
\titlespacing*{\paragraph}
{0pt}{3.25ex plus 1ex minus .2ex}{1.5ex plus .2ex}

\usepackage{movie15}

\usepackage{cite}
\usepackage{caption}
\usepackage[bottom,hang,flushmargin]{footmisc} 

\usepackage{hyperref}
\definecolor{darkred}{RGB}{150,0,0}
\definecolor{darkgreen}{RGB}{0,150,0}
\definecolor{darkblue}{RGB}{0,0,200}
\hypersetup{colorlinks=true, linkcolor=darkred, citecolor=darkgreen, urlcolor=darkblue}

\newtheorem{theorem}{Theorem}[section]
\newtheorem{lemma}[theorem]{Lemma}
\newtheorem{corollary}[theorem]{Corollary}

\newtheorem{definition}[theorem]{Definition}

\newcommand{\eps}{\varepsilon}

\newcommand{\beq}{\begin{equation}}
\newcommand{\eeq}{\end{equation}}

\newcommand{\nn}{\nonumber}

\newcommand{\A}{{\mtx{A}}}

\newcommand{\Iden}{{\mtx{I}}}

\newcommand{\z}{{\mtx{z}}}

\newcommand{\Cc}{\mathcal{C}}
\newcommand{\Bc}{\mathcal{B}}
\newcommand{\Sc}{\mathcal{S}}
\newcommand{\Dc}{\mathcal{D}}

\newcommand{\Pc}{\mathcal{P}}

\newcommand{\Nn}{\mathcal{N}}

\newcommand{\vb}{{\mtx{v}}}
\newcommand{\ub}{\mtx{u}}
\newcommand{\w}{\mtx{w}}

\newcommand{\li}{\left<}
\newcommand{\ri}{\right>}

\newcommand{\ab}{\vct{a}}

\newcommand{\h}{\vct{h}}
\newcommand{\g}{\vct{g}}

\newcommand{\Tc}{\mathcal{T}}

\newcommand{\des}{{\z}}

\newcommand{\zeronorm}[1]{\left\|#1 \right\|_{\ell_0}}

\newcommand{\opnorm}[1]{\left\|#1\right\|}

\newcommand{\onenorm}[1]{\left\|#1\right\|_{\ell_1}}
\newcommand{\twonorm}[1]{\left\|#1\right\|_{\ell_2}}
\newcommand{\tn}[1]{\left\|#1\right\|_{\ell_2}}

\newcommand{\abs}[1]{\left|#1\right|}

\newcommand{\comment}[1]{}

\newcommand{\x}{\vct{x}}
\newcommand{\y}{\vct{y}}

\definecolor{emmanuel}{RGB}{255,127,0}

\newcommand{\R}{\mathbb{R}}
\newcommand{\Pro}{\mathbb{P}}

\newcommand{\E}{\operatorname{\mathbb{E}}}

\newcommand{\vct}[1]{\bm{#1}}
\newcommand{\mtx}[1]{\bm{#1}}

\numberwithin{equation}{section} 

\def \endprf{\hfill {\vrule height6pt width6pt depth0pt}\medskip}
\newenvironment{proof}{\noindent {\bf Proof} }{\endprf\par}

\newcommand*\samethanks[1][\value{footnote}]{\footnotemark[#1]}

\title{
Sharp Time--Data Tradeoffs for Linear Inverse Problems}
\author{Samet Oymak\thanks{Department of Electrical Engineering and Computer Science, UC Berkeley, Berkeley CA}~\thanks{Simons Institute for the Theory of Computing, UC Berkeley, Berkeley CA}  \quad  Benjamin Recht\samethanks[1]~\thanks{Department of Statistics, UC Berkeley, Berkeley CA} \quad Mahdi
  Soltanolkotabi\thanks{Ming Hsieh Department of Electrical Engineering, University of Southern California, Los Angeles, CA} }

\date{November 2015}

\begin{document}
\maketitle
\begin{abstract}

In this paper we characterize sharp time-data tradeoffs for optimization problems used for solving linear inverse problems.  We focus on the minimization of a least-squares objective subject to a constraint defined as the sub-level set of a penalty function.  We present a unified convergence analysis of the gradient projection algorithm applied to such problems.  We sharply characterize the convergence rate associated with a wide variety of random measurement ensembles in terms of the number of measurements and structural complexity of the signal with respect to the chosen penalty function.  The results apply to both convex and nonconvex constraints, demonstrating that a linear convergence rate is attainable even though the least squares objective is not strongly convex in these settings. When specialized to  Gaussian measurements our results show that such linear convergence occurs when the number of measurements is merely $4$ times the minimal number required to recover the desired signal at all (a.k.a.~the phase transition).  We also achieve a slower but geometric rate of convergence precisely above the phase transition point. Extensive numerical results suggest that the derived rates exactly match the empirical performance.

\end{abstract}  
  
\section{Introduction}

In science and engineering today we gather data at unprecedented scales. Yet in many applications ranging from imaging to online advertisement and financial trading we are interested in making inference and predictions under a fixed time budget. Efficient learning from these large and high-dimensional datasets poses new challenges
\begin{itemize}
\item What algorithms should we use under a fixed time budget? 
\item How much of the data should we use? Should we use all of the data or just parts of it?
\item How many passes (or iterations) of the algorithm is required to get to an accurate solution?
\end{itemize}
At the heart of answering these questions is the ability to predict runtime of optimization algorithms as a function of the required accuracy and the size of data. That is, we need to understand precise tradeoffs between run time, data size and accuracy of various optimization algorithms.

In this paper we characterize sharp time-data tradeoffs for optimization problems used for solving linear inverse problems. The last decade has witnessed a flurry of activity in understanding when and how it is possible to solve such problems. Based on this growing literature a clear picture has emerged as to when convex programming techniques are able to find a unique structured solution to under-determined linear systems with generic coefficients. Often a convex function is minimized subject to linear measurement constraints where the convex cost function captures the ``structure" complexity of the unknown signal. The state of the art literature on this topic characterizes a structure--data complexity trade off, in which the number of linear measurements or \emph{data complexity} is related to a quantity which represents the \emph{structural complexity} of the signal. Indeed, there is a precise curve which characterizes the minimum number of measurements or data as a function of the structural complexity of the unknown solution. This curve is an asymptotic phase transition of sorts for convex programs, on one side of this curve convex programming succeeds in exactly recovering the unknown signal on the other side it fails with high probability. Thus providing a precise tradeoff between data complexity (number of measurements) and structural complexity.

Despite the tremendous amount of progress on data--structure tradeoffs for convex programs much less is known about the role played by computation. There are, of course, various ways to solve a given convex program (e.g. interior point methods, gradient descent, etc). However, the rate of convergence of various convex solvers as well as how this rate depends on the data complexity (number of measurements) is not well understood.\footnote{We would like to point out that there are a few recent papers that address computational issues and we shall review this literature in greater detail in Section \ref{sec:prior}.} Confounding the matter even further the convex cost functions used in linear inverse problems are often non-smooth and not strongly convex so that crude bounds on the rate of convergence of these solvers as predicted by traditional convex optimization literature are often very pessimistic and rather far from the empirically observed rates of convergence.

Convex cost functions are not the only way to impose structure. Indeed in some applications it may be more appropriate to use non-convex cost functions or deploy greedy algorithms to accurately capture the structure of the unknown signal. This paper presents a unified theoretical framework for convergence rates of projected gradient methods for finding structured solutions to under-determined linear equations. Our theoretical framework is very general and covers projection on both convex and non-convex sets.

For convex sets we precisely characterize rates of convergence of various projected gradient algorithms as a function of the ambient dimension, the number of linear measurements and the structural complexity of the solution set. We prove that our rates of convergence for these algorithms come with explicit constants that are rather sharp for linear inverse problems. Indeed, based on our detailed numerical results our rates are loose by at most a factor of $1.18$. As a result our work suggests very precise time--data tradeoffs for linear inverse problems.

We also characterize the rates of convergence of projected gradient algorithms to the unknown structured solution even when the structure is best characterized by a non-convex set. Examples include imposing sparsity via $\ell_0$ or $\ell_p$ balls with $p<1$ that are known to be non-convex. Establishing sharp convergence rates to the unknown solution for non-convex sets is particularly challenging. Indeed for non-convex sets (unlike convex sets) it is not even clear that such projected gradient algorithms should be able to find a sensible solution. Such algorithms originate from non-convex optimization problems and one may expect them to get trapped in local optima. Perhaps unexpectedly, we show that this is not the case and even these non-convex projected gradient algorithms converge to the unknown signal. Furthermore, as in the convex case we provide sharp rates of convergence for these algorithms.

\subsection{Structured signal recovery from linear measurements}
\label{sec:strucsig}
We wish to discern an unknown but ``structured" signal $\vct{x}\in\R^n$ from $m$ linear measurements 
of the form $\vct{y}=\mtx{A}\vct{x}$.  $\mtx{A}\in\R^{m\times n}$ is the measurement matrix and $\vct{y}\in\R^m$ are the measurements. We wish to estimate the unknown signal $\vct{x}$ from such linear measurements. However, in the applications of our interest typically the number of equations $m$ is significantly smaller than the number of variables $n$ so that there are infinitely many solutions obeying the linear constraints. However, it may still be possible to recover the signal by exploiting knowledge of its ``structure". To this aim, let $f:\R^n\rightarrow \R$ be a cost function that reflects some notion of ``complexity" of the ``structured" solution. It is natural to use the following optimization problem to recover the signal.
\begin{align}
\label{opt1}
\hat{\vct{x}}=\underset{\vct{z}\in\R^n}{\arg\min}\quad f(\vct{z})\quad\text{subject to}\quad \vct{y}=\mtx{A}\vct{x}.
\end{align}

In fact in practical applications there is also some noise in our measurements and we observe noisy samples 
\begin{align}
\label{noisesamp}
\vct{y}=\mtx{A}\vct{x}+\vct{w},
\end{align}
where $\vct{w}\in\R^m$ represents the noise. It is then natural to modify \eqref{opt1} and solve
\begin{align}
\label{mainalgopt}
\twonorm{\vct{y}-\mtx{A}\vct{z}}^2\quad\text{subject to}\quad f(\vct{z})\le R,
\end{align}
where $R$ is a tuning parameter. A standard approach to solve this problem is via projected gradient descent. In particular, define the set 
\begin{align*}
\mathcal{K}=\{\vct{z}\in\R^n:\text{ }f(\vct{z})\le R\}.
\end{align*}
Throughout we shall use $\mathcal{P}_{\mathcal{K}}(\vct{z})$ to denote the Euclidean projection of $\vct{z}$ onto the set $\mathcal{K}$. That is,
\begin{align*}
\mathcal{P}_{\mathcal{K}}(\vct{z})=\underset{\bar{\vct{z}}\in\mathcal{K}}{\arg\min} \twonorm{\vct{z}-\bar{\vct{z}}}^2.
\end{align*}
Starting from an initial point $\vct{z}_0$ (often set to $0$). The Projected Gradient Descent (PGD) algorithm proceeds as follows
\begin{align}
\label{updateoflin}
\vct{z}_{\tau+1}=\mathcal{P}_{\mathcal{K}}\left(\vct{z}_\tau+\mu\mtx{A}^*(\vct{y}-\mtx{A}\vct{z}_{\tau})\right).
\end{align}
We note that without the projection step this iteration update is just gradient descent on the quadratic objective with learning parameter $\mu$. For the sake of exposition, throughout most of  this paper we shall assume that the tuning parameter is optimally tuned so that $R=f(\vct{x})$. We shall see later on in Section \ref{sec:stabilitytoparam} how to relax this assumption. 

\subsection{Precise rates of convergence using cone-restricted eigenvalues}

We wish to characterize the rates of convergence for the projected gradient update \eqref{updateoflin}.  Using no prior information, standard analysis suggests that this method will converge for convex $f$ at a rate of $O(1/\tau)$ where $\tau$ is the number of iterations.   For non-convex functions $f$, it is not immediately clear that this iteration will converge at all. To get a sense of what properties of the function $f$ and the structure of the signal affect the rate of convergence let us start with a simple example where there are more measurements than unknowns. Also assume that the signal has no structure or equivalently $f(\vct{z})=0$ for all $\vct{z}$. In this case the update takes the form
\begin{align*}
\vct{z}_{\tau+1}=\vct{z}_\tau+\mu\mtx{A}^*(\vct{y}-\mtx{A}\vct{z}_\tau).
\end{align*}
The latter is equivalent to
\begin{align}
\label{mytempupdate}
\vct{z}_{\tau+1}-\vct{x}=(\mtx{I}-\mu\mtx{A}^*\mtx{A})(\vct{z}_\tau-\vct{x})-\mu\mtx{A}^*\vct{w}.
\end{align}
This formula implies the following naive convergence guarantee
\begin{align}
\label{simplecase}
\twonorm{\vct{z}_{\tau+1}-\vct{x}}\le \opnorm{\mtx{I}-\mu\mtx{A}^*\mtx{A}}\twonorm{\vct{z}_\tau-\vct{x}}+\mu\opnorm{\mtx{A}}\twonorm{\vct{w}}.
\end{align}
This simple example suggests that the spectral norms $\opnorm{\mtx{I}-\mu\mtx{A}^*\mtx{A}}$ and $\opnorm{\mtx{A}}$ play an important role in the convergence of the problem. However, in the under-determined case these rates are rather naive. In fact, $\opnorm{\mtx{I}-\mu\mtx{A}^*\mtx{A}}$ may no longer be smaller than one, so that one may expect the iterates to diverge. However, in the underdetermined and structured case the error $\vct{z}_\tau-\vct{x}$ is not an arbitrary vector. Therefore, even though the spectral norms of these matrices are large the gain of these iterates when acting on the error term $\vct{z}_\tau-\vct{x}$ may not be too large. This suggests  that some form of singular values (perhaps restricted to a particular set) may still play a role even in the underdetermined case. To arrive at the right form of these sets and singular values we need a few definitions.

\begin{definition}[Descent set and cone] \label{decsetcone} The \emph{set of descent} of the function $f$ at a point $\vct{x}$ is defined as
\begin{align*}
{\cal D}_f(\vct{x})=\Big\{\vct{h}:\text{ }f(\vct{x}+\vct{h})\le f(\vct{x})\Big\}.
\end{align*}
The \emph{cone of descent} is defined as a closed cone $\mathcal{C}_f(\vct{x})$ that contains the descent set, i.e.~$\mathcal{D}_f(\vct{x})\subset\mathcal{C}_f(\vct{x})$. The \emph{tangent cone} is the conic hull of the descent set. That is, the smallest closed cone $\mathcal{C}_f(\vct{x})$ obeying $\mathcal{D}_f(\vct{x})\subset\mathcal{C}_f(\vct{x})$.
\end{definition}

The set $\Dc_f(\vct{x})$ is the set of feasible (non-ascent) perturbations of $f(\cdot)$ at $\vct{x}$ and the descent cone $\mathcal{C}$ contains all such directions. The reason the descent cone plays an important role in our results is due to the fact that the errors in our iterates belong to the tangent cone. This suggests that we need to look at the singular values of the matrices $\mtx{I}-\mu\mtx{A}^*\mtx{A}$ and $\mtx{A}$ restricted to this tangent cone. Indeed, the next theorem establishes such a result.

\begin{theorem}\label{master} Let $\vct{x}$ be an arbitrary vector in $\R^n$, $f:\R^n\rightarrow\R$ a proper function\footnote{A proper function is a function whose sub-level sets are closed.}, and $\mathcal{C}_f(\vct{x})$ be a cone of descent of $f$ at $\vct{x}$ and set $\mathcal{C}=\mathcal{C}_f(\vct{x})$.  Suppose $\mtx{A}\in\R^{m\times n}$ is a Gaussian map and let $\vct{y}=\mtx{A}\vct{x}+\vct{w}\in\R^m$ be $m$ linear noisy samples. To estimate $\x$, starting from a point $\des_0$, we apply the PGD update
\begin{align}
\des_{\tau+1}=\mathcal{P}_{\mathcal{K}}(\des_\tau+\mu\A^*(\y-\A\des_\tau)),\nn 
\end{align}
with $\mathcal{K}=\{\vct{z}\in\R^n:\text{ }f(\vct{z})\le f(\vct{x})\}$. Let $\kappa_f$ be a constant that is equal to $1$ for convex $f$ and equal to $2$ for non-convex $f$. Starting from the initial point $\vct{z}_0=\vct{0}$ the update \eqref{updateoflin} obeys
\begin{align}
\label{ratemin}
\twonorm{\vct{z}_\tau-\vct{x}}\le\left(\kappa_f\cdot\rho(\mu)\right)^\tau\twonorm{\vct{x}}+\kappa_f\frac{1-\left(\kappa_f\cdot\rho(\mu)\right)^\tau}{1-\kappa_f\cdot\rho(\mu)}\xi_\mu(\mtx{A})\twonorm{\vct{w}}.
\end{align}
Here $\rho(\mu)$ is the convergence rate defined as
\begin{align*}
\rho(\mu):=\rho(\mu,\mtx{A},f,\vct{x})=&\underset{\vct{u},\vct{v}\in\mathcal{C}_f(\vct{x})\cap \mathcal{B}^{n}}{\sup}\vct{u}^*\left(\mtx{I}-\mu\mtx{A}^*\mtx{A}\right)\vct{v},
\end{align*}
and $\xi_\mu(\mtx{A})$ is the noise amplification factor defined as
\begin{align*}
\xi_\mu(\mtx{A}):=\xi_\mu(\mtx{A},f,\vct{x},\vct{w})=\mu\cdot\underset{\vct{v}\in\mathcal{C}_f(\vct{x})\cap\mathcal{B}^{n}}{\sup}\vct{v}^*\mtx{A}^*\frac{\vct{w}}{\twonorm{\vct{w}}}.
\end{align*}
\end{theorem}
This theorem establishes a counter part to the simple results we mentioned for the overdetermined case \eqref{simplecase}. The only change is that the spectral norms of $\mtx{I}-\mu\mtx{A}^*\mtx{A}$ and $\mtx{A}$ are replaced with restricted versions. Indeed, for properly chosen values of the learning parameter $\mu$ the convergence rate $\rho(\mu)$ can be significantly smaller than $\opnorm{\mtx{I}-\mu\mtx{A}^*\mtx{A}}$. In particular when $\rho(\mu)$ is smaller than one (smaller than $1/2$ in the non-convex case) the above theorem establishes a geometric rate of convergence even in the underdetermined case regardless of whether $f$ is convex or non-convex.

We would like to mention that our results is much more broadly applicable than quadratic cost functions $g(\vct{z})=\frac{1}{2}\twonorm{\vct{y}-\mtx{A}\vct{z}}^2$. Indeed, focusing on the noiseless case $\vct{w}=0$ and any twice continuously differentiable cost function $g(\vct{z})$. Our results holds for this case as well with the convergence rate
\begin{align*}
\rho(\mu):=\rho(\mu,g,f,\vct{x})=&\underset{\vct{u},\vct{v}\in\mathcal{C}_f(\vct{x})\cap \mathcal{B}^{n},\vct{z}\in\R^n}{\sup}\vct{u}^*\left(\mtx{I}-\mu\nabla^2 g(\vct{z})\right)\vct{v}.
\end{align*}
We defer study of this more general case to a future publication.

So far we have described the rate of convergence in terms of the cone-restricted eigenvalues. At this point however it is still not clear how the rate of convergence $\rho(\mu)$ and the noise amplification factor $\xi_\mu(\mtx{A})$ precisely depends on the quantities of interest such as the number of measurements, the signal structure, the function $f$ and the learning parameter. 

In this paper we shall focus on a variety of random ensembles for the measurement matrix $\mtx{A}$. For these random models what affects the convergence rate $\rho(\mu)$ and noise amplification factor $\xi_\mu(\mtx{A})$ is the size of the tangent cone $\mathcal{C}_f(\vct{x})$. The reason this descent cone plays an important role in our results is due to the fact that the difference between the PGD iterates of \eqref{updateoflin} and the signal $\vct{x}$ lie in the tangent cone. As a result it is natural for the success of projected gradient descent to depend on the ``size" of this set. The larger the descent cone the faster the rate of convergence. There are of course different ways to characterize the ``size" of a set. We shall use the notion of Gaussian width defined below to characterize the size of the set.

\begin{definition}[Gaussian width] The Gaussian width of a set $\mathcal{C}\in\R^n$ is defined as:
\begin{align*}
\omega(\mathcal{C}):=\mathbb{E}_{\vct{g}}[\underset{\vct{z}\in\mathcal{C}}{\sup}~\langle \vct{g},\vct{z}\rangle],
\end{align*}
where the expectation is taken over $\vct{g}\sim\mathcal{N}(\vct{0},\mtx{I}_n)$.
\end{definition}
We will use the Gaussian width of the intersection of the descent cone and the unit sphere ($\omega(\mathcal{C}_f(\vct{x})\cap\mathcal{B}^{n})$) to quantify the ``size" of the tangent cone. Interestingly, this parameter also characterizes the ``complexity" of our unknown signal $\vct{x}$ and is often directly related to the number of parameters or degrees of freedom of the structured signal. For example, when using $f(\vct{z})=\onenorm{\vct{z}}$ for a signal with $s$ non-zero coefficients $\omega\approx 2s\log(n/s)$. 

It is known that in the noiseless version of the problem \eqref{opt1} leads to exact recovery when the number of measurements is sufficiently large i.e. $m\ge m_0$. This minimal number of measurements $m_0$ is roughly equal to square of the mean width of the tangent cone of $f$ at $\vct{x}$, i.e.~$m_0\approx \omega^2\left(\mathcal{C}_f(\vct{x})\cap\Bc^n\right)$ \cite{Cha, McCoy}.  This simple formula provides a precise relationship between number of measurements (data complexity) and the structural complexity of the unknown signal. In the next section we will see that for random measurement matrices both the rate of convergence and the noise amplification factor can also be bounded rather precisely in terms of the number of measurements $m$ and this notion of structural complexity ($m_0$); providing precise tradeoffs between the rate of convergence, the number of measurements and the structural complexity of the signal. 

In Section \ref{pgd theory} we develop sharp bounds for convergence rates of projected gradients. Then in Section \ref{numerics} through extensive simulations we confirm that our theoretical bounds are rather sharp both in terms of sample and computational complexity. In Section \ref{sec:prior} we review the vast amount of related literature on this topic. We discuss our results and some future directions in Section \ref{discuss}. We end the paper by proving our results in Section \ref{secproofs}.

\section{Theoretical results for projected gradient descent}\label{pgd theory}
In this section we shall explain our main theoretical results. More specifically we wish to precisely characterize the rate of convergence $\rho(\mu)$ and the noise amplification factor $\xi_\mu$ of Theorem \ref{master} for different random ensembles and different values of $\mu$. Our focus will be on projected gradient descent (PGD). 

As we mentioned earlier we are interested in solving structured linear inverse problem via optimization of the form 
\begin{align}
\label{mainalgoptbig}
\frac{1}{2}\twonorm{\vct{y}-\mtx{A}\vct{z}}^2\quad\text{subject to}\quad f(\vct{z})\le R,
\end{align}
with $R=f(\vct{x})$. As we shall see later in Section \ref{sec:stabilitytoparam} this assumption can be relaxed. 

Naturally the convergence/lack of convergence as well as the rate of convergence of projected gradient descent depends on the learning parameter $\mu$. A large value of the learning parameter will lead to faster convergence to the optimal solution. However, if the learning parameter is too large projected gradient may not converge to the optimal solution or may converge only if the number of measurements are large in comparison with the minimal number of measurements as dictated by the structural complexity of the unknown signal. On the other hand projected gradient descent may still converge with a smaller choice of learning parameter, albeit at a slower rate, even when the number of measurements is not too high. In short, for a fixed value of structural complexity we expect a tradeoff between computational and data complexity. Larger choices of the learning parameter increases the speed of convergence of the algorithm but require a higher data complexity (number of measurements) for convergence. Whereas smaller choices of the learning parameter are more efficient in terms of the data complexity but have slower convergence rates. To provide a rigorous understanding of such tradeoffs we study three different regimes for the learning parameters.
\begin{itemize}
\item A greedy choice of learning parameter of size of the order $\mu\approx 1/m$. This leads to a linear rate of convergence with a sufficiently large number of measurements.
\item A conservative choice of learning parameter of size of the order $\mu\approx 1/n$. The convergence rate in this case is geometric but achieves the best sample complexity (in fact this sample complexity is sharp for convex functions).
\item A structure dependent choice of learning parameter which depends on the structural complexity of the unknown signal. This choice offers a tradeoff between the previous two choices. It achieves a linear rate but at a reduced sample complexity compared to the greedy choice.
\end{itemize}

\subsection{Gaussian maps}
\label{Gaussmaps}
The main focus of this paper is on Gaussian maps. In particular we assume $\mtx{A}\in\R^{m\times n}$ 
has independent $\mathcal{N}(0,1)$ entries. We shall explain how our results can be generalized to non-Gaussian maps in Section \ref{sec:nonGauss}. In all of our results we shall make use of the following definition for the phase transition function which characterizes the minimum required number of measurements.
\begin{definition}[phase transition function]\label{PTcurve}
Let $\vct{x}$ be an arbitrary vector in $\R^n$, $f:\R^n\rightarrow\R$ be a proper function, and $\mathcal{C}_f(\vct{x})$ be a cone of descent of $f$ at $\vct{x}$ and set $\omega=\omega(\mathcal{C}_f(\vct{x})\cap\mathcal{B}^{n})$. Also let $\phi(t)=\sqrt{2}\frac{\Gamma(\frac{t+1}{2})}{\Gamma(\frac{t}{2})}\approx \sqrt{t}$. We define the phase transition function as
\begin{align*}
\mathcal{M}(f,\vct{x},\eta)=\phi^{-1}(\omega+\eta)\approx (\omega+\eta)^2.
\end{align*}
We shall often use the short hand $m_0=\mathcal{M}(f,\vct{x},\eta)$ with the dependence on $f,\vct{x},\eta$ implied. We note that for convex $f$ based on \cite{Cha, McCoy} $m_0$ is exactly the minimum number of measurements required for the program \eqref{mainalgoptbig} to succeed in recovering the unknown signal $\vct{x}$ with high probability (in the absence of noise). With some overloading, even for non-convex $f$, we shall refer to $m_0$ as ``phase transition" or ``minimum number of measurements".   
\end{definition}

\subsubsection{Linear rates of convergence via a greedy learning parameter}
In this section, we shall focus on a greedy learning parameter which is roughly of size $\mu\approx 1/m$. This greedy step size allows us to ensure a linear rate of convergence to the optimal solution with near minimal number of samples. In the theorem below and throughout we shall use the short-hand $b_m=\phi(m)=\sqrt{2}\frac{\Gamma(\frac{m+1}{2})}{\Gamma(\frac{m}{2})}\approx \sqrt{m}$.

\begin{theorem}\label{first prop}  Let $\vct{x}\in\R^n$ and $\vct{w}\in\R^m$ be arbitrary vectors and $f:\R^n\rightarrow\R$ be a proper function. Suppose $\mtx{A}\in\R^{m\times n}$ is a Gaussian map and let $\vct{y}=\mtx{A}\vct{x}+\vct{w}\in\R^m$ be $m$ linear noisy samples. To estimate $\x$, starting from a point $\des_0$, we apply the PGD update
\begin{align}
\label{myrealupdate}
\des_{\tau+1}=\mathcal{P}_{\mathcal{K}}(\des_\tau+\mu\A^*(\y-\A\des_\tau)),
\end{align}
with $\mathcal{K}=\{\vct{z}\in\R^n:\text{ }f(\vct{z})\le f(\vct{x})\}$. Let $\kappa_f$ be a constant that is equal to $1$ for convex $f$ and equal to $2$ for non-convex $f$. Set the learning parameter to $\mu=\frac{1}{b_m^2}\approx\frac{1}{m}$. Furthermore, let $m_0=\mathcal{M}(f,\vct{x},\eta)$ defined by \ref{PTcurve}  be the minimum number of measurements required by the phase transition curve. Then as long as
\begin{align}
\label{nummeaslin}
m>8\kappa_f^2 m_0,
\end{align}
starting from the initial point $\vct{z}_0=\vct{0}$ the update \eqref{myrealupdate} obeys
\begin{align}
\label{ratemin}
\twonorm{\vct{z}_\tau-\vct{x}}\le\left(\frac{8\kappa_f^2m_0}{m}\right)^{\frac{\tau}{2}}\twonorm{\vct{x}}+\sqrt{\frac{\pi}{2}}\frac{\kappa_f}{1-\sqrt{8\kappa_f^2\frac{m_0}{m}}}\frac{\sqrt{m_0}}{m}\twonorm{\vct{w}},
\end{align}
with probability at least $1-9e^{-\frac{\eta^2}{8}}$.
\end{theorem}
We would like to point out that for convex functions $f$ in fact a sharper convergence result is possible. This follows from a more careful analysis which we detail in Section \ref{pfstruct} (more specifically plugging $\mu=1/b_m^2$ in \eqref{mutempbndt}). That is, under the assumptions and setup of Theorem \ref{first prop} for convex functions $f$ we can show that as long as
\begin{align*}
m> \frac{7+3\sqrt{5}}{2} m_0\approx 6.85 m_0,
\end{align*}
then \eqref{ratemin} can be replaced with
\begin{align*}
\twonorm{\vct{z}_\tau-\vct{x}}\le\left(\frac{7+3\sqrt{5}}{2}\frac{m_0}{m}\right)^{\frac{\tau}{2}}\twonorm{\vct{x}}+\frac{\sqrt{\pi}}{8}(3+\sqrt{5})\frac{1}{1-\sqrt{\frac{7+3\sqrt{5}}{2}\frac{m_0}{m}}}\frac{\sqrt{m_0}}{m}\twonorm{\vct{w}}.
\end{align*}
To parse Theorem \ref{first prop} first let us consider the noiseless case where $\vct{w}=0$. The first interesting and perhaps surprising aspect of this result is its generality: It applies not only to convex cost functions but also to non-convex ones! As we mentioned earlier the optimization problem in \eqref{opt1} is not known to be tractable for non-convex $f$.  Regardless, the theorem above shows that with a near minimal number of measurements, projected gradient descent provably converges to the global optimum of \eqref{opt1} without getting trapped in any local optima. 


We note that when condition \eqref{nummeaslin} is satisfied the convergence rate of projected gradient descent ($\sqrt{8\kappa_f^2m_0/m}$) is guaranteed to be strictly less than a constant smaller than one and therefore the algorithm will converge with a linear rate. In particular the number of iterations required by the algorithm to get within a relative error of $\epsilon$ ($\twonorm{\vct{z}_\tau-\vct{x}}/\twonorm{\vct{z}}\le\epsilon$) is roughly of size $2\log(\frac{1}{\epsilon})/(\log\frac{m}{m_0})$. This should be contrasted with the best known rates $1/\epsilon$ and $1/\sqrt{\epsilon}$ obtained by traditional results on non-smooth convex optimization.

As we mentioned earlier \cite{Cha, McCoy} characterized the precise number of measurements required for the optimization problem \eqref{opt1} to succeed when using a convex cost function $f$. This minimal number of measurements was a function of the structural complexity, described by the phase transition function $m_0=\phi(\omega)\approx \omega^2$. The above result shows a linear convergence rate when the number of measurements are larger than $m\ge 6.85m_0$ for convex functions and $m\ge 32m_0$ for non-convex cost functions. That is we get linear convergence to the global optimum when the number of measurements are a small constant factor times the phase transition. 

Another intriguing aspect of the above result is that it suggests an interesting tradeoff between sample complexity (number of measurements) and computational complexity. To see this note that the larger the number of measurements as compared with the phase transition function the faster the rate of convergence  in \eqref{ratemin}. More precisely if $m\ge 8c\kappa_f^2 m_0$ the converge rate is of the form
\begin{align*}
\twonorm{\vct{z}_\tau-\vct{x}}\le\left(\frac{8\kappa_f^2m_0}{m}\right)^{\frac{\tau}{2}}\twonorm{\vct{x}}=\left(\frac{1}{c}\right)^{\frac{\tau}{2}}\twonorm{\vct{x}}.
\end{align*}
Therefore, for fixed structural complexity more samples leads to a faster convergence rate which to some extent leads to a smaller computational complexity and vice versa. The further away we are from the phase transition the faster the convergence rate.
A natural question at this point is whether the constant in front of the phase transition function in \eqref{nummeaslin} is sharp. Naturally this constant depends on the learning parameter. We shall show in later sections that other choices of the learning parameter allow for smaller constants. This decrease in terms of data complexity however leads to slightly slower rates of convergence which in turn leads to higher computational complexity. Regardless, we would like to note that our numerical simulations in Section \ref{sec:numerical} indicate that for the choice of learning parameter $\mu\approx\frac{1}{m}$ our results are a small constant away from the actual data complexity required by the algorithm. For example for promoting sparse structures via the $\ell_1$ norm the minimal number of measurement required for the PGD update to converge is roughly around $5.8m_0$ when the sparsity level $s$ is equal to the signal dimension $n$. This is rather close to the prediction $m>6.85m_0$ of \eqref{nummeaslin}. Interestingly the constant in front of $m_0$ is strictly larger than one. In fact, for fully dense signals ($s=n$) we can show $\rho(1/m)\approx \sqrt{(\sqrt{2}-1)^{-2}n/m}\approx \sqrt{5.8n/m}$ as $n\rightarrow\infty$ which demonstrates that the number of measurements must obey $m> 5.8n\approx 5.8m_0$ to ensure $\rho(1/m)<1$. This suggest that while the choice $\mu\approx\frac{1}{m}$ leads to a linear rate of convergence when $m$ is a constant factor away from the phase transition, this choice may be too greedy for convergence slightly above the phase transition. In order to get convergence closer to the phase transition we shall study other choices for the learning parameter in the coming sections.

Finally, in the presence of noise the algorithm converges to a solution which is within a small radius of the structured signal. This radius of convergence can be shown to be min-max optimal. That is, no algorithm can lead to a solution that is significantly closer to the structured solution. We also remark that the residual term in \eqref{ratemin} has the exact same form (up to small and explicit constants) as the residual term one would get from analyzing \eqref{mainalgoptbig} \cite{oymak2013simple,cai2014geometrizing}.

\subsubsection{Geometric convergence above the phase transition via a conservative learning parameter}
In this section, we shall focus on a conservative learning parameter which is roughly of size $\mu\approx 1/n$. This conservative step size allows us to ensure convergence with minimal amount of samples for convex functions. 
\begin{theorem}\label{ALTthm} Assume the same setup as Theorem \ref{first prop}. Furthermore, assume the penalty function $f$ is convex and set the learning parameter to
\begin{align*}
\mu=\frac{0.99}{\left(\sqrt{m}+\sqrt{n}\right)^2}.
\end{align*}
Then as long as
\begin{align*}
m>m_0,
\end{align*}
the update \eqref{updateoflin} obeys
\begin{align*}
\twonorm{\vct{z}_\tau-\vct{x}}\le\left(1-\frac{0.3}{m+n}\left(\sqrt{m}-\sqrt{m_0}\right)^2\right)^\tau\cdot\twonorm{\vct{x}}+\frac{3.5}{(1-\sqrt{\frac{m_0}{m}})^2}\frac{\sqrt{m_0}}{m}\twonorm{\vct{w}},
\end{align*}
with probability at least $1-2e^{-\frac{\eta^2}{2}}-e^{-\gamma n}$ with $\gamma$ a fixed numerical constant.
\end{theorem}
This theorem shows that exactly above the phase transition $m\ge m_0$, projected gradient descent converges to within a small radius of the structured solution at a geometric rate. While the rate of convergence is geometric it is not linear. Indeed, the rate of convergence slightly above the phase transition is roughly of size $1-\mathcal{O}(\frac{m}{n})$ so that for a relative error of $\epsilon$, the required number of iterations is of the order of $2\frac{n}{m}\log\left(\frac{1}{\epsilon}\right)/(\log \frac{m}{m_0})$. Comparing this result with that of the previous section we note that the computational complexity has increased by a factor of size roughly $n/m$.

\subsubsection{Structure dependent choice of learning parameter}
In the previous two sections we saw convergence results for two learning parameters. These learning parameters while dimension dependent did not depend in anyway on the structure of the unknown solution. A natural question is whether a choice of learning parameter that depends on the structure of the signal can lead to improved convergence results. This is the subject of the next theorem. 

\begin{theorem}\label{PGthm} Assume the same setup as Theorem \ref{first prop}. Furthermore, assume the penalty function $f$ is convex and set the learning parameter to
\begin{align}
\label{structlearn}
\mu=\frac{m}{b_m^2}\frac{2-\sqrt{2m_0}\sqrt[4]{m}\cdot\left(\sqrt{m}-\sqrt{m_0}\right)^{-\frac{3}{2}}}{\left(m_0-2\sqrt{mm_0}+2m\right)}.
\end{align}
Then as long as
\begin{align}
\label{nummeasPGthm}
m>4m_0,
\end{align}
starting from $\vct{z}_0=\vct{0}$ the update \eqref{updateoflin} obeys
\begin{align}
\label{ratePGthm}
\twonorm{\vct{z}_\tau-\vct{x}}\le\left(1-\psi\left(\frac{m_0}{m}\right)\right)^\tau\cdot\twonorm{\vct{x}}+\frac{4}{5}\frac{\sqrt{2\pi}}{\psi\left(\frac{m_0}{m}\right)}\frac{\sqrt{m_0}}{m}\twonorm{\vct{w}},
\end{align}
with probability at least $1-3e^{-\frac{\eta^2}{8}}-e^{-\frac{\eta^2}{2}}$. Here,
\begin{align*}
\psi(\gamma)=2\frac{\left(\sqrt{2}(1-\sqrt{\gamma})^{1.5}-\sqrt{\gamma}\right)^2}{\left(\gamma-2\sqrt{\gamma}+2\right)^2}\approx 1-\sqrt{4\gamma}.
\end{align*}
\end{theorem}
As with our results in Theorem \ref{first prop} the above theorem establishes a linear rate of convergence with near minimal number of measurements. We would also like to point out that similar to our previous results the rate of convergence in \eqref{ratePGthm} exhibits the right behavior in the sense that the larger the number of measurements (the smaller $\gamma$), the faster the rate of convergence. However, the constant in front of the phase transition function in this theorem is sharper by a factor of $2$. Our numerical results in Section \ref{sec:numerical} show that this constant is sharp. That is, with the learning parameter as chosen in \eqref{structlearn} the minimal number of measurements required for the PGD update to converge to the actual solution is very close to $4m_0$ which is precisely the value predicted by \eqref{nummeasPGthm}.

\subsection{Lower bounds on the convergence rate}
We already explained that based on our numerical simulations in Section \ref{sec:numerical} it seems that using a learning parameter of size $\mu\approx 1/m$, projected gradient descent will not converge slightly above the phase transition. Therefore, a more conservative choice of the learning parameter (as in this section) is required to guarantee converge in this ``data poor" regime. A natural question at this point is whether it is possible to have a linear convergence rate using this conservative choice of learning parameter, specially in the ``data poor" regime (slightly above the phase transition). While Theorem \ref{ALTthm} above shows that a geometric rate is possible, the theorem below refutes the possibility of a faster linear rate when the learning parameter is on the of order $1/n$.
\begin{theorem} \label{thm:converse rate} Consider the setup of Theorem \ref{ALTthm} with $\w=0$ and assume that $f(\cdot)$ is convex. For any $\epsilon>0$, there exists a radius $r(\epsilon)$ such that, with probability $1-\exp(-\frac{\eta^2}{2})$, starting from any point $\z_0$ satisfying $\tn{\z_0-\x}\leq r(\epsilon)$, the PGD updates \eqref{updateoflin} obey
\begin{align}
\label{minrate}
\tn{\z_\tau-\x}\geq(1-\epsilon)^\tau\cdot\max\left(1-{\mu(\sqrt{m}+\sqrt{m_0})^2},0\right)^\tau \tn{\z_0-\x}.
\end{align}
\end{theorem}
When $m\geq m_0$, setting the learning parameter to $\mu=\mathcal{O}(\frac{1}{n})$ in \eqref{minrate} guarantees a lower bound on the convergence rate of size $1-\mathcal{O}(\frac{m}{n})$. This also suggests that, if one wishes to obtain an $o(1)$ convergence rate, than the learning parameter should indeed be $\gtrsim\frac{1}{m}$, which is consistent with the greedy choice of Theorem \ref{first prop}.

\subsection{Stability to tuning parameters}
\label{sec:stabilitytoparam}
So far we have assumed that the parameter $R$ is tuned perfectly and is set to $R=f(\vct{x})$. Of course in practice $f(\vct{x})$ is not known in advance. In this section we shall explain how our results on the performance of projected gradient descent change if $R\neq f(\vct{x})$ in \eqref{mainalgopt}. For the sake of exposition we shall focus on the noiseless case and assume $\vct{w}=\vct{0}$. As it will become clear in our proofs our results easily generalize to the noisy case. Also we shall only state our results in the setup of Theorem \ref{first prop}. All of our results easily generalize to the setup of other theorems.

\begin{theorem}\label{sensitivitythm} Let $\vct{x}\in\R^n$ be an arbitrary vector in $\R^n$ and $f:\R^n\rightarrow\R$ be a proper function. Suppose $\mtx{A}\in\R^{m\times n}$ is a Gaussian map and let $\vct{y}=\mtx{A}\vct{x}\in\R^m$ be $m$ linear samples. To estimate $\x$, starting from a point $\des_0$, we apply the PGD update \eqref{updateoflin} 
with $\mathcal{K}=\{\vct{z}\in\R^n:\text{ }f(\vct{z})\le R\}$. Let $\kappa_f$ be a constant that is equal to $1$ for convex $f$ and equal to $2$ for non-convex $f$. Set the learning parameter to $\mu=\frac{1}{b_m^2}\approx\frac{1}{m}$. Furthermore, let $m_0=\mathcal{M}(f,\vct{x},\eta)$ defined by \ref{PTcurve}  be the minimum number of measurements required by the phase transition curve. Then as long as
\begin{align}
\label{nummeaslin2}
m>8\kappa_f^2m_0,
\end{align}
starting from the initial point $\vct{z}_0=\vct{0}$ the update \eqref{updateoflin} obeys the following condition for all $R<f(\vct{x})$,
\begin{align}
\label{ratesensl}
\twonorm{\vct{z}_\tau-\vct{x}}\le\rho^{\tau}\twonorm{\vct{x}}+\frac{3-\kappa_f}{1-\rho}\twonorm{\vct{x}-\mathcal{P}_{\mathcal{K}}(\vct{x})},\quad\text{with}\quad\rho=\sqrt{8}\kappa_f\sqrt{\frac{m_0}{m}},
\end{align}
with probability at least $1-8e^{-\frac{\eta^2}{8}}$. Furthermore, if $f$ is absolutely homogenous starting from the initial point $\vct{z}_0=\vct{0}$ for all $R>f(\vct{x})$ the update \eqref{updateoflin} obeys
\begin{align}
\label{ratesensu}
\twonorm{\vct{z}_\tau-\vct{x}}\le\rho^\tau\twonorm{\vct{x}}+\frac{\kappa_f+1+2\rho}{1-\rho}\left(\frac{R}{f(\vct{x})}-1\right)\twonorm{\vct{x}},\quad\text{with}\quad\rho=\sqrt{8}\kappa_f\sqrt{\frac{m_0'}{m}},
\end{align}
with probability at least $1-8e^{-\frac{\eta^2}{8}}$. Here $m_0'=\mathcal{M}(f,\frac{R}{f(\vct{x})}\vct{x},\eta)$.\footnote{Note that if $f$ is a norm, $m_0'=m_0$.}
\end{theorem}
Comparing Theorem \ref{sensitivitythm} with its counterpart in Theorem \ref{first prop} we see that there is an extra term due to the mismatch between the tuning parameter $R$ and $f(\vct{x})$. This extra ``mismatch error" goes to zero as $R\rightarrow f(\vct{x})$. This is of course natural as if we do not know $R$ precisely we can not hope to recover the signal $\vct{x}$ exactly. In fact we would like to note that this mismatch term takes a near optimal form. In particular when $R<f(\vct{x})$ one can not hope to recover a signal that is closer to $\vct{x}$ than $\mathcal{P}_{\mathcal{K}}(\vct{x})$. Therefore, the smallest error one can hope for is $\twonorm{\vct{x}-\mathcal{P}_{\mathcal{K}}(\vct{x})}$. Interestingly, \eqref{ratesensl} shows that one can get to this near minimal error via projected gradient descent. Another interesting aspect of the theorem above is that the mismatch error terms exhibit the correct behavior in terms of the data complexity. Larger number of measurements lead to smaller values of the rate $\rho$ which in turn lead to a smaller mismatch error.
\subsection{Non-Gaussian maps}
\label{sec:nonGauss}
Our results are not restricted to Gaussian ensembles and apply to a wide variety of random measurements. We shall state our results for non-Gaussian maps in the setup of Theorem \ref{first prop} and in the noiseless case ($\vct{w}=\vct{0}$). All of our results easily generalize to the setup of other theorems as well as noisy measurements. We will discuss two general class of random ensembles. One general class of ensembles are isotropic sub-Gaussian matrices.
\begin{definition}[Isotropic Sub-Gaussian (ISG) matrices] \label{subgauss def}$\A\in\R^{m\times n}$ is a $\Delta$-subgaussian matrix if its rows are independent of each other and for all $1\leq i\leq m$ the $i$th row $\ab_i^*$ satisfies 
\begin{itemize}
\item $\E[\ab_i]=0$.
\item $\E[\ab_i\ab_i^*]=\Iden_n$.\footnote{The assumption that the population covariance matrix is identity is not necessary. In fact our result generalizes to any covariance matrix. The only change is that the required number of measurements will now scale with the condition number of the covariance matrix.}
\item For any vector $\vb$, $\Pro(|\ab_i^*\vb|\geq t\twonorm{\vct{v}})\leq \exp(-\frac{t^2}{\Delta^2})$.
\end{itemize}
\begin{theorem}\label{nonGaussThm}Let $\vct{x}$ be an arbitrary vector in $\R^n$, $f:\R^n\rightarrow\R$ a proper function. Suppose $\mtx{A}\in\R^{m\times n}$ is an Isotropic sub-Gaussian map with parameter $\delta$ and let $\vct{y}=\mtx{A}\vct{x}\in\R^m$ be $m$ linear samples. To estimate $\x$, starting from a point $\des_0$, we apply the PGD update \eqref{updateoflin}
with $\mathcal{K}=\{\vct{z}\in\R^n:\text{ }f(\vct{z})\le f(\vct{x})\}$. Set the learning parameter to $\mu=\frac{1}{m}$. Furthermore, let $m_0=\mathcal{M}(f,\vct{x},\eta)$ be defined by \ref{PTcurve}  be the minimum number of measurements required by the phase transition curve. Also assume that
\begin{align}
\label{nummeaslinng}
m>c_\Delta\cdot m_0,
\end{align}
for a fixed numerical constant $c_\Delta$ depending only on $\Delta$. 
Then, there exists a constant $c_1>0$ and an event of probability at least $1-e^{-c_1\eta^2}$, such that on this event, starting from the initial point $\vct{z}_0=\vct{0}$ the update \eqref{updateoflin} obeys
\begin{align}
\label{rateminng}
\twonorm{\vct{z}_\tau-\vct{x}}\le\left(c_\Delta\frac{m_0}{m}\right)^{\frac{\tau}{2}}\twonorm{\vct{x}}.
\end{align}
\end{theorem}

\end{definition}
Results with sub-Gaussian ensembles provide useful theoretical insights. In many practical applications other measurement ensembles may be of interest. Indeed, measurements that are physically realizable in most signal processing applications (e.g. MRI) are based on Fourier ensembles. Also, sub-Gaussian ensembles are dense and do not have fast multiplication. So it is desirable to have measurement ensembles where applying $\mtx{A}$ and its transpose can be implemented with a fast algorithm (e.g. applying the discrete Fourier transform via the FFT algorithm). Our results also apply to a wide variety of ensembles that are physically implementable and have fast multiplication.
\begin{definition}[Subsampled Orthogonal with Random Sign (SORS) matrices]\label{SORSdef} Let $\mtx{F}\in\R^{n\times n}$ denote an orthonormal matrix obeying
\begin{align}
\label{BOS}
\mtx{F}^*\mtx{F}=\mtx{I}\quad\text{and}\quad\max_{i,j}\abs{\mtx{F}_{ij}}\le \frac{\Delta}{\sqrt{n}}.
\end{align}
Define the random subsampled matrix $\mtx{H}\in\R^{m\times n}$ with i.i.d.~rows chosen uniformly at random from the rows of $\mtx{F}$. Now we define the Subsampled Orthogonal with Random Sign (SORS) measurement ensemble as $\mtx{A}=\mtx{H}\mtx{D}$, where $\mtx{D}\in\R^{n\times n}$ is a random diagonal matrix with the diagonal entries i.i.d.~$\pm 1$ with equal probability.
\end{definition}
To simplify exposition, in the definition above we have focused on SORS matrices based on subsampled orthonormal matrices $\mtx{H}$ with i.i.d.~rows chosen uniformly at random from the rows of an orthonormal matrix $\mtx{F}$ obeying \eqref{BOS}. However, our results continue to hold for SORS matrices defined via a much broader class of random matrices $\mtx{H}$ with i.i.d.~rows chosen according to a probability measure on Bounded Orthonormal Systems (BOS). Please see \cite[Section 12.1]{foucart2013random} for further details on such ensembles.

\begin{theorem}\label{SOSthm}Consider the same setup as Theorem \ref{first prop} using a Subsampled Orthogonal with Random Sign (SORS) matrix. Also assume that
\begin{align}
\label{nummeaslinng}
m>c_\Delta(\eta+1)^2\cdot \left(\log n\right)^4\cdot m_0 ,
\end{align}
for a fixed numerical constant $c_\Delta$ depending only on $\Delta$. 
Then, there is an event of probability at least $1-4e^{-\frac{\eta^2}{8}}$, such that on this event, starting from the initial point $\vct{z}_0=\vct{0}$ the update \eqref{updateoflin} obeys
\begin{align}
\label{rateminng}
\twonorm{\vct{z}_\tau-\vct{x}}\le\left(c_\Delta(\eta+1)^2\frac{m_0}{m}\log^4 n\right)^{\frac{\tau}{2}}\twonorm{\vct{x}}.
\end{align}
\end{theorem}
Theorems \ref{nonGaussThm} and \ref{SOSthm} above extend Theorem \ref{first prop} to the class of ISG and SORS matrices albeit at a loss in terms of the constants (log factors). We would like mention that our theoretical results apply to a more general class of matrices. Indeed as it will become clear in the proofs (Sections \ref{app:subgaussian} and \ref{proofSORS}) \eqref{nummeaslinng} and \eqref{rateminng} continue to hold for any matrix obeying the following isometry property.
\begin{align}
\label{isometry}
\max_{\vb\in \mathcal{S}}\abs{\twonorm{\A\vb}^2-\twonorm{\vb}^2}\leq \Delta,
\end{align}
where 
\begin{align}
\mathcal{S}=\mathcal{C}_{-}\bigcup\mathcal{C}_{+},\quad\text{with}\quad\Cc_-=\Cc_f(\vct{x})\cap \mathbb{S}^{n-1}-\Cc_f(\vct{x})\cap \mathbb{S}^{n-1},\quad \Cc_+=\Cc_f(\vct{x})\cap \mathbb{S}^{n-1}+\Cc_f(\vct{x})\cap \mathbb{S}^{n-1},\nn
\end{align}
and the columns of $\mtx{A}$ are properly normalized to have Euclidean norm equal to one.
We emphasize that finding maps satisfying \eqref{isometry} is an active and exciting area of research.

\section{Numerical results}\label{numerics}
In this section we corroborate the theoretical findings of the previous sections via experiments on synthetic data as well as a few experiments on natural images. We begin with some synthetic experiments.
\label{sec:numerical}

\subsection{Synthetic simulations}
In the simulations of this section we focus on sparse recovery problems using various cost functions $f$. We run our simulations with signals of fixed dimension $n=1000$. We however vary the sparsity level $s$, number of measurements $m$, and the learning parameter $\mu$. In our experiments we also use two different random ensembles for the sensing matrix $\mtx{A}\in\R^{m\times n}$: Gaussian with i.i.d.~entries $\mathcal{N}(0,1)$ and Bernoulli with i.i.d.~entries taking values $0$ and $1$ with equal probability.

\subsubsection{Phase transitions}
\label{secnumpt}

\begin{figure}[t]
\centering
\includegraphics[width=0.7\textwidth]{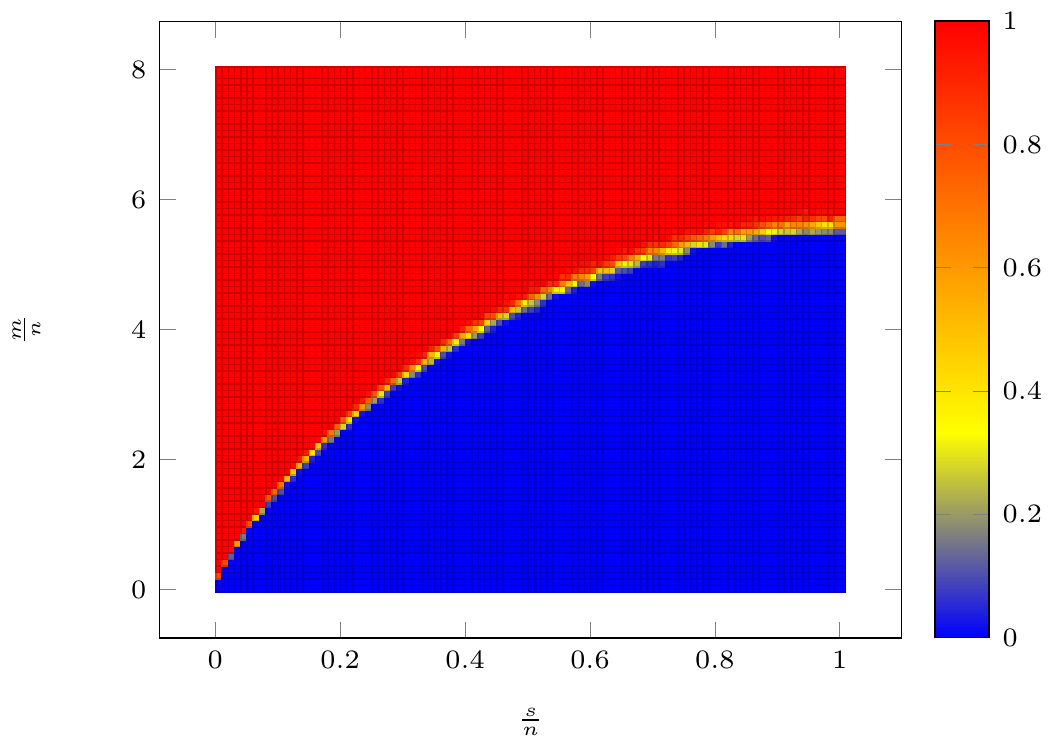}
\caption{
\textbf{Phase transition phenomenon for projected gradient iterations.}
This diagram shows the empirical probability that $\ell_1$-projected gradient with $\mu=1/m$ successfully identifies a vector $\vct{x}\in\R^{n=1000}$ with $s$ nonzero entries from m random measurements of the form $\vct{y} =\mtx{A}\vct{x}$. Here, $\mtx{A }\in\R^{m\times n}$ is a Gaussian matrix with entries i.i.d.~$\mathcal{N}(0,1)$. The colormap tapers between red and blue where red represents certain success, while blue represents certain failure.}
\label{PTmain}
\end{figure}
The aim of the simulations of this section is to explore how the phase transition curve changes with different functions $f$ and different learning parameters $\mu$. We use random signals $\vct{x}\in\R^n$ and random sensing matrices $\mtx{A}\in\R^{m\times n}$. The support of the signal is of size $s$ and is chosen at random with the values on the support distributed i.i.d.~$\mathcal{N}(0,1)$. We use $\vct{y}=\mtx{A}\vct{x}\in\R^m$ as our measurements and apply the PGD update with various learning parameters $\mu$ starting from $\vct{z}_0=\vct{0}$. We stop after $500$ iterations and record the empirical probability of success. The empirical probability of success is an average over $50$ trials, where in each instance, we generate new random signals and measurement matrices. We declare a trial successful if the relative error of the reconstruction $\twonorm{\hat{\vct{x}}-\vct{x}}/\twonorm{\vct{x}}$ falls below $10^{-3}$. 

Figure \ref{PTmain} depicts the empirical success probabilities via a color map for different sparsity levels $s$ and number of measurements $m$. Red represents certain success, while blue represents certain failure. In the experiments of this figure the measurements are Gaussians, the objective function is $f(\vct{z})=\onenorm{\vct{z}}$ and the learning parameter is set to $\mu=1/m$. This curve clearly shows that there is a phase transition curve for the number of measurements as a function of the sparsity. On one side of this curve PGD updates is successful with high probability on the other side it fails with high probability. In future depictions of phase transitions in this section we shall only plot the phase transition curve as found by a boundary detection routine in lieu of the complete colormap.\footnote{In particular MATLAB bwtraceboundary function.}

We now turn our attention to comparing the phase transition curves obtained for the sparse recovery problem with different cost functions and learning parameters. We shall use $\ell_p$-norms with $p=0,1/2,1$ as our cost functions (i.e.~$f(\vct{z})=\zeronorm{\vct{z}},\|\vct{z}\|_{\ell_{1/2}} , and \onenorm{\vct{z}}$).\footnote{We note that the proximal function of $\ell_{1/2}$ has a closed form solution e.g.~see \cite{zongben2012l1}. We have utilized this closed form solution of the proximal function to implement projection on to the one-half ball.} In our experiments we use both binary and Gaussian matrices.

\begin{figure}[t]
        \centering
\begin{tikzpicture}[scale=1] 
\begin{groupplot}[group style={group size=2 by 1,horizontal sep=1cm,xlabels at=edge bottom, ylabels at=edge left,xticklabels at=edge bottom},xlabel=$\frac{s}{n}$,
        ylabel=$\frac{m}{n}$, legend style={at={(0.72,0.38)},anchor=north,legend cell align=left}]
 \nextgroupplot[title={Gaussian matrices}]

        \addplot [red,line width=1pt] table[x index=0,y index=3]{./THMPT};\addlegendentry{p=1 (Thm \ref{first prop})}
        \addplot [blue,line width=1pt] table[x index=0,y index=1]{./PTmatopt_0_p_0_Reg_1};\addlegendentry{$p=0$}
         \addplot [teal,line width=1pt] table[x index=0,y index=1]{./PTmatopt_0_p_0.5_Reg_1};\addlegendentry{$p=\frac{1}{2}$}
          \addplot [magenta,line width=1pt] table[x index=0,y index=1]{./PTmatopt_0_p_1_Reg_1};\addlegendentry{$p=1$}

        \nextgroupplot[title={Bernoulli matrices}]

        \addplot [red,line width=1pt] table[x index=0,y index=3]{./THMPT};\addlegendentry{p=1 (Thm \ref{first prop})}
        \addplot [blue,line width=1pt] table[x index=0,y index=1]{./PTmatopt_1_p_0_Reg_1};\addlegendentry{$p=0$}
         \addplot [teal,line width=1pt] table[x index=0,y index=1]{./PTmatopt_1_p_0.5_Reg_1};\addlegendentry{$p=\frac{1}{2}$}
          \addplot [magenta,line width=1pt] table[x index=0,y index=1]{./PTmatopt_1_p_1_Reg_1};\addlegendentry{$p=1$}

        \end{groupplot}
\end{tikzpicture}
\caption{
\textbf{Phase transition curve with a greedy learning parameter.}\\
This diagram shows the empirical phase transition curve for $\ell_p$-projected gradient with $\mu=1/m$ for successful recovery of a vector $\vct{x}\in\R^{n=1000}$ with $s$ nonzero entries from m random measurements of the form $\vct{y} =\mtx{A}\vct{x}$. Here, $\mtx{A }\in\R^{m\times n}$ is a Gaussian or Bernouli random matrix. We demonstrate results using $\ell_0,\ell_{1/2},$ and $\ell_1$. We also draw the prediction of Theorem \ref{first prop} for $\ell_1$-PGD for comparison. We see that the prediction of Theorem \ref{first prop} is accurate up to a multiplicative factor of $1.18$.}
\label{fig:greedy}
\end{figure}

\begin{figure}
\centering
\begin{subfigure}[b]{0.95\textwidth} 
        \centering
\begin{tikzpicture}[scale=1] 
\begin{groupplot}[group style={group size=2 by 1,horizontal sep=1cm,xlabels at=edge bottom, ylabels at=edge left,xticklabels at=edge bottom},xlabel=$\frac{s}{n}$,
        ylabel=$\frac{m}{n}$, legend style={at={(0.84,0.38)},anchor=north,legend cell align=left}]
 \nextgroupplot[title={Gaussian matrices}]
        
        \addplot [gray!70,line width=2pt] table[x index=0,y index=1]{./DonohoCurve};\addlegendentry{$\ell_1$ PT}

        \addplot [blue,line width=1pt] table[x index=0,y index=1]{./PTmatopt_0_p_0_Reg_4};\addlegendentry{$p=0$}
         \addplot [teal,line width=1pt] table[x index=0,y index=1]{./PTmatopt_0_p_0.5_Reg_4};\addlegendentry{$p=\frac{1}{2}$}
          \addplot [magenta,line width=1pt] table[x index=0,y index=1]{./PTmatopt_0_p_1_Reg_4};\addlegendentry{$p=1$}

        \nextgroupplot[title={Binary matrices}]
        
        \addplot [gray!70,line width=2pt] table[x index=0,y index=1]{./DonohoCurve};\addlegendentry{$\ell_1$ PT}
        
        \addplot [blue,line width=1pt] table[x index=0,y index=1]{./PTmatopt_1_p_0_Reg_4};\addlegendentry{$p=0$}
         \addplot [teal,line width=1pt] table[x index=0,y index=1]{./PTmatopt_1_p_0.5_Reg_4};\addlegendentry{$p=\frac{1}{2}$}
          \addplot [magenta,line width=1pt] table[x index=0,y index=1]{./PTmatopt_1_p_1_Reg_4};\addlegendentry{$p=1$}
          
        \end{groupplot}
\end{tikzpicture}
\caption{\textbf{Phase transition curve with a conservative learning parameter.}
This diagram shows the empirical phase transition curve for $\ell_p$-projected gradient with $\mu=1/\left(\sqrt{m}+\sqrt{n}\right)^2$ for successful recovery of a vector $\vct{x}\in\R^{n=1000}$ with $s$ nonzero entries from m random measurements of the form $\vct{y} =\mtx{A}\vct{x}$. Here, $\mtx{A }\in\R^{m\times n}$ is a Gaussian or Bernouli random matrix. We demonstrate results using $\ell_0,\ell_{1/2},$ and $\ell_1$. We also draw the prediction of Theorem \ref{ALTthm} for $\ell_1$-PGD for comparison. We see that the prediction of Theorem \ref{ALTthm} is a near perfect match with the empirical simulations.} \label{fig:conservative} \end{subfigure} 
\\
\begin{subfigure}[b]{0.95\textwidth} 
       \centering
\begin{tikzpicture}[scale=1] 
\begin{groupplot}[group style={group size=2 by 1,horizontal sep=1cm,xlabels at=edge bottom, ylabels at=edge left,xticklabels at=edge bottom},xlabel=$\frac{s}{n}$,
        ylabel=$\frac{m}{n}$, legend style={font=\tiny,at={(0.24,1)},anchor=north,legend cell align=left}]
 \nextgroupplot[title={Gaussian matrices}]
        
        \addplot [gray!70,line width=2pt] table[x index=0,y index=2]{./THMPT};\addlegendentry{pred. Thm \ref{PGthm}}
        \addplot [blue,line width=1pt] table[x index=0,y index=1]{./PTCVXmatopt_0_p_1_Reg_1};\addlegendentry{$\mu=1/m$}

          \addplot [magenta,line width=1pt] table[x index=0,y index=1]{./PTCVXmatopt_0_p_1_Reg_3};\addlegendentry{$\mu$ as in Theorem \ref{PGthm}}
          \addplot [teal,line width=1pt] table[x index=0,y index=1]{./PTmatopt_0_p_1_Reg_4};\addlegendentry{$\mu=1/\left(\sqrt{m}+\sqrt{n}\right)^2$}
     
        \nextgroupplot[title={Binary matrices}]
        
        \addplot [gray!70,line width=2pt] table[x index=0,y index=2]{./THMPT};\addlegendentry{p=1 (Thm \ref{PGthm})}
        \addplot [blue,line width=1pt] table[x index=0,y index=1]{./PTCVXmatopt_1_p_1_Reg_1};\addlegendentry{$\mu=1/m$}

          \addplot [magenta,line width=1pt] table[x index=0,y index=1]{./PTCVXmatopt_1_p_1_Reg_3};\addlegendentry{$\mu$ as in Theorem \ref{PGthm}}
          \addplot [teal,line width=1pt] table[x index=0,y index=1]{./PTmatopt_1_p_1_Reg_4};\addlegendentry{$\mu=1/\left(\sqrt{m}+\sqrt{n}\right)^2$} 
        \end{groupplot}
\end{tikzpicture}
\caption{\textbf{Comparison of different learning parameters $\mu$ for $\ell_1$-PGD.}\\
This diagram shows the empirical phase transition curve for $\ell_1$-projected gradient with three different learning parameters: greedy choice $\mu=1/m$, a conservative choice $\mu=1/\left(\sqrt{m}+\sqrt{n}\right)^2$ as well as the structure dependent choice of Theorem \ref{PGthm}. This curves are about recovery of a vector $\vct{x}\in\R^{n=1000}$ with $s$ nonzero entries from m random measurements of the form $\vct{y} =\mtx{A}\vct{x}$. Here, $\mtx{A }\in\R^{m\times n}$ is a Gaussian or Bernouli random matrix. We also draw the prediction of Theorem \ref{PGthm} for $\ell_1$-PGD for comparison. We see that the prediction of Theorem \ref{PGthm} is a near perfect match with the empirical simulations.
}
\label{fig:choices}
\end{subfigure} 
\end{figure}

We first consider the greedy learning parameter $\mu= \frac{1}{m}$ and show the empirical phase transition curves in Figure \ref{fig:greedy}. These figures show that the phase transition curves for Bernoulli and Gaussian matrices are essentially identical. This is in line with the universality of phase transitions observed previously, e.g.~see \cite{universal}. These plots also suggest that non-convex cost functions may require fewer measurements for successful recovery. While the curves are close to each other, this experiment shows that $\ell_{1/2}$-PGD requires the smallest sample complexity and $\ell_0$-PGD outperforms $\ell_1$-PGD. This is perhaps surprising as one would expect $p=1/2$ to perform somewhere between $p=0$ and $p=1$.  We also plot our theoretical prediction for the phase transition for $\mu=1/m$ (equation \eqref{nummeaslin} from Theorem \ref{first prop}). We observe that the empirical phase transitions and our theoretical predictions are rather close. In particular, our theoretical prediction for $p=1$ is off by at most a factor of $1.18$. Our predictions for $p=0$ and $p=1$ is not as sharp due to the extra factor $\kappa_f=2$ in \eqref{nummeaslin}.

We now focus on phase transition curve for the more conservative learning parameter $\mu=1/(\sqrt{m}+\sqrt{n})^2$. For PGD with $\ell_1$ projection, Theorem \ref{ALTthm} predicts that PGD phase transition matches that of $\ell_1$-minimization (a.k.a.~the Donoho-Tanner curve \cite{DonPhaseTrans,DonNonNegative}). In fact as mentioned earlier the statement of Theorem \ref{ALTthm} holds for any convex function and not just $\ell_1$. Figure \ref{fig:conservative} confirms this fact empirically and shows that our theoretical prediction is indeed tight. This plot also shows the phase transition curves of $\ell_0$-PGD and $\ell_{1/2}$-PGD is below that of $\ell_1$-PGD. These curves confirm that starting from zero nonconvex PGD updates beat the phase transition of $\ell_1$ minimization (i.e.~the Donoho-Tanner curve)! We note that our theory for the conservative tuning parameter only applies to convex regularizers. Characterizing $\ell_0$-PGD or $\ell_{1/2}$-PGD curves via the conservative choice of tuning parameter $\mu= \frac{1}{(\sqrt{m}+\sqrt{n})^2}$ is left to future research.

Finally, in Figure \ref{fig:choices}, we compare the empirical phase transitions of different choices of tuning parameters for $\ell_1$ PGD. In addition to the greedy ($\mu=1/m$) and conservative ($\mu=1/\left(\sqrt{m}+\sqrt{n}\right)^2$) choices of the learning parameter we also use the structure dependent choice of the learning parameter (equation \eqref{structlearn} of Theorem \ref{PGthm}). The structure dependent choice of the learning parameter aimed to have a linear convergence rate (similar to the greedy choice $\mu=1/m$) but with improved sample complexity. Figure \ref{fig:choices} shows that this is indeed the case, as the phase transition curve of this choice of learning parameter is below that of the greedy choice. In particular our theoretical prediction in Theorem \ref{PGthm}, is a near identical match with the empirical phase transition. The difference between various choices of the learning parameter $\mu$ shows that there is a tradeoff between convergence rate and sample complexity. The greedy choice of learning parameter leads to a linear convergence rate but requires more measurements for successful recovery when compared to the phase transition of $\ell_1$ minimization (note that this phase transition is not drawn in this figure but it essentially an identical match with the curve of $\mu=1/(\sqrt{m}+\sqrt{n})^2$). In comparison the conservative choice of learning parameter does lead to exact recovery precisely above the $\ell_1$ phase transition. The convergence rate is no longer linear (it is however still geometric). Theorem \ref{thm:converse rate} shows that in fact this behavior is sharp and when using PGD updates a linear rate of convergence can not be obtained exactly above the phase transition.

\subsubsection{Rates of convergence}
\label{sec:rateofconvnum}
\begin{figure}[t]
        \centering
        \includegraphics[width=0.5\textwidth]{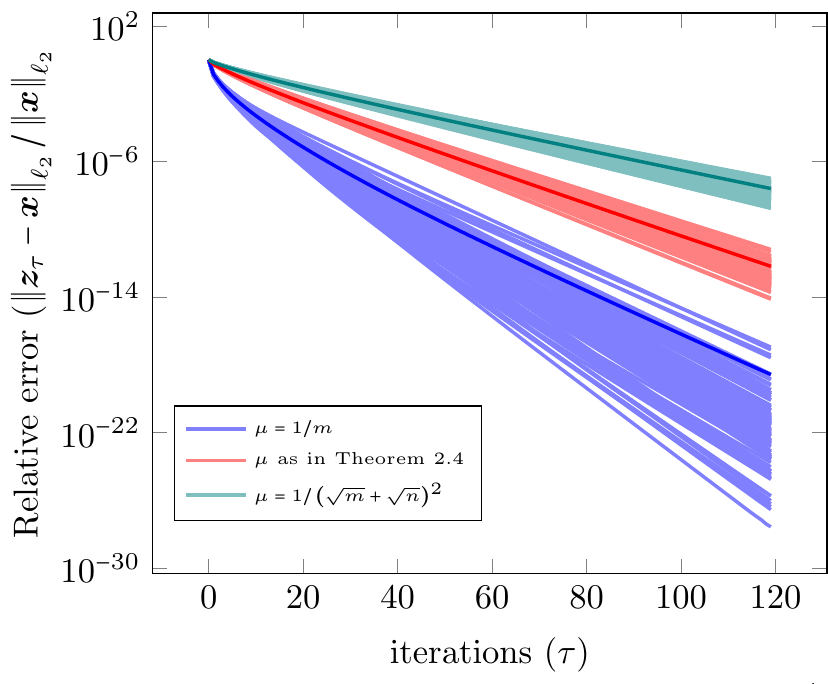}

\caption{\textbf{Rate of convergence of $\ell_1$-PGD for different learning parameters.}\\
This diagram shows the empirical rates of convergence for $\ell_1$-projected gradient with three different learning parameters: greedy choice $\mu=1/m$, a conservative choice $\mu=1/\left(\sqrt{m}+\sqrt{n}\right)^2$ as well as the structure dependent choice of Theorem \ref{PGthm}. These curves are about recovery of a vector $\vct{x}\in\R^{n=2000}$ with $s=40$ nonzero entries from $m=1200$ random measurements of the form $\vct{y} =\mtx{A}\vct{x}$. Here, $\mtx{A }\in\R^{m\times n}$ is a Gaussian random matrix. Curves show relative error of $\ell_1$ projected gradient descent for $50$ iterations and show the convergence rate for $100$ trials as well as the average over these trials in bold. We see that all three step sizes show geometric rate of convergence and that the larger step sizes lead to a faster convergence rate.
}
\label{fig:choices2}
\end{figure}

The aim of the simulations of this section is to explore how the rate of convergence of project gradient depends on the number of measurements $m$ and the structure complexity of the unknown signal $m_0$ as well as the learning parameter $\mu$. Similar to the simulations of the previous section we use random signals $\vct{x}\in\R^n$ and random sensing matrices $\mtx{A}\in\R^{m\times n}$. The support of the signal is of size $s$ and is chosen at random with the values on the support distributed i.i.d.~$\mathcal{N}(0,1)$. We use $\vct{y}=\mtx{A}\vct{x}\in\R^m$ as our measurements and apply the PGD update with various learning parameters $\mu$ starting from $\vct{z}_0=\vct{0}$. In the experiments of this section $n=2000$.

In our experiment we shall focus on sparse recovery using $\ell_1$ minimization. We run $\ell_1$ projected gradient descent for $50$ iterations and record the convergence rate for $100$ trials. We depict these curves in Figure \ref{fig:choices2} as well as the average over these trials in bold. We see that all three step sizes show geometric rate of convergence and that the larger step sizes lead to a faster convergence rate.

 \begin{figure}[t]
        \centering
\begin{tikzpicture}[scale=1] 
\begin{groupplot}[group style={group size=1 by 1,horizontal sep=1cm,xlabels at=edge bottom, ylabels at=edge left,xticklabels at=edge bottom},xlabel=iterations ($\tau$),
        ylabel=Normalized relative error, legend style={font=\tiny,at={(0.84,0.98)},anchor=north,legend cell align=left}]
        
        \nextgroupplot[title={Normalized rates of convergence}]
        
        \addplot [magenta,line width=1pt] table[x index=0,y index=1]{./convRateGreedy};\addlegendentry{$s=0.01n$}
        \addplot [teal,line width=1pt] table[x index=0,y index=2]{./convRateGreedy};\addlegendentry{$s=0.02n$}
        \addplot [blue,line width=1pt] table[x index=0,y index=3]{./convRateGreedy};\addlegendentry{$s=0.05n$}
        \addplot [orange,line width=1pt] table[x index=0,y index=4]{./convRateGreedy};\addlegendentry{$s=0.1n$}
        \addplot [black,line width=1pt] table[x index=0,y index=5]{./convRateGreedy};\addlegendentry{$s=0.2n$}
        \addplot [red,line width=1pt] table[x index=0,y index=6]{./convRateGreedy};\addlegendentry{$s=0.5n$}

        \end{groupplot}
\end{tikzpicture}
\caption{\textbf{Normalized convergence rates for different sparsity levels.}\\
This diagram shows the empirical rates of convergence for $\ell_1$-projected gradient. These curves are about recovery of a vector $\vct{x}\in\R^{n=2000}$ with $s$ nonzero entries from m random measurements of the form $\vct{y} =\mtx{A}\vct{x}$. Here, $\mtx{A }\in\R^{m\times n}$ is a Gaussian random matrix and different curves denote different values of $s$. The horizontal axis corresponds to the number of iterations. The vertical axis represents the median relative error from $100$ trials normalized by the ratio ${\left((1-\gamma)\frac{\tilde{m}_0}{m}\right)}^{\frac{\tau}{2}}$ for different sparsity levels. Here, $\gamma=0.075$ and $\tilde{m}_0$ is the value obtained from the the empirical phase transition for the greedy choice $\mu=1/m$ (the curve is depicted in Figure \ref{PTmain})). These curves confirm that the convergence rate scales like $\sqrt{\frac{\tilde{m}_0}{m}}$ as predicted by Theorem \ref{first prop}.}
\label{fig:rates}
\end{figure}

To investigate the nature of the dependence of the geometric rate of convergence on the number of measurements and the minimum required number of measurements (which in turn depends on sparsity) we consider different sparsity levels $s=0.01n, 0.02n, 0.05n, 0.1n, 0.2n, $ and $0.5n$. The number of measurements used for each sparsity level is set to $m=8m_0$, where $m_0$ is obtained from the phase transition of $\ell_1$ (Donoho-Tanner curve). For each sparsity level we run $\ell_1$-PGD with learning parameter $\mu=1/m$ for $50$ iterations and record the relative error. In Figure \ref{fig:rates} we show the median relative error from $100$ trials normalized by the ratio ${\left((1-\gamma)\frac{\tilde{m}_0}{m}\right)}^{\frac{\tau}{2}}$ for different sparsity levels. Here, $\gamma=0.075$ and $\tilde{m}_0$ is the value obtained from the empirical phase transition for the greedy choice $\mu=1/m$ (the curve is depicted in Figure \ref{PTmain})). We see that except for a burn-in period across different sparsity levels and many iterations the normalized relative error is constant. This shows that the convergence rate scales like $\sqrt{\frac{\tilde{m}_0}{m}}$. This is exactly the type of behavior predicted by Theorem \ref{first prop}. In particular, Theorem \ref{first prop} finds an upper bound on the convergence rate of the form $(\frac{8m_0}{m})^{\frac{\tau}{2}}$ where $8m_0$ is an upper bound on the empirical phase transition $\tilde{m}_0$. Figure \ref{fig:rates} indicates that our bounds can be sharpened in practice by replacing the upper bound $8m_0$ by the empirical phase transition $\tilde{m_0}$.

\subsubsection{More data, less work}
Our results suggests interesting tradeoffs between data and computational resources. To see this for the moment let us focus on Gaussian design matrices and convex $f$. As we mentioned in real applications of linear inverse problems we observe noisy samples $\vct{y}=\mtx{A}\vct{x}+\vct{w}$ and wish to solve
\begin{align}
\label{optlass}
\hat{\vct{x}}=\underset{\vct{z}}{\arg\min}\quad\twonorm{\vct{y}-\mtx{A}\vct{z}}^2\quad\text{subject to}\quad f(\vct{z})\le f(\vct{x}).
\end{align}
Let $\epsilon$ be the desired relative accuracy of the optimal solution, i.e. $\epsilon=\twonorm{\hat{\vct{x}}-\vct{x}}/\twonorm{\vct{x}}$. It is known (e.g. see \cite{StojLAS, oymak2013simple,plan2015generalized}) that for Gaussian measurements and convex $f$  
\begin{align*}
\frac{\twonorm{\hat{\vct{x}}-\vct{x}}}{\twonorm{\vct{x}}}=\epsilon\lesssim \frac{1}{\sqrt{m}-\sqrt{m_0}}\frac{\twonorm{\vct{w}}}{\twonorm{\vct{x}}}.
\end{align*}
We can think of $\epsilon$ as the ``statistical accuracy" of our solution that is the answer we would get if we solved the optimization problem \eqref{optlass} with arbitrary precession. Now let us compare this answer with our results for projected gradient descent which states that using $\mu\approx \frac{1}{m}$ after $\tau$ iterations for $m\ge 6.85m_0$ we have
\begin{align*}
\frac{\twonorm{\vct{z}_\tau-\vct{x}}}{\twonorm{\vct{x}}}\le \left(6.85\frac{m_0}{m}\right)^{\frac{\tau}{2}}+\frac{\sqrt{\pi}}{8}(3+\sqrt{5})\frac{1}{\left(1-\sqrt{6.85\frac{m_0}{m}}\right)}\frac{\sqrt{m_0}}{m}\frac{\twonorm{\vct{w}}}{\twonorm{\vct{x}}}.
\end{align*}
Simple manipulations imply that for $m\ge 6.85m_0$ we have
\begin{align}
\label{precision}
\frac{\twonorm{\vct{z}_\tau-\vct{x}}}{\twonorm{\vct{x}}}\le \left(6.85\frac{m_0}{m}\right)^{\frac{\tau}{2}}+9\epsilon.
\end{align}
The latter expression is an upper bound on the ``numerical accuracy" we get when using projected gradient descent. Since our statistically accuracy is $\epsilon$, it is natural to aim for a numerical precision of the same order e.g.~$10\epsilon$. In order to guarantee this numerical accuracy, based on \eqref{precision} the number of iterations must obey
\begin{align*}
\tau\ge 2\frac{\log\left(\frac{1}{\epsilon}\right)}{\log\left(\frac{m}{6.85m_0}\right)}.
\end{align*}
This expression is rather interesting as it shows that the larger $m$ or data complexity is, the smaller the number of iterations required to get to a specific numerical accuracy of $10\epsilon$. To gain insight on the computational complexity note that in each iteration we have to apply $\mtx{A}$ and its transpose $\mtx{A}^*$ once and also perform a projection onto the set $\mathcal{K}$. Using $\mathcal{X}$ to denote the computational complexity of an operation, the total computation complexity to arrive at a numerical accuracy of $\epsilon$ is
\begin{align*}
\text{Time budget}\propto&\mathcal{X}(\text{numerical accuracy of } \epsilon)\\
&\ge2\frac{\log\left(\frac{10}{\epsilon}\right)}{\log\left(\frac{m}{6.85m_0}\right)}\bigg(\mathcal{X}(\text{apply }\mtx{A})+\mathcal{X}(\text{apply }\mtx{A}^*)+\mathcal{X}(\text{projection on to }\mathcal{K})\bigg).
\end{align*}
As a result we can characterize the required number of measurements as a function of the structural complexity (via $m_0$), time budget and the cost of applying $\mtx{A}$, its transpose and projection onto the set $\mathcal{K}$. Focusing on Gaussian matrices and assuming that the projection is negligible compared to applying $\mtx{A}$ and its transpose. We note that for many interesting problems that arise in practice this is true. For example, when $f$ is the $\ell_1$ norm the computational complexity of projection is on the order of $n\log n$ which is less costly than applying $\mtx{A}$ and its transpose. Therefore, we can deduce that
\begin{align*}
\text{computational complexity}\propto \frac{m}{\log(\frac{m}{6.85m_0})}n.
\end{align*}
This formula shows that as we increase the data complexity $m$ the computational complexity decreases up to a certain point ($m\le 21.75m_0$). However, if we increase the number of measurements further the computational complexity will start increasing again.

The tradeoff between data and computational resources is even more interesting for SORS ensembles specially when applying $\mtx{A}$ and its transpose is independent of the number of measurements $m$. For example, for subsampled Fourier with random sign or subsampled Hadamard and random sign we have
\begin{align*}
\text{computational complexity}\propto \frac{n\log n}{\log(\frac{m}{c(\log n)^4m_0})}.
\end{align*}
The further we increase the data complexity the smaller the computational complexity: More data means less work!

To demonstrate that this is indeed the correct behavior let us focus on the computational complexity when $f(\vct{z})=\onenorm{\vct{z}}$ and when the structured signal has sparsity $s$. We run projected gradient descent until the relative error is small, i.e.~$\twonorm{\vct{z}_\tau-\vct{x}}/\twonorm{\vct{x}}\le10^{-3}$. We use a sparsity level of $s=0.025n$ and compare the real run time of the algorithm with theoretical predictions of the form $\alpha_G\frac{m}{\log\left(\frac{m}{\beta_G m_0}\right)}$ and $\alpha_F\frac{1}{\log\left(\frac{m}{\beta_F m_0}\right)}$ for Gaussian and Fourier with random sign ensembles. The values of $\alpha$ and $\beta$ are chosen to minimize the Euclidean norm of the run time as compared to these theoretical predictions. We note that the best choices are $\beta_G=4.054$ and $\beta_F=3.5$. Figure shows that our theoretical predictions is a near ideal match with the actual run time of the algorithm (This is with the caveat that we can not predict $\beta_G$ and $\beta_F$ precisely of course. Our precise theoretical prediction for $\beta_G$ was $8$ but as we mentioned before this constant is only an upper bound and not precise but interestingly only off by a factor of $1.7$). As a result these curves confirm our earlier observation that for Gaussian ensembles increasing the number of measurements (or data complexity) decreases the runtime as long as $m\le 10.25m_0$ but further increases in the data complexity will lead to an increase in the run time. For the Fourier with random sign ensemble however, increasing the data complexity will only lead to further decrease in the run time. More data, less work!

 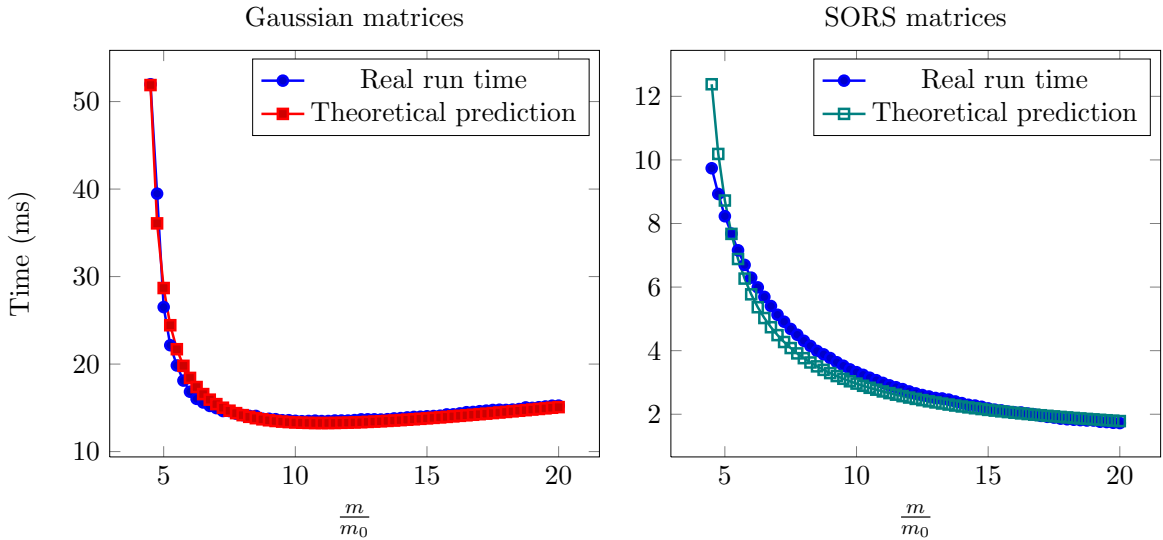
\begin{figure}[!h]
        \centering
\begin{tikzpicture}[scale=0.95] 
\begin{groupplot}[group style={group size=2 by 1,horizontal sep=1cm,xlabels at=edge bottom, ylabels at=edge left,xticklabels at=edge bottom},xlabel=$\frac{m}{m_0}$,
        ylabel=Time (ms)]
 \nextgroupplot[title={Gaussian matrices}]
        
        \addplot +[blue,line width=1pt] table[x index=0,y index=1]{./mdlwgr};\addlegendentry{Real run time}
        \addplot +[red,line width=1pt] table[x index=0,y index=1]{./mdlwgt};\addlegendentry{Theoretical prediction}
        
        \nextgroupplot[title={SORS matrices}]
        \addplot +[mark=*,solid,blue,line width=1pt] table[x index=0,y index=1]{./mdlwfr};\addlegendentry{Real run time}
        \addplot +[mark=square ,solid,teal,line width=1pt] table[x index=0,y index=1]{./mdlwft};\addlegendentry{Theoretical prediction}
        \end{groupplot}
\end{tikzpicture}
\caption{Blue curves show actual run time of projected gradient descent in milliseconds. Red curves shows the best scaling of the functions $\frac{m}{\log\left(\frac{m}{4.054m_0}\right)}$ and $\frac{1}{\log\left(\frac{m}{3.5m_0}\right)}$ for the Gaussian and sub-sampled Fourier with random sign ensembles, respectively.}
\label{SNRs}
\end{figure}

\subsection{Experiments on natural images}

In this section we demonstrate the utility of an image denoiser for image recovery from compressed measurements. We refer the reader to \cite{Montanaritalk, metzler2014denoising} for similar simulations with the AMP algorithm. For this purpose we shall use the CBM3D denoiser of \cite{dabov2007image}. First let us demonstrate that CBM3D is indeed well suited for denoising. To this aim we consider two images and add Gaussian noise to them. The first image is of the Eram Garden in the central Iranian city of Shiraz. The second image depicts fine nerve fibres highlighted with a fluorescent dye. Figure \ref{imdenoisf} shows the performance of CBMD3 on these two images with added Gaussian noise. The signal to noise ratio in these images is reported in terms of PSNR ($10\log_{10}($number of pixels $/\twonorm{\hat{\vct{x}}-\vct{x}}^2)$). These figures clearly show that CBMD3 is a good denoiser for images in the presence of Gaussian noise.

\begin{figure}
\centering
\begin{subfigure}[b]{0.3\textwidth} \includegraphics[width=\textwidth]{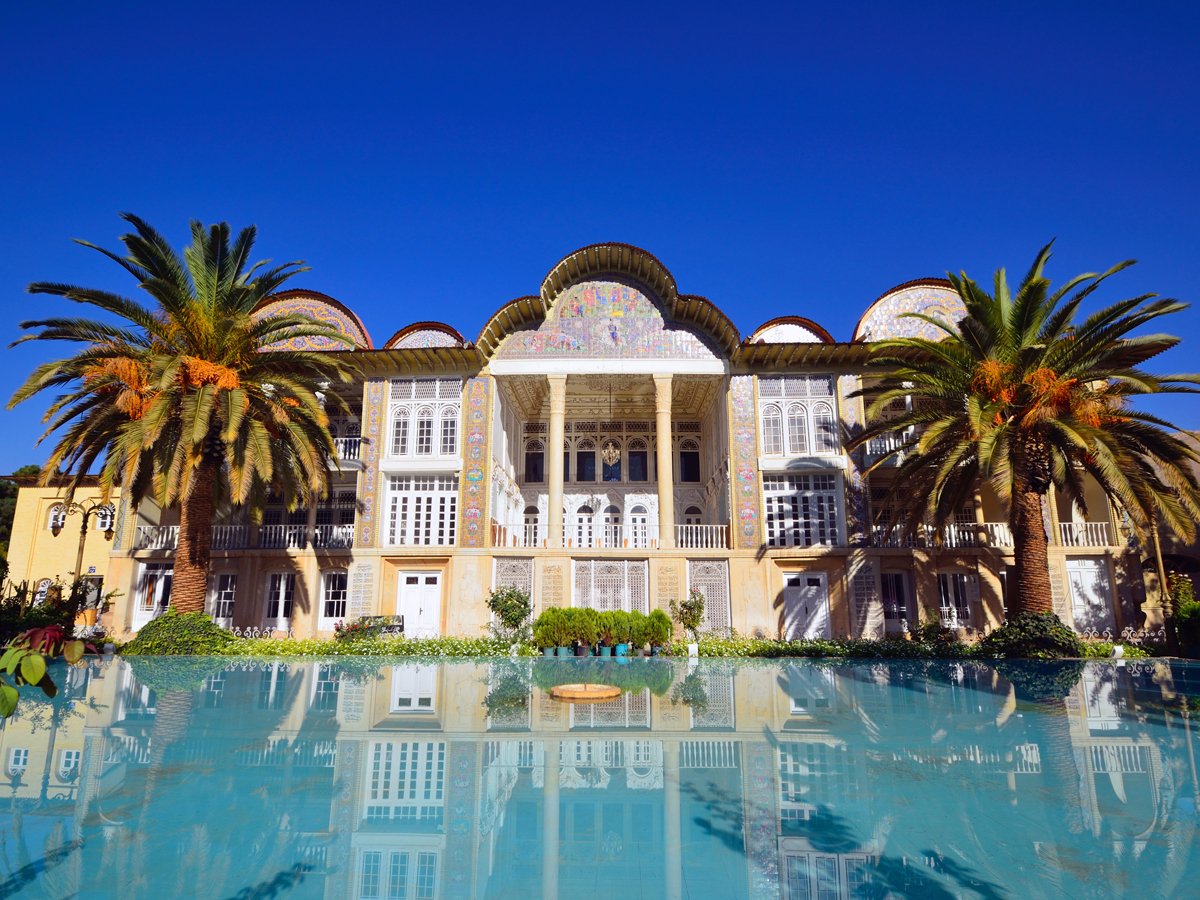} \caption{Eram Garden, Shiraz.} \label{fig:shiraz} \end{subfigure} 
\begin{subfigure}[b]{0.3\textwidth} \includegraphics[width=\textwidth]{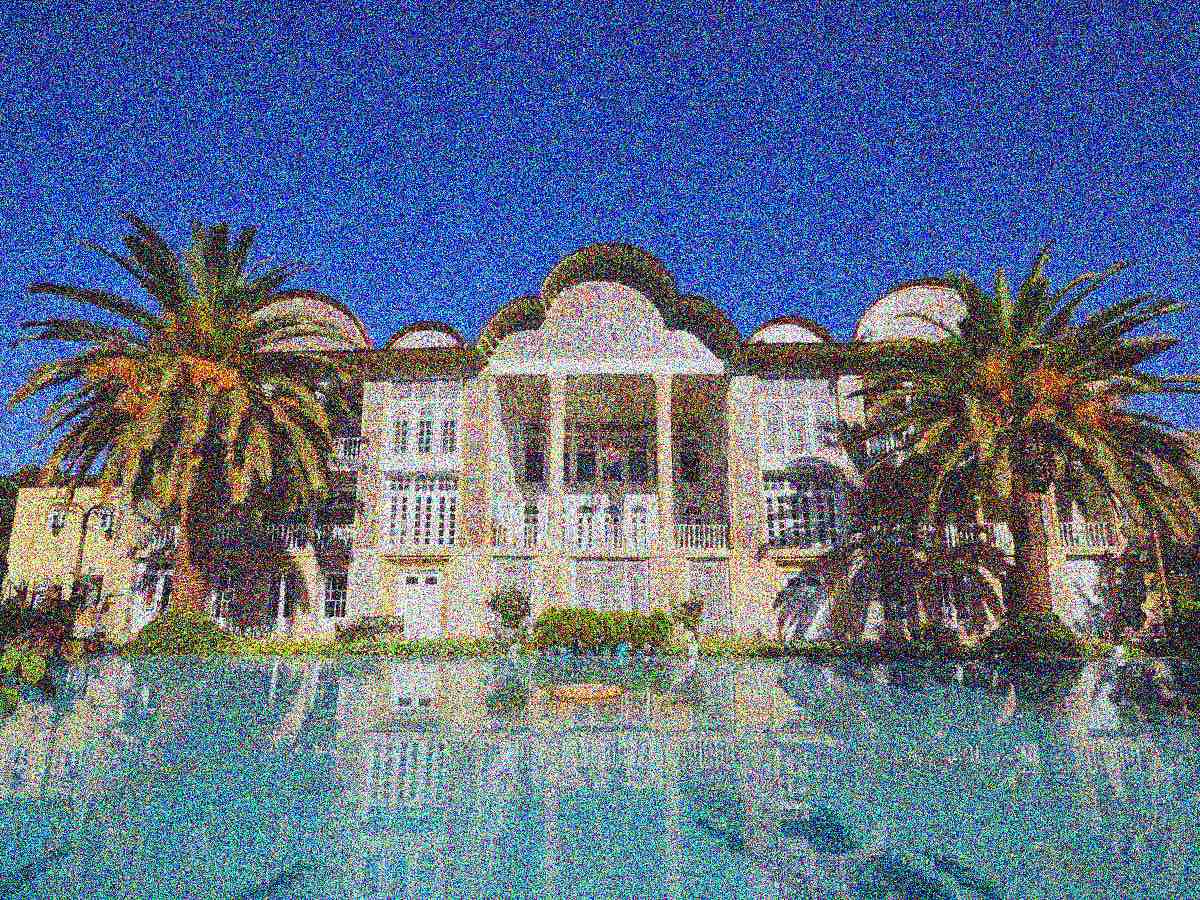} \caption{Noisy image, PSNR=8.13.} \label{fig:shiraznoisy} \end{subfigure} 
\begin{subfigure}[b]{0.3\textwidth} \includegraphics[width=\textwidth]{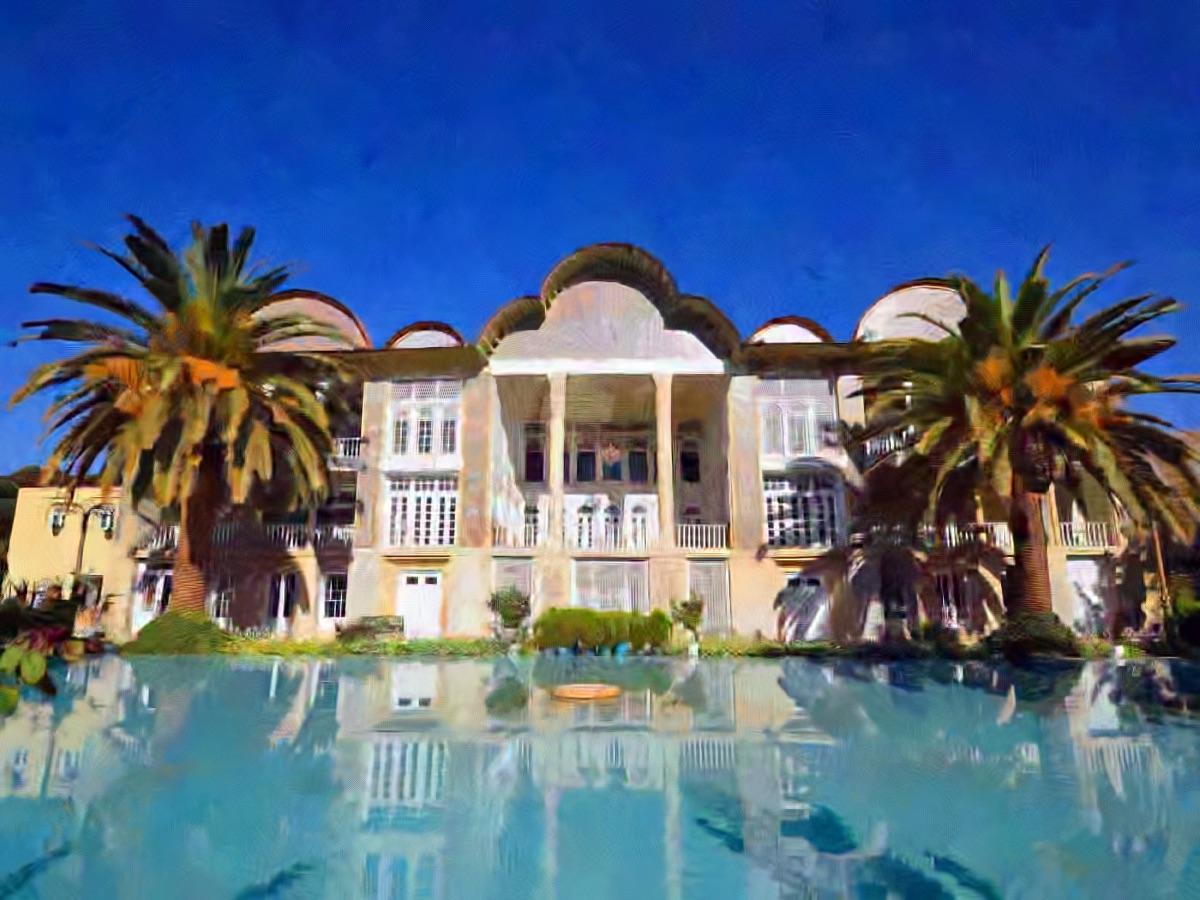} \caption{\hspace{-0.1cm}Denoised image, PSNR=24.3.} \label{fig:shirazdenoisy} \end{subfigure} 
\\
\begin{subfigure}[b]{0.3\textwidth} \includegraphics[width=\textwidth]{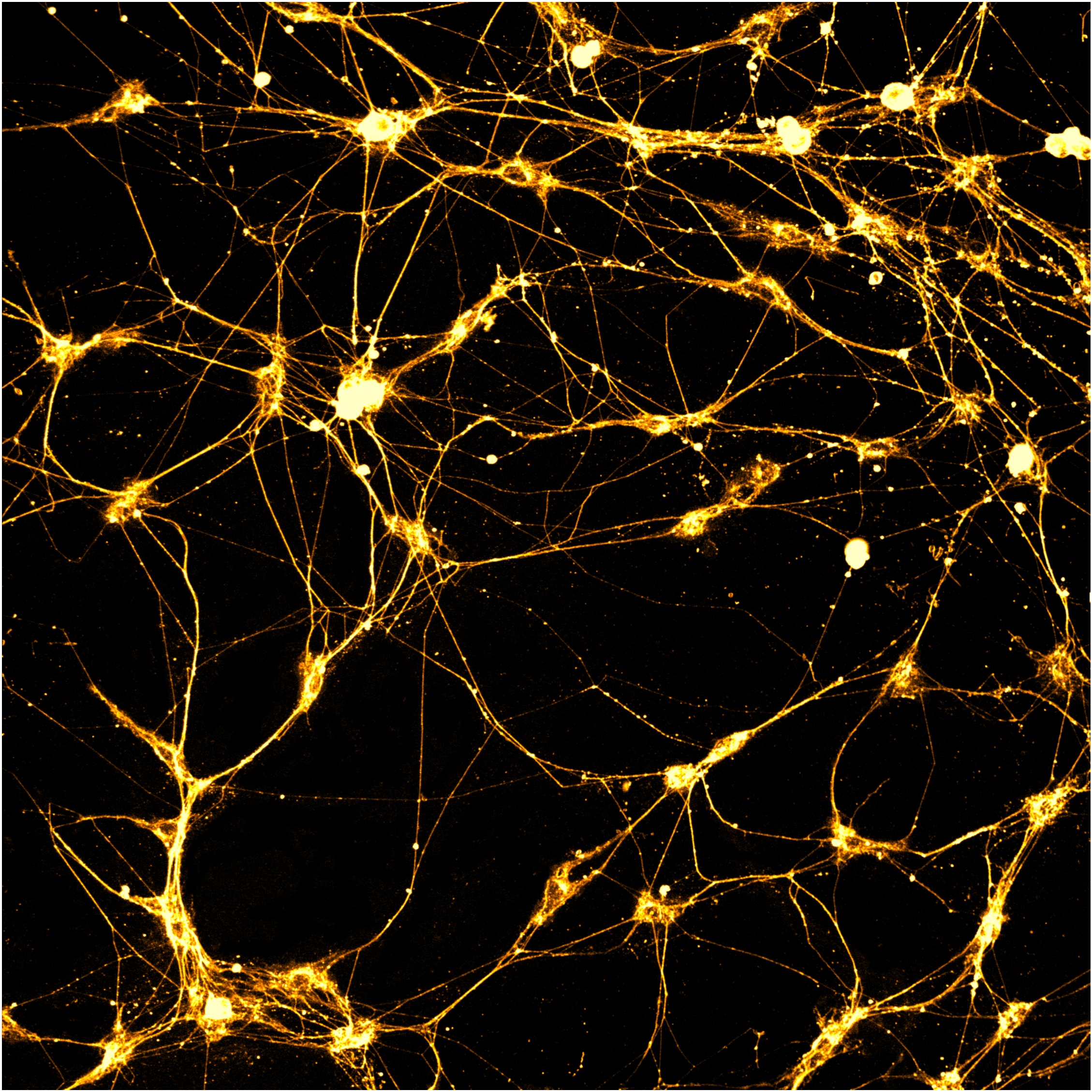} \caption{Fine nerve fibres.} \label{fig:gull} \end{subfigure}
\begin{subfigure}[b]{0.3\textwidth} \includegraphics[width=\textwidth]{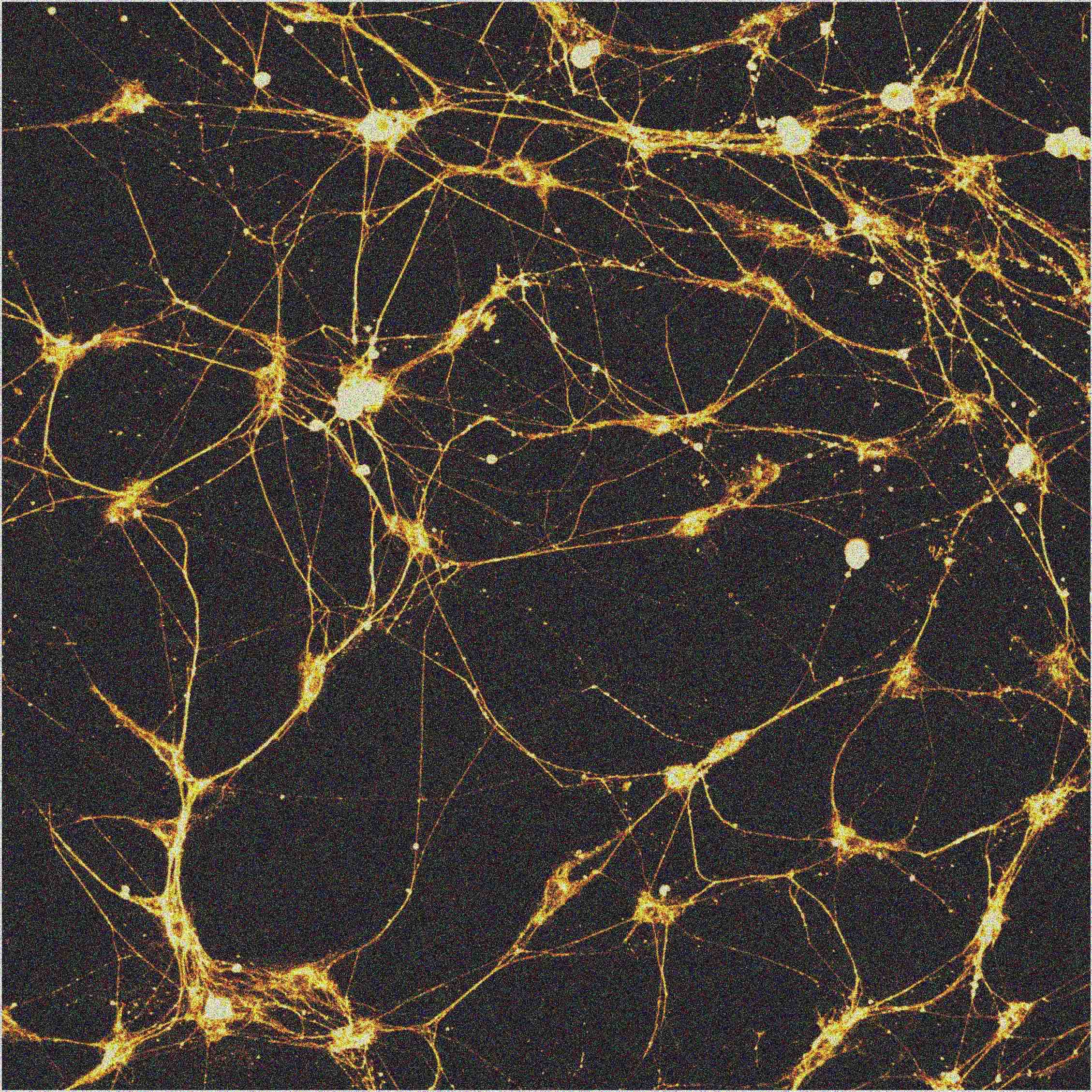} \caption{Noisy image, PSNR = 8.13.} \label{fig:tiger} \end{subfigure} 
\begin{subfigure}[b]{0.3\textwidth} \includegraphics[width=\textwidth]{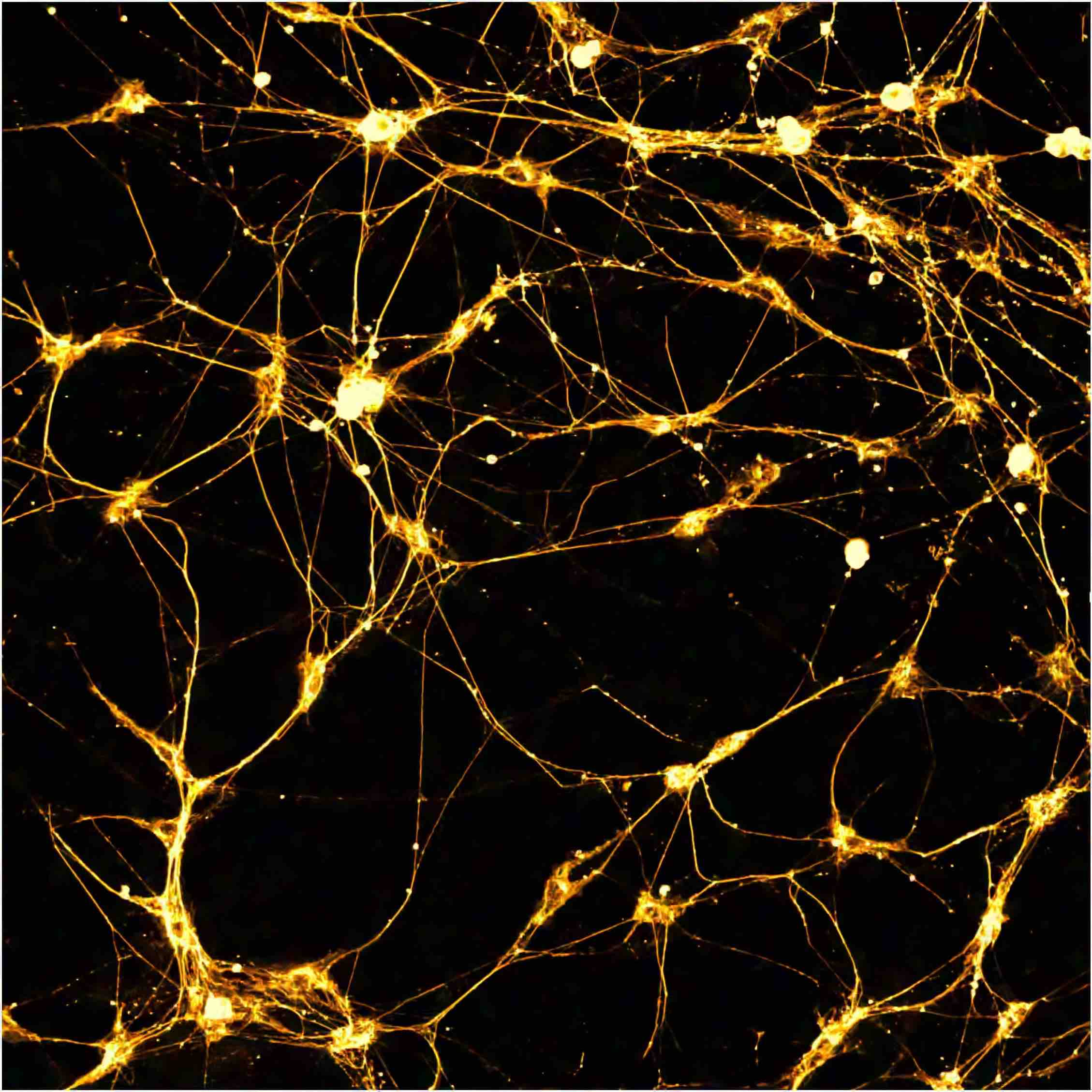} \caption{\hspace{-0.1cm}Denoised image, PSNR=26.2.} \label{fig:mouse} \end{subfigure} 
\caption{Image denoising using CBM3D.}\label{imdenoisf}
\end{figure}

We now turn our attention to using the CBM3D denoiser for recovery of an image from undersampled linear measurements. Our measurement process consists of modulating each pixel of the image by a random i.i.d. $\pm 1$ mask, applying a two dimensional Discrete Fourier Transform (DCT), and then picking a random subset of size $m$ of these DCT coefficients. Since each photograph is in color, we apply this measurement process to each color band for a total of $3m$ measurements. We shall use the short-hand $\mathcal{A}: \R^{3n}\rightarrow\R^{3m}$ to denote this linear measurement process, where $\mathcal{A}(\vct{x})$ takes a color image $\vct{x}\in\R^{3n}$ (with a total of $3n$ pixels) and outputs a vector $\vct{y}=\mathcal{A}(\vct{x})\in\R^{3m}$ containing $3m$ measurements. We note that this measurement process is an example of SORS matrices obtained from Bounded Orthogonal Systems (as discussed in Section \ref{sec:nonGauss}) and as a result our theory applies to this measurement process. 

We now wish to recover the original image from such under-sampled measurements. We start from $\vct{z}_0=\vct{0}$ and apply the following proximal gradient updates
\begin{align}
\label{proximalupdate}
\vct{z}_{\tau+1}=\mathcal{S}\left(\vct{z}_\tau+\mu\mathcal{A}^*(\vct{y}-\mathcal{A}(\vct{z}_\tau));\lambda_\tau\right).
\end{align}
Here $\mathcal{S}$ is the denoising procedure with tunning parameter $\lambda_\tau$. We shall use the CBM3D denoiser for $\mathcal{S}$. We note that this iterative update is directly related to the projected gradient update discussed throughout this paper. Indeed, the projected gradient update with the level $R$ set to $f(\vct{x})$ can be viewed as a denoiser $\mathcal{S}$ with $\lambda_\tau$ tuned optimally. In a companion paper \cite{ORScomp}, we will discuss how the theoretical framework of this paper generalizes to the proximal update in \eqref{proximalupdate} even when $\lambda_\tau$ is not optimally tuned. We shall use $\lambda_\tau=\lambda_0 \gamma^\tau$ in our experiments. This choice is based on results in \cite{ORScomp} which suggest that this is a good tuning strategy. We now apply the proximal update with the CBM3D denoiser with $\lambda_0=28.8675\twonorm{\vct{x}}\frac{\sqrt{n}}{m}$ and $\gamma=0.95$. We remind the reader that the image has $3n$ total pixels ($n$ pixels in each color band) we vary the ratio $m/n$ in the interval $[0,1]$. For each value of $m$ we run $200$ iterations of the proximal update \eqref{proximalupdate} and record the relative error $\frac{\twonorm{\hat{\vct{x}}-\vct{x}}}{\twonorm{\vct{x}}}$ (color images are viewed as a large vector). We report the results in Figure \ref{ims}. This figure indicates that the relative error decreases as a function of the number of measurements. Furthermore, this value falls below $5\%$ for $m\ge0.3n$ and $m\ge0.4n$ for the Eram Garden and Fine nerve fibres images, suggesting that $30\%$ and $40\%$ under-sampling may be enough to approximately recover the image. Figure \ref{imdenois} shows that indeed the recovered images are good. We note that we did not expect projected gradient to recover the images exactly as the de-noiser may not be perfect in capturing the structure of a real image. The correct analogy in sparse recovery literature is that the image is not ``exactly" sparse. Rather it is only approximately sparse.

\begin{figure}
\centering
\begin{tikzpicture}[scale=1] 
\begin{groupplot}[group style={group size=1 by 1,horizontal sep=1cm,xlabels at=edge bottom, ylabels at=edge left,xticklabels at=edge bottom},xlabel=under-sampling ratio ($\frac{m}{n}$),
        ylabel=Relative error, legend style={font=\tiny,at={(0.775,0.98)},anchor=north,legend cell align=left}]
        \nextgroupplot[title={Relative error vs. number of measurements.}]
        \addplot [magenta,line width=1pt] table[x index=0,y index=1]{./shirazerrs};\addlegendentry{Eram Garden}
        \addplot [teal,line width=1pt] table[x index=0,y index=1]{./braincellerrs};\addlegendentry{Fine nerve fibres}
 \end{groupplot}
\end{tikzpicture}
\caption{Performance of CBM3D based proximal gradient descent on two images. Figure depicts relative error of the recovered image ($\twonorm{\hat{\vct{x}}-\vct{x}}/\twonorm{\vct{x}}$) as a function of the number of measurements $m$.}\label{ims}
\end{figure}
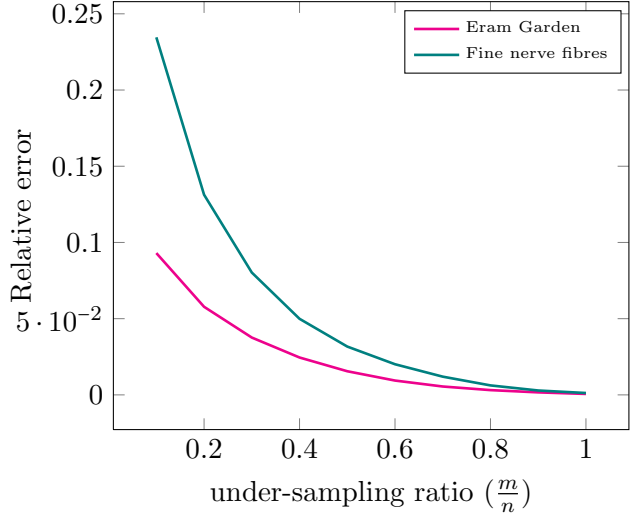


\begin{figure}
\centering
\begin{subfigure}[b]{0.475\textwidth} \centering\includegraphics[width=1\textwidth]{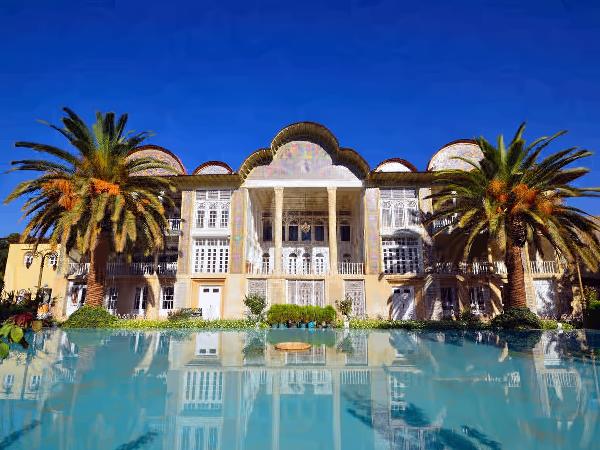} \caption{$m=0.3n$, PSNR=$29.1007$.} \label{fig:shiraz} \end{subfigure} 
\begin{subfigure}[b]{0.475\textwidth} \centering\includegraphics[width=0.75\textwidth]{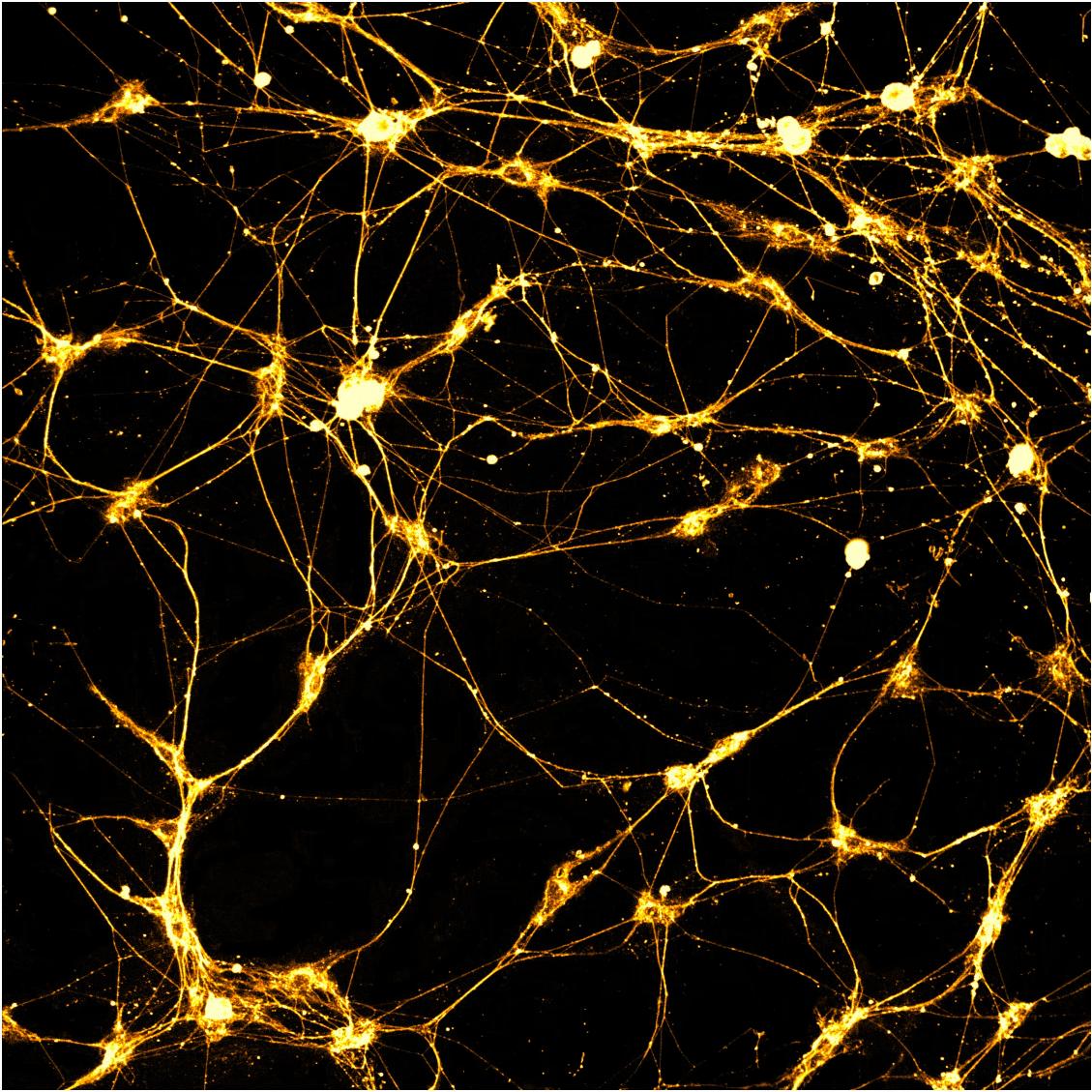} \caption{$m=0.4n$, PSNR=$38.2208$.} \label{fig:shiraznoisy} \end{subfigure} 
\caption{Recovered images from $3m$ under-sampled DCT measurements using CBM3D proximal gradient descent. Here, $3n$ is the total number of pixels, $n$ pixels for each color band.}\label{imdenois}
\end{figure}

\section{Prior Art}
\label{sec:prior}
Optimization problems such as \eqref{opt1} and in particular $\ell_1$-minimization techniques have been emppirically used to find structured solutions to underdetermined linear systems in multiple scientific fields \cite{claerbout1973robust, hogbom1974aperture,Tikh, DonAtom}. While there were some theoretical results describing the success of these algorithms in recovering structured solutions, the required number of measurements were sub-optimal. The last decade has seen a flurry of activity surrounding these problems in part due to the sample optimal results obtained for certain random measurement ensembles in sparse recovery \cite{Comp, Don1} and rank minimization \cite{RechtFazel, MatCom, keshavan2009matrix, gross2011recovering} problems. 

A significant focus in the literature has been on characterizing the precise number of measurements required for recovering the structured signal exactly via \eqref{opt1} in the special case where the entries of the measurement matrix are i.i.d.~real-valued standard normal variables. For sparse/matrix recovery problem the papers \cite{Don1, candes2006near, RudelSparse, RechtFazel} established sample optimal results up to numerical constants. For sparse recovery problems . These constants were later made more precise in \cite{Cha, candes_recht12}. Earlier Donoho \cite{DonCentSym} and Donoho and Tanner \cite{DonPhaseTrans, DonNonNegative, Don3} showed that the precise number of measurements (as a function of sparsity) required for the success of $\ell_1$ minimization can be characterized by an asymptotic phase transition (a.k.a.~the Donoho-Tanner curve). More recently, heuristic arguments from statistical physics where used to explain this phase transition \cite{AMP}. These heuristics were justified in \cite{Mon, BayMon} for Gaussian matrices and in \cite{BLM12} for matrices with sub-Gaussian entries perturbed by small Gaussian noise. Related, in \cite{Sto1} Stojnic established lower bounds on the number of measurements that matched the Donoho-Tanner curve in an asymptotic regime. Please also see \cite{rechtxu, Oym2} for related results for rank minimization via the Nuclear norm. More recently, Chandrasekaran, Recht, Parrilo and Willsky \cite{Cha} derived the first precise non-asymptotic bounds via Gordon's lemma \cite{Gor}. It was shown in \cite{McCoy} that these lower bound are asymptotically sharp. Please also see \cite{StojLAS, StojSOCP, StojSharp} for results of a similar flavor in special cases. In Theorem \ref{ALTthm} we have shown that projected gradient descent with a properly chosen step size can match these asymptotically sharp results for convex functions, providing an algorithmic proof. Furthermore, our results in Theorem \ref{first prop} are more generally applicable and also apply to any function including non-convex functions albeit with a loss in terms of a small constant. 

For more structured measurement ensembles, such as subsampled Fourier matrices, Candes and Tao \cite{Comp} obtained near optimal sample complexity results for $\ell_1$ minimization that were sharp upto logarithmic factors. These results were further improved in \cite{RudelSparse, cheraghchi2013restricted}. Please also see \cite{WainLasso, wainwright2009information, PlanRIPless, rauhut2010, foucart2013random, krahmer2014suprema} and references therein for similar results for many other ensembles. Please also see \cite{MatCom, candes2014phase, WF} for other structured ensembles arising in rank minimization problems. Moving beyond sparse recovery and rank minimization \cite{negahban2009unified, agarwal2010fast} establishes near sample optimal results (up to logarithmic factors) for decomposable norms. Such norms are convex and include important norms such as $\ell_1$ and nuclear norms. In this paper utilizing our results from a companion paper \cite{oymak2015isometric} we have established near sample optimal results for SORS matrices that apply to any function convex or non-convex.  To the extent of our
knowledge this is the first sample optimal result using a computational friendly matrix that applies to general signal structures and functions. Indeed, we are unaware of any sample optimal results even for general convex functions.

The prior works we have discussed so far are not algorithmic in nature in the sense that often the properties of a convex optimization problem have been studied without focus on how such problems are actually solved. While in principal such problems can be solved in polynomial time via interior point methods for many application first order methods may be more suitable. Bounds on convergence rates of first order schemes \cite{nesterov2004introductory} for general convex functions are known. However, since the penalty functions used for solving linear inverse problems are non-smooth, these bounds are often very pessimistic and rather far from the actual rates of convergence. For example \cite{becker2011nesta,beck2009fast,ji2009accelerated}, establish sub-geometric upper bounds on convergence rate for sparse and low-rank estimation problems via proximal/projected gradient algorithms such as \cite{nesterov2007gradient}. Related see also interesting works by Goldfarb, Osher, Yin and collaborators \cite{osher2005iterative, yin2008bregman, chartrand2008iteratively}. In practice the convergence rate of these algorithms are often geometric. Agarwal, Negahban and Wainwright \cite{agarwal2010fast} showed that when the cost function is a decomposable norm the rate of convergence of projected gradient descent is indeed geometric. See also \cite{loh2013regularized} for related results/generalizations. The authors obtained geometric convergence rates in terms of restricted eigenvalues. Specialized to Gaussian measurements these results do not come with explicit or sharp constants. In contrast we have developed precise convergence rates that applies to any function (convex or non-convex). In the case of Gaussian measurements we have also provided sharp constants. More recently, a few papers focus on a penalized formulation of \eqref{opt1}. For example, Xiao and Zhang \cite{xiao2013proximal} obtained geometric rates of convergence based on the Restricted Isometry Property. These results were further generalized by Eghbali and Fazel to decomposable norms \cite{eghbali2015decomposable}. We defer the study of proximial algorithms for general functions to a future publication. 
Convex functions are not the only way to enforce structure and often it may be more suitable to use non-convex functions. For sparse recovery problems the first results on geometric (and in fact linear) convergence of such problems are due to Tropp and Gilbert \cite{OMP2} and Garg and Khandekar \cite{garg2009gradient} who studied convergence of greedy hard-thresholding algorithms. See also related greedy strategies \cite{CoSamp, needell2009uniform} which have similar guarantees. More broadly convergence of iterative hard thresholding and its variants such as singular value hard thresholding have been studied by multiple authors \cite{blumensath2009, jain2009guaranteed, cai2010singular,jain2014iterative, figueiredo2007gradient, mazumder2010spectral, bahmani2013unifying}. These authors show that under RIP or matrix RIP conditions such algorithms have a linear convergence rate to the structured solution. Please also see related results \cite{foucart2010gelfand, saab2008stable} on characterizing the properties of the optimal solution when using non-convex $\ell_p$ (with $p<1$) minimization problems for recovering sparse signals. We note that the latter results do not ensure that $\ell_p$ minimization is tractable as they do not have convergence guarantees. Our general framework covers any function and thus allows for precise analysis of rates of convergence such algorithms as well. In particular, for Gaussian ensembles we are able to also provide sharp constants.

A few recent papers study computation-statistical tradeoffs for machine learning problems. For example, in \cite{ChaJor} Chandrasekaran and Jordan focus on such tradeoffs in the context of denoising via hierarchies of convex relaxations. In contrast, our focus is on a fixed cost function and linear inverse problems. More specifically, we focus on computation-statistical tradeoffs based on the convergence rate of projected gradients. We refer to Agarwal's dissertation \cite{agarwal2012computational} and references therein for other examples. Finally, in \cite{bruer2014time,bruer2014designing} the authors study computation-statistical tradeoffs for linear inverse problems. See also the very recent paper characterizing statistical properties of the optimal solution \cite{li2015geometric} based on results in \cite{vershynin2014estimation} brought to our attention by Volkan Cevher. The focus in these papers are not directly on projected gradients. Rather, the authors focus on ``lowering the computational cost through additional smoothing". In particular, the rates of convergence used in the analysis of these papers are not geometric and are based on more classic results where the error is inversely proportional to the square of the number of iterations.

Finally, we would like to mention related work surrounding the Approximate Message Passing (AMP) by Donoho, Montanari and Maleki \cite{AMPmain}. See also \cite{rangan2011generalized, rangan2013fixed} for generalizations of AMP. The updates in AMP are similar to a projected gradient update with an additional memory term similar in nature to Nesterov's accelerated scheme \cite{nesterov1983method}. However, motivated by ideas in statistical physics this algorithm comes with a very specific recipe for the coefficient of this memory term (a.k.a.~the Onsager term) \cite{DonAcu}. For Gaussian measurement ensembles and separable functions of the form $f(\vct{x})=\sum_{i=1}^n g(\vct{x}_i)$ with $g$ a pseudo-Lipschitz function it has been shown that the convergence of this algorithm can be described asymptotically via certain state evolution equations \cite{BayMon, Mon}. An important consequence of this analysis is that AMP asymptotically achieves a linear rate of convergence of $\sqrt{\frac{m_0}{m}}$ for pseudo-Lipschitz and separable functions exactly above the phase transition of \eqref{opt1} ($m>m_0$).  See also \cite{DonAcu, zheng2015does} for predictions of performance of the AMP framework for general functions using highly plausible but non-rigorous statistical physics arguments. In contrast, our results are non-asymptotic, rigorous and apply to any function. However, our theoretical results and empirical studies show that projected gradient descent can not achieve a linear rate of convergence exactly above the phase transition of \eqref{opt1} even when the cost function is separable. We believe that our theoretical framework may provide useful insights for precise, non-asymptotic analysis of AMP for general functions.

\section{Discussion}\label{discuss}
In this paper we have discussed sharp time-data tradeoffs for linear inverse problems. We would like to pause to discuss how our results can be further improved.

\begin{itemize}
\item \textbf{From denoising to compressive sensing.}
In addition to the connection between the rate of convergence and sample complexity, our results also connect de-noising and compressed sensing. That is, we establish a precise connection between the capability of a function $f$ when used for de-noising and its capability for recovering a structured signal via the PGD updates. To be more specific, it is easy to show
\begin{theorem}\label{denoising}
Consider a structured signal $\vct{x}\in\R^n$ and assume we observe a noisy version of this signal $\vct{x}+\vct{w}$, where $\vct{w}\in\R^n$ is Gaussian noise distributed as $\mathcal{N}(\vct{0},\sigma^2\mtx{I})$. Furthermore, let $\mathcal{P}_{\mathcal{K}}$ with $\mathcal{K}=\{\vct{z}\in\R^n \text{ }|\text{ }f(\vct{z})\le f(\vct{x})\}$. Then
\begin{align*}
\sup_{\sigma>0}\frac{\mathbb{E}\bigg[\twonorm{\mathcal{P}_{\mathcal{K}}\left(\vct{x}+\vct{w}\right)-\vct{x}}^2\bigg]}{\sigma^2}\le 4\omega^2(\mathcal{C}_f(\vct{x})\cap\Bc^n)\approx 4m_0.
\end{align*} 
\end{theorem}
This theorem combined with the results in this paper shows that the min-max rate of denoising is intimately related with the phase transition curve for exact signal recovery when using $f$ (up to a small constant less than 4). To the extent of our knowledge this is the first result to establish this connection for general $f$. For convex $f$ please see \cite{ChaJor} for analogous results to Theorem \ref{denoising} above. We note that the results for convex $f$ are sharper in the sense that it is known that the the min-max rate of denoising coincides with the phase transition curve for exact recovery. The connection between compressive sensing and denoising was conjectured in an asymptotic setting for the Approximate Message Passing (AMP) algorithm \cite{AMP,AMPmain} by Donoho, Johnson and Montanari in \cite{DonAcu}. Please also see \cite{metzler2014denoising} for related discussions and numerical simulations. Theorem \ref{denoising} above combined with the results proven in this paper provides strong evidence that this conjecture holds non-asymptotically for projected gradients using any function (it essentially proves this conjecture holds upto a small multiplicative constant). For convex $f$, prior work \cite{BayMon, Mon} established this connection first for LASSO and later on for general $f$ \cite{oymakProx, McCoy, Chandrasekaran2012} with a precise constant of $1$. We believe that our analysis may shed light on a complete proof of the conjecture of \cite{DonAcu} and we consider this a very interesting future research direction. 
\item \textbf{Sharper sample complexity for non-Gaussians.} Compared to our results for Gaussian ensembles our sample complexity results and convergences rates for non-Gaussian ensembles (ISG and SOS) are sub-optimal by constant/log factors. In particular we did not have results with sharp constants for these ensembles. Providing sharp constants in our sample complexity bounds is an interesting direction for future research. In particular, our experiments suggest that the sample complexity for SOS matrices matches that of Gaussian ensembles, i.e.~the phase transition for Gaussian and SOS matrices are essentially identical.

\item \textbf{Sharper time--data tradeoffs.} The numerical simulations in Section \ref{sec:rateofconvnum} and in particular Figure \ref{fig:choices2} suggest that the rate of convergence roughly scales like $\sqrt{\tilde{m}_0/m}$ where $\tilde{m}_0$ is the empirical phase transition and $m$ is the number of measurements. Establishing such a precise upper bound on the convergence rate is an interesting future direction as it would lead to more precise time-data tradeoffs. Related, the sample complexity of Theorems \ref{ALTthm} and \ref{PGthm} seem to match the phase transition curves obtained by numerical simulations in Section \ref{secnumpt}. However, these two theorems where only limited to convex functions. It would be interesting to establish a similarly sharp result for non-convex functions. While we did prove results for non-convex functions in Theorem \ref{first prop} our numerical results in Figure \ref{fig:greedy} seem to suggest that (a) $\kappa_f<2$ for non-convex functions and (b) the convergence rate even for convex functions can possibly be improved to $\rho(\mu)\le\frac{1}{\sqrt{2}-1}\sqrt{\frac{m_0}{m}}\approx\sqrt{5.8\frac{m_0}{m}}$ (and in turn the requirement on sample complexity can be improved to $m\ge \frac{1}{(\sqrt{2}-1)^2}m_0\approx 5.8 m_0$). Closing these gaps for the non-convex case is particularly important as it would help explain the surprising numerical results in Figure \ref{fig:choices} which seems to suggest that  $\ell_{1/2}$ works better than both $\ell_0$ and $\ell_1$ for sparse recovery problems.

\item \textbf{Extensions to other update strategies and general loss functions.} While in this paper we have focused on projected gradient updates we would like to point out that our results naturally extend to many other update strategies such as proximal gradients, stochastic updates etc. These results will be the subject of a companion paper \cite{ORScomp}. Also our results can be generalized to other loss functions (i.e.~non-$\ell_2$) in a straightforward manner which will be detailed in a future paper.
\end{itemize}

\section{Proofs}\label{secproofs}
In this section we shall prove all of the stated theorems. We will first prove the deterministic result of Theorem \ref{master} and then show in later sections how all other theorems follow by bounding (with high probability) the convergence rate $\rho(\mu)$ and noise interaction term $\xi_\mu(\mtx{A})$ of this theorem for various random ensembles and learning parameters $\mu$. 

Throughout we shall use $\mathbb{S}^{n-1}$ and $\mathcal{B}^n$ to denote the unit sphere and Euclidean ball of $\R^n$. For a set $\mathcal{K}\in\R^n$ and a point $\vct{x}\in\R^n$, we use $\mathcal{P}_{\mathcal{K}}(\vct{x})$ to denote the projection of $\vct{x}$ onto the set $\mathcal{K}$. For a set $\mathcal{K}$ and point $\vct{v}\in\R^n$ we define
\begin{align*}
\text{dist}(\vct{v},\mathcal{K})=\min_{\vct{u}\in\mathcal{K}}\text{ }\twonorm{\vct{v}-\vct{u}}.
\end{align*}
We pause to remind the reader of the definition of the descent set and tangent cone from Definition \ref{decsetcone}. The \emph{set of descent} of the function $f$ at a point $\vct{x}$ is defined as
\begin{align*}
{\cal D}_f(\vct{x})=\Big\{\vct{h}:\text{ }f(\vct{x}+\vct{h})\le f(\vct{x})\Big\}.
\end{align*}
The \emph{tangent cone} $\mathcal{C}_f(\vct{x})$ of $f$ at a point $\vct{x}$ is the conic hull of the descent set. That is, the smallest closed cone $\mathcal{C}_f(\vct{x})$ obeying $\mathcal{D}_f(\vct{x})\subset\mathcal{C}_f(\vct{x})$. For ease of exposition throughout the proof of this theorem we shall use the shorthand $\mathcal{D}$ and $\mathcal{C}$ to refer to $\mathcal{D}_f(\vct{x})$ and $\mathcal{C}_f(\vct{x})$.

We also remind the reader that $\rho(\mu)$ is the convergence rate defined as
\begin{align*}
\rho(\mu):=\rho(\mu,\mtx{A},f,\vct{x})=&\underset{\vct{u},\vct{v}\in\mathcal{C}\cap \mathbb{S}^{n-1}}{\sup}\vct{u}^*\left(\mtx{I}-\mu\mtx{A}^*\mtx{A}\right)\vct{v},
\end{align*}
$\xi_\mu(\mtx{A})$ is the noise amplification factor defined as
\begin{align*}
\xi_\mu(\mtx{A}):=\xi_\mu(\mtx{A},f,\vct{x},\vct{w})=\mu\cdot\underset{\vct{v}\in\mathcal{C}\cap\mathbb{S}^{n-1}}{\sup}\vct{v}^*\mtx{A}\frac{\vct{w}}{\twonorm{\vct{w}}},
\end{align*}
and $\kappa_f$ is a constant equal to one for convex functions and equal to two for non-convex functions. In order to prove our results we first state and prove a few results regarding cones in higher dimension.

\subsection{Cone and projection identities}
In this section we will gather a few results regarding higher dimensional cones and projections that are used throughout the proofs. We begin with a well known lemma about projections onto convex sets. This result is standard and we skip the proof.
\begin{lemma}\label{projcvxthm} Assume $\mathcal{K}\in\R^n$ is a convex set and $\vct{v}\in\R^n$. Then for every $\vct{u}\in\mathcal{K}$ we have
\begin{align}\label{projcvxthm1}
\twonorm{\vct{v}-\mathcal{P}_{\mathcal{K}}(\vct{v})}^2+\twonorm{\vct{u}-\mathcal{P}_{\mathcal{K}}(\vct{v})}^2\le\twonorm{\vct{v}-\vct{u}}^2.
\end{align}
In particular if $\vct{0}\in\mathcal{K}$, then
\begin{align}\label{projcvxthm2}
\twonorm{\vct{v}-\mathcal{P}_{\mathcal{K}}(\vct{v})}^2+\twonorm{\mathcal{P}_{\mathcal{K}}(\vct{v})}^2\le\twonorm{\vct{v}}^2.
\end{align}
\end{lemma}
We now state a result concerning projection onto cones.
\begin{lemma}\label{firstconelem} Let $\Cc\subset\R^n$ be a closed cone and $\vb\in\R^n$. The followings two identities hold
\begin{align}
\tn{\vb}^2=&\tn{\vb-\mathcal{P}_\Cc(\vb)}^2+\tn{\mathcal{P}_\Cc(\vb)}^2\label{orthogonal},\\
\tn{\mathcal{P}_\Cc(\vb)}=&\sup_{\ub\in\Cc\cap\Bc^{n}} \ub^*\vb\label{lem:identity}.
\end{align}
\end{lemma}
\begin{proof} 
The first identity is known as Moreau's decomposition and is standard. Let us turn our attention to proving \eqref{lem:identity}. If $\mathcal{P}_\Cc(\vb)=\vct{0}$, then $\vct{0}$ is the closest point to $\vb$. This implies that for all $t\ge0$ and all $\vct{u}\in\mathcal{C}$ we have
\begin{align*}
\twonorm{\vct{v}}^2\le\twonorm{\vct{v}-t\vct{u}}^2\quad\Rightarrow\quad\vct{u}^*\vct{v}\le t\twonorm{\vct{u}}^2.
\end{align*}
Taking the limit $t\rightarrow 0$, we conclude that for all $\vct{u}\in\mathcal{C}$ we have $\vct{u}^*\vct{v}\le 0$. Also since $\mathcal{P}_\Cc(\vb)=\vct{0}\in(\mathcal{C}\cap\mathcal{B}^n)$
\begin{align*}
\sup_{\ub\in\Cc\cap\Bc^{n}} \ub^*\vb=0=\tn{\mathcal{P}_\Cc(\vb)},
\end{align*}
holds which proves \eqref{lem:identity} when $\mathcal{P}_{\mathcal{C}}(\vct{v})=\vct{0}$. To address the case when $\mathcal{P}_{\mathcal{C}}(\vct{v})\neq\vct{0}$ set $\hat{\vct{u}}=\frac{\mathcal{P}_\Cc(\vb)}{\tn{\mathcal{P}_\Cc(\vb)}}$. Clearly $\hat{\vct{u}}^*\vb=\tn{\mathcal{P}_\Cc(\vb)}$ and $\hat{\vct{u}}\in\Cc\cap\mathcal{B}^{n}$ so that
\begin{align*}
\sup_{\ub\in\Cc\cap\Bc^{n}} \ub^*\vb\ge \hat{\vct{u}}^*\vct{v}=\tn{\mathcal{P}_\Cc(\vb)}.
\end{align*}
Therefore, to prove \eqref{lem:identity} it suffices to show $\sup_{\ub\in\Cc\cap\Bc^{n}} \ub^*\vb\le \tn{\mathcal{P}_\Cc(\vb)}$. Suppose this is not the case and there exists $\vct{u}'$ such that $\left({\vct{u}'}\right)^*\vb> \tn{\mathcal{P}_\Cc(\vb)}$. Set $\tilde{\vct{u}}=\vct{u}'/\twonorm{\vct{u}'}$ and note that since $\vct{u'}\in\mathcal{C}\cap\mathcal{B}^n$ we must have $\vb^*\tilde{\vct{u}}> \tn{\mathcal{P}_\Cc(\vb)}$. Now choosing $\vct{w}=(\vb^*\tilde{\vct{u}})\tilde{\vct{u}}\subset \Cc$ we have $\twonorm{\vct{w}}=\vct{v}^*\tilde{\vct{u}}>\twonorm{\mathcal{P}_{\mathcal{C}}(\vct{v})}$. The latter together with \eqref{orthogonal} implies
\begin{align*}
\text{dist}^2(\vb,\Cc)=\tn{\vb}^2-\tn{\mathcal{P}_\Cc(\vb)}^2>\tn{\vb}^2-\tn{\vct{w}}^2= \tn{\vb-\vct{w}}^2,
\end{align*}
which is in contradiction with $\vct{w}\in\Cc$. 
\end{proof}
%
The following lemma is straightforward and follows from the fact that translation preserves distances.
\begin{lemma} \label{simplem}Suppose $\mathcal{K}\subset\R^n$ is a closed set. The projection onto $\mathcal{K}$ obeys
\begin{align*}
\Pc_{\mathcal{K}}(\x+\vb)-\vct{x}=\Pc_{\mathcal{K}-\{\x\}}(\vb).
\end{align*}
\end{lemma}
The next lemma compares the length of a projection onto a set to the length of projection onto the conic approximation of the set. Our approach is in part inspired by well known results in geometric functional analysis e.g.~see \cite[Lemma 8.2]{plan2014high} for related calculations.
\begin{lemma} [Comparison of projections] \label{prop compare} Let $\mathcal{D}$ be a closed and nonempty set that contains $\vct{0}$. Let $\Cc$ be a nonempty and closed cone containing $\mathcal{D}$ ($\mathcal{D}\subset\mathcal{C}$). Then for all $\vct{v}\in\R^n$,
\begin{align}
\label{firstsetcomp}
\tn{\Pc_\Dc(\vb)}\le 2\tn{\Pc_\Cc(\vb)}
\end{align}
Furthermore, if $\Dc$ is a convex set. Then for all $\vb\in\R^n$, 
\begin{align}
\label{cvxsetcomp}
\tn{\Pc_\Dc(\vb)}\le \tn{\Pc_\Cc(\vb)}.
\end{align}
\end{lemma}
\begin{proof} Let us first show the statement \eqref{cvxsetcomp} for convex $\Dc$. Since $\Dc\subset \Cc$, $\tn{\vb-\Pc_\Cc(\vb)}\leq \tn{\vb-\Pc_\Dc(\vb)}$. Furthermore, when $\mathcal{D}$ is convex using the fact that $\vct{0}\in\mathcal{D}$ by Lemma \ref{projcvxthm} equation \eqref{projcvxthm2} we know that $ \tn{\vb-\Pc_\Dc(\vb)}^2+ \tn{\Pc_\Dc(\vb)}^2\leq  \tn{\vb}^2$. The latter two identities together with Lemma \ref{firstconelem} equation \eqref{orthogonal} yields
\begin{align*}
\tn{\Pc_\Cc(\vb)}^2=\tn{\vb}^2-\tn{\vb-\Pc_\Cc(\vb)}^2\geq \tn{\vb}^2-\tn{\vb-\Pc_\Dc(\vb)}^2\geq  \tn{\Pc_\Dc(\vb)}^2.
\end{align*}
Now, we turn our attention to proving \eqref{firstsetcomp} for nonconvex $\Dc$. By definition of projection onto sets
\begin{align*}
\tn{\vb}^2\geq \tn{\vb-\Pc_\Dc(\vb)}^2=\tn{\vb}^2+\tn{\Pc_\Dc(\vb)}^2-2\li\vb,\Pc_\Dc(\vb)\ri.
\end{align*}
This yields
\begin{align*}
\tn{\Pc_\Dc(\vb)}^2\le 2\li\vb,\Pc_\Dc(\vb)\ri=2\li\vb,\frac{\Pc_\Dc(\vb)}{\tn{\Pc_\Dc(\vb)}}\ri\tn{\Pc_\Dc(\vb)}\quad\Rightarrow\quad \tn{\Pc_\Dc(\vb)}\le2\li\vb,\frac{\Pc_\Dc(\vb)}{\tn{\Pc_\Dc(\vb)}}\ri.
\end{align*}
Since $\mathcal{C}$ is a cone that contains $\mathcal{D}$, $\frac{\Pc_\Dc(\vb)}{\tn{\Pc_\Dc(\vb)}}$ has unit length and belongs to $\Cc$ we have
\begin{align*}
\tn{\Pc_\Dc(\vb)}\le2\li\vb,\frac{\Pc_\Dc(\vb)}{\tn{\Pc_\Dc(\vb)}}\ri\l\le2\cdot\underset{\vct{u}\in\mathcal{C}\cap\mathcal{B}^{n}}{\sup}\text{ }\vct{u}^*\vct{v}\le2\twonorm{\mathcal{P}_{\mathcal{C}}(\vct{v})},
\end{align*}
where in the last inequality we have used Lemma \ref{firstconelem} equation \eqref{lem:identity}. This completes the proof of \eqref{firstsetcomp}.
\end{proof}
We end this section by proving a lemma that shows that scaling a vector increases its projection onto set.
\begin{lemma}\label{incprojlem}Let $\mathcal{D}\subset\R^n$ be a nonempty set. For any $t_2\ge t_1\ge 0$
\begin{align*}
\twonorm{\mathcal{P}_{\mathcal{D}}\left(t_1\vct{w}\right)}\le\twonorm{\mathcal{P}_{\mathcal{D}}\left(t_2\vct{w}\right)},
\end{align*} 
holds for all $\vct{w}\in\R^n$.
\end{lemma}
\begin{proof}
By definition of projection onto a set we have
\begin{align}
\label{templemproj}
\twonorm{\mathcal{P}_{\mathcal{D}}\left(t_1\vct{w}\right)-t_1\vct{w}}^2\le\twonorm{\mathcal{P}_{\mathcal{D}}\left(t_2\vct{w}\right)-t_1\vct{w}}^2\quad\text{and}\quad
\twonorm{\mathcal{P}_{\mathcal{D}}\left(t_2\vct{w}\right)-t_2\vct{w}}^2\le\twonorm{\mathcal{P}_{\mathcal{D}}\left(t_1\vct{w}\right)-t_2\vct{w}}^2.
\end{align}
Summing these two identities we conclude that for any $t_2\ge t_1$
\begin{align}
\label{templemproj2}
\langle\vct{w},\mathcal{P}_{\mathcal{D}}(t_2\vct{w})\rangle\ge\langle\vct{w},\mathcal{P}_{\mathcal{D}}(t_1\vct{w})\rangle,
\end{align}
holds for all $\vct{w}\in\R^n$. Now rearranging the first inequality in \eqref{templemproj} and using \eqref{templemproj2} we conclude that for $t_2\ge t_1$
\begin{align*}
\twonorm{\mathcal{P}_{\mathcal{D}}(t_2\vct{w})}^2-\twonorm{\mathcal{P}_{\mathcal{D}}(t_1\vct{w})}^2\ge 2t_1\left(\langle\vct{w},\mathcal{P}_{\mathcal{D}}(t_2\vct{w})\rangle-\langle\vct{w},\mathcal{P}_{\mathcal{D}}(t_1\vct{w})\rangle\right)\ge 0,
\end{align*}
completing the proof.
\end{proof}

\subsection{Proof of the deterministic result (Proof of Theorem \ref{master})}
\label{pffmaster}
To prove Theorem \ref{master} note that if we apply the PGD update
\begin{align*}
\vct{z}_{\tau+1}=\mathcal{P}\left(\vct{z}_\tau+\mu_\tau\mtx{A}^*(\vct{y}-\mtx{A}\vct{z}_\tau)\right),
\end{align*}
then the difference between our iterates and the actual solution $\vct{h}_\tau=\vct{z}_\tau-\vct{x}$ is inside the descent set $\mathcal{D}$. Thus we have the following chain of  inequalities
\begin{align}
\label{interpfthm12}
\tn{\h_{\tau+1}}=\tn{\des_{\tau+1}-\x}&=\tn{\Pc_{\mathcal{K}}(\des_{\tau}+\mu_\tau\A^*(\y-\A\des_\tau))-\x}\nonumber\\
&=\tn{\Pc_{\mathcal{K}-\{\vct{x}\}}\left(\des_{\tau}-\vct{x}+\mu_\tau\A^*(\y-\A\des_\tau)\right)}\nonumber\\
&\overset{(a)}{=}\tn{\Pc_{\mathcal{D}}\left(\vct{z}_{\tau}-\vct{x}+\mu_\tau\A^*(\y-\A\des_\tau)\right)}\nonumber\\
&\overset{(b)}{\leq} \kappa_f\tn{\Pc_\Cc(\des_{\tau}-\vct{x}+\mu_\tau\A^*(\y-\A\des_\tau))}\nonumber\\
&=\kappa_f\tn{\Pc_\Cc(\h_{\tau}+\mu_\tau\A^*(\w-\A\h_\tau))}\nonumber\\
&\leq \kappa_f\sup_{\vb\in\Cc\cap\Bc^{n}} \vb^*[(\Iden_n-\mu_\tau\A^*\A)\h_\tau+\mu_\tau\A^*\w]\nonumber\\
&\leq \kappa_f\big[\sup_{\vb\in\Cc\cap\Bc^{n}} \vb^*(\Iden_n-\mu_\tau\A^*\A)\h_\tau+\mu_\tau\cdot \sup_{\vb\in\Cc\cap\Bc^{n}} \vb^*\A^*\w\big]\nonumber\\
&\leq \kappa_f\big[\rho(\mu_\tau)\tn{\h_\tau}+\mu_\tau\cdot\sup_{\vb\in\Cc\cap\Bc^{n}} \vb^*\A^*\w\big].
\end{align}
In the inequalities above (a) follows from Lemma \ref{simplem} and (b) follows from Lemma \ref{prop compare}. 
Now note that 
\begin{align*}
\mu_\tau\cdot\underset{\vct{v}\in\mathcal{C}\cap\mathcal{B}^n}{\sup}\vct{v}^*\mtx{A}^*\vct{w}=\mu_\tau\cdot\left(\underset{\vct{v}\in\mathcal{C}\cap\mathcal{B}^n}{\sup}\vct{v}^*\mtx{A}^*\frac{\vct{w}}{\twonorm{\vct{w}}}\right)\twonorm{\vct{w}}=\xi_{\mu_\tau}(\mtx{A})\twonorm{\vct{w}}.
\end{align*}
Plugging the latter into \eqref{interpfthm12} we arrive at
\begin{align*}
\twonorm{\vct{h}_{\tau+1}}\le\kappa_f\rho(\mu_\tau)\twonorm{\vct{h}_\tau}+\kappa_f\cdot\xi_{\mu_\tau}(\mtx{A})\twonorm{w}.
\end{align*}
We now apply this inequality recursively with $\mu_\tau=\mu$ which yields
\begin{align*}
\twonorm{\vct{h}_{\tau}}\le&\kappa_f\rho(\mu)\left(\kappa_f\rho(\mu)\twonorm{\vct{h}_{\tau-2}}+\kappa_f\cdot\xi_\mu(\mtx{A})\right)+\kappa_f\cdot\xi_{\mu}(\mtx{A})\twonorm{w}\\
\le&\left(\kappa_f\rho(\mu)\right)^\tau\twonorm{\vct{h}_0}+\big[\left(\kappa_f\rho(\mu)\right)^{\tau-1}+\left(\kappa_f\rho(\mu)\right)^{\tau-2}+\ldots+\left(\kappa_f\rho(\mu)\right)+1\big]\kappa_f\cdot\xi_\mu(\mtx{A})\twonorm{\vct{w}}\\
\le&\left(\kappa_f\rho(\mu)\right)^\tau\twonorm{\vct{h}_0}+\frac{1-\left(\kappa_f\rho(\mu)\right)^\tau}{1-\kappa_f\rho(\mu)}\kappa_f\cdot\xi_\mu(\mtx{A})\twonorm{\vct{w}},
\end{align*}
concluding the proof.

\subsection{Proof of results with Gaussian ensembles (Proof of Theorems \ref{first prop}, \ref{ALTthm}, \ref{thm:converse rate} and \ref{PGthm})}
Our proofs regarding Gaussian ensembles follow from Theorem \ref{master} by controlling the convergence rate $\rho(\mu)$ and the noise amplifications factor $\xi_\mu(\mtx{A})$ with simple algebraic manipulations. The argument for controlling the noise amplification factor $\xi_\mu(\mtx{A})$ for Theorems \ref{first prop}, \ref{ALTthm} and \ref{PGthm} is essentially identical and is stated in Section \ref{cntrlnamp}. The major different in the proof of these theorems however is in the way we control the convergence rate $\rho(\mu)$. This requires more care and different proof strategies are needed for the different choices of the learning parameter $\mu$ in each theorem. We shall  explain how we control the convergence rate $\rho(\mu)$ for each of these theorems in Sections \ref{secprooffirstprop}, \ref{ExactPT}, and \ref{pfstruct}. We shall detail the proof of Theorem \ref{thm:converse rate} which provides a lower bound on the convergence rate in Section \ref{lowbnd}. Before we explain these proofs however, we need to state and prove a few essential results on Gaussian comparisons and the supremum of Gaussian processes which play a crucial role in our proofs.

\subsubsection{Gaussian comparison inequalities}
\begin{lemma}[Slepian's inequality] If $X_t$ and $Y_t$ are a.s. bounded, Gaussian processes on $T$ such that $\E[X_t]=\E[Y_t]$ and $\E[X_t^2]=\E[Y_t^2]$ for all $t\in T$ and
\begin{align*}
\E[(X_t-X_s)^2]\le\E[(Y_t-Y_s)^2],
\end{align*}
for all $s,t\in T$, then for all real $t$,
\begin{align}
\label{firstSlep}
\E[\underset{t\in T}{\sup}\text{ }X_t]\le\E[\underset{t\in T}{\sup}\text{ }Y_t].
\end{align}
Furthermore,
\begin{align}
\label{secondSlep}
\mathbb{P}\big\{\underset{t\in T}{\bigcup}[X_t>\eta_t]\big\}\le\mathbb{P}\big\{\underset{t\in T}{\bigcup}[Y_t>\eta_t]\big\}.
\end{align}
\end{lemma}
\begin{lemma}[Gordon's restricted eigen-values]\cite[Theorem A]{Gor}\label{GTescapelem} Let $\mathcal{C}\in\R^n$ be a cone and let $\mtx{A}$ be and $m\times n$ matrix with entries independently drawn from the standard normal distribution $\mathcal{N}(0,1)$. Also define
\begin{align*}
b_m=\E[\twonorm{\vct{g}}],
\end{align*}
where $\vct{g}\in\R^m$ is distributed as $\mathcal{N}(\vct{0},\mtx{I}_m)$. Then for all $\vct{u}\in\mathcal{C}$
\begin{align}
\label{GordonL}
\frac{\twonorm{\mtx{A}\vct{u}}}{\twonorm{\vct{u}}}\ge b_m -\omega(\mathcal{C}\cap\mathbb{S}^{n-1})-\eta
\end{align}
holds with probability at least $1-e^{-\frac{\eta^2}{2}}$. Furthermore,  for all $\vct{u}\in\mathcal{C}$ we have
\begin{align}
\label{GordonU}
\frac{\twonorm{\mtx{A}\vct{u}}}{\twonorm{\vct{u}}}\le b_m +\omega(\mathcal{C}\cap\mathbb{S}^{n-1})+\eta
\end{align}
holds with probability at least $1-e^{-\frac{\eta^2}{2}}$.
\end{lemma}
We shall now state a generalization of Gordon's escape through the mesh lemma. The proof of this lemma is deferred to Appendix \ref{proofGTtypelem}.
\begin{lemma}\label{GTtypelem} Let $\mathcal{T}\subset\R^n$ and define
\begin{align*}
b_m=\E[\twonorm{\vct{g}}],
\end{align*}
where $\vct{g}\in\R^m$ is distributed as $\mathcal{N}(\vct{0},\mtx{I}_m)$. Furthermore, define
\begin{align*}
\sigma(\mathcal{T}):=\underset{\vct{v}\in\mathcal{T}}{\max}\text{ }\twonorm{\vct{v}}.
\end{align*}
Then 
\begin{align*}
|\twonorm{\mtx{A}\vct{u}}-b_m\twonorm{\vct{u}}|\le \omega(\mathcal{T})+\eta,
\end{align*}
holds for all $\vct{u}\in\mathcal{T}$ with probability at least 
\begin{align*}
1-4e^{-\frac{\eta^2}{8\sigma^2(\mathcal{T})}}.
\end{align*}
\end{lemma}
Finally, we shall make use of the following lemma proven in Appendix \ref{GammaComplemP}.
\begin{lemma}\label{GammaComplem}
Define $\phi(t)=\sqrt{2}\frac{\Gamma(\frac{t+1}{2})}{\Gamma(\frac{t}{2})}\approx \sqrt{t}$. Then for $0\le m_0\le m$ we have
\begin{align*}
\frac{\phi(m_0)}{\sqrt{m_0}}\le \frac{\phi(m)}{\sqrt{m}}.
\end{align*}
\end{lemma}
\subsubsection{Controlling the noise amplification factor $\xi_\mu(\mathbf{A})$}
\label{cntrlnamp}
In this section we shall show how to bound the noise amplification factor. Note that for a fixed vector $\vct{w}$, $\mtx{A}^*\frac{\vct{w}}{\twonorm{\vct{w}}}$ has the same distribution as a random Gaussian vector $\vct{g}\in\R^n$ with i.i.d. $\mathcal{N}(0,1)$ entries. Therefore,
\begin{align*}
\frac{\xi_\mu(\mtx{A})}{\mu}=\sup_{\vct{v}\in\mathcal{C}\cap\mathcal{B}^n}\text{ }\vct{g}^*\vct{v}.
\end{align*}
Applying standard results on concentration of supremum of Gaussian processes
\begin{align*}
\sup_{\vct{v}\in\mathcal{C}\cap\mathbb{S}^{n-1}}\text{ }\vct{g}^*\vct{v}\le \omega\left(\mathcal{C}\cap\mathbb{S}^{n-1}\right)+\eta\le\sqrt{m_0},
\end{align*}
holds with probability at least $1-e^{-\frac{\eta^2}{2}}$.
\subsubsection{Controlling the convergence rate $\rho(\mu)$ with $\mu\approx\frac{1}{m}$ (Proof of Theorem \ref{first prop})}
\label{secprooffirstprop}

In this section we wish to bound the convergence rate $\rho(\mu)$ as stated in Theorem \ref{first prop}. More specifically, let $m_0=\phi^{-1}(\omega+\eta)$ with $\phi(t)=\sqrt{2}\frac{\Gamma(\frac{t+1}{2})}{\Gamma(\frac{t}{2})}\approx \sqrt{t}$ be the minimum number of measurements required by the phase transition curve (please see Definition \ref{PTcurve} for a reminder). Then we will show that as long as
\begin{align}
\label{nummeaslinpf}
m>8\kappa_f^2 m_0,
\end{align}
there is an event of probability at least $1-8e^{-\frac{\eta^2}{8}}$, such that on this event for $\mu=\frac{1}{b_m^2}:=\frac{1}{\left(\phi(m)\right)^2}$
\begin{align}
\label{rhoboundproof}
\rho(\mu)\leq \sqrt{8\frac{m_0}{m}}.
\end{align}
To this aim define the sets
\begin{align*}
\mathcal{T}_{-}=\mathcal{C}\cap\mathbb{S}^{n-1}-\mathcal{C}\cap\mathbb{S}^{n-1}\quad\text{and}\quad\mathcal{T}_{+}=\mathcal{C}\cap\mathbb{S}^{n-1}+\mathcal{C}\cap\mathbb{S}^{n-1},
\end{align*}
and note that
\begin{align*}
\underset{\vct{v}\in\mathcal{T}_{-}}{\sup}\twonorm{\vct{v}}\le 2\quad\text{and}\quad\underset{\vct{v}\in\mathcal{T}_{+}}{\sup}\twonorm{\vct{v}}\le 2.
\end{align*}
Furthermore,
\begin{align*}
\omega\left(\mathcal{T}_{-}\right)=\E\big[\sup_{\vct{u},\vct{v}\in\mathcal{C}\cap\mathcal{B}^n} \vct{g}^T(\vct{u}-\vct{v})\big]\le \E[\sup_{\vct{u}\in\mathcal{C}\cap\mathcal{B}^n} \vct{g}^T\vct{u}+\sup_{\vct{v}\in-\mathcal{C}\cap\mathbb{S}^{n-1}} \vct{g}^T\vct{v}\big]=2\omega(\mathcal{C}\cap\mathbb{S}^{n-1}),
\end{align*}
and similarly, $\omega(\mathcal{T}_{+})\le 2\omega(\mathcal{C}\cap\mathbb{S}^{n-1})$. For ease of presentation we shall use the short hand $\omega:=\omega(\mathcal{C}\cap\mathbb{S}^{n-1})$. Thus, applying Lemma \ref{GTtypelem} on $\mathcal{T}_{-}$ and $\mathcal{T}_{+}$ we conclude that for all $\vct{u},\vct{v}\in\mathcal{C}\cap\mathbb{S}^{n-1}$
\begin{align}
\label{temp22pf1}
\frac{1}{b_m^2}\tn{\A(\ub-\vb)}^2&\leq \left(\tn{\ub-\vb}+2\frac{\left(\omega+\eta\right)}{b_m}\right)^2
\end{align}
and
\begin{align}
\label{temp22pf2}
\frac{1}{b_m^2}\tn{\A(\ub+\vb)}^2&\ge \left(\max\big\{\tn{\ub+\vb}-2\frac{\left(\omega+\eta\right)}{b_m},0\big\}\right)^2
\end{align}
hold with probability at least $1-8e^{-\frac{\eta^2}{8}}$.

If $\twonorm{\vct{u}+\vct{v}}\ge 2\frac{(\omega+\eta)}{b_m}$, then using equations \eqref{temp22pf1} and \eqref{temp22pf2} we can conclude that for all $\vct{u},\vct{v}\in\mathcal{C}\cap\mathbb{S}^{n-1}$
\begin{align}
\label{temp22pf3}
\ub^*(\mtx{I}-\frac{\A^*\A}{b_m^2})\vb&=\frac{1}{4}\bigg[(\ub+\vb)^*(\Iden-\frac{\A^*\A}{b_m^2})(\ub+\vb)-(\ub-\vb)^*(\Iden-\frac{\A^*\A}{b_m^2})(\ub-\vb)\bigg]\nonumber\\
&=\frac{1}{4}\bigg[\twonorm{\vct{u}+\vct{v}}^2-\frac{1}{b_m^2}\twonorm{\mtx{A}(\vct{u}+\vct{v})}^2-\twonorm{\vct{u}-\vct{v}}^2+\frac{1}{b_m^2}\twonorm{\mtx{A}(\vct{u}-\vct{v})}^2\bigg]\nonumber\\
&\leq \frac{(\omega+\eta)}{b_m}(\tn{\ub+\vb}+\tn{\ub-\vb})\nonumber\\
&\le 2\sqrt{2}\frac{(\omega+\eta)}{b_m},
\end{align}
 $\tn{\ub+\vb}+\tn{\ub-\vb}\le 2\sqrt{2}$ which follows from the fact that $\twonorm{\vct{u}+\vct{v}}^2+\twonorm{\vct{u}-\vct{v}}^2=4$.

If $\twonorm{\vct{u}+\vct{v}}< 2\frac{(\omega+\eta)}{b_m}$, then utilizing equations \eqref{temp22pf1} and \eqref{temp22pf2} again we can conclude that for all $\vct{u},\vct{v}\in\mathcal{C}\cap\mathbb{S}^{n-1}$
\begin{align}
\label{temp22pf4}
\ub^*(\mtx{I}-\frac{\A^*\A}{b_m^2})\vb&=\frac{1}{4}\bigg[(\ub+\vb)^*(\Iden-\frac{\A^*\A}{b_m^2})(\ub+\vb)-(\ub-\vb)^*(\Iden-\frac{\A^*\A}{b_m^2})(\ub-\vb)\bigg]\nonumber\\
&=\frac{1}{4}\bigg[\twonorm{\vct{u}+\vct{v}}^2-\frac{1}{b_m^2}\twonorm{\mtx{A}(\vct{u}+\vct{v})}^2-\twonorm{\vct{u}-\vct{v}}^2+\frac{1}{b_m^2}\twonorm{\mtx{A}(\vct{u}-\vct{v})}^2\bigg]\nonumber\\
&\overset{(a)}{\le}\frac{1}{4}\bigg[\twonorm{\vct{u}+\vct{v}}^2-\twonorm{\vct{u}-\vct{v}}^2+\frac{1}{b_m^2}\twonorm{\mtx{A}(\vct{u}-\vct{v})}^2\bigg]\nonumber\\
&\overset{(b)}{\le}\frac{(\omega+\eta)^2}{b_m^2}+\frac{1}{4}\left(\frac{1}{b_m^2}\twonorm{\mtx{A}(\vct{u}-\vct{v})}^2-\twonorm{\vct{u}-\vct{v}}^2\right)\nonumber\\
&\leq 2\frac{(\omega+\eta)^2}{b_m^2}+\frac{(\omega+\eta)}{b_m}\tn{\ub-\vb}\nonumber\\
&\le2\frac{(\omega+\eta)^2}{b_m^2}+2\frac{(\omega+\eta)}{b_m}\nonumber\\
&\overset{(c)}{\le} 2\sqrt{2}\frac{(\omega+\eta)}{b_m},
\end{align}
where (a) follows from equations \eqref{temp22pf1} and \eqref{temp22pf2}, (b) from the fact that $\twonorm{\vct{u}+\vct{v}}< 2\frac{(\omega+\eta)}{b_m}$ and (c) holds as long as $b_m\ge 2\sqrt{2}(\omega+\eta)$. 

Now combining \eqref{temp22pf3} and \eqref{temp22pf3} for $\mu=1/b_m^2$ we have
\begin{align*}
\rho(\mu)=\sup_{\vct{u},\vct{v}\in\mathcal{C}\cap\mathbb{S}^{n-1}}\vct{u}^*\left(\mtx{I}-\frac{1}{b_m^2}\mtx{A}^*\mtx{A}\right)\vct{v}\le \sqrt{8}\frac{(\omega+\eta)}{b_m},
\end{align*}
as long as $\sqrt{8}\frac{(\omega+\eta)}{b_m}\le 1$. The proof of \eqref{rhoboundproof} and the theorem is complete by noting that
\begin{align*}
\frac{\omega+\eta}{b_m}\le \sqrt{\frac{m_0}{m}},
\end{align*}
which follows from Lemma \ref{GammaComplem} since by definition $b_m=\phi(m)$ and $\omega+\eta=\phi(m_0)$.


\subsubsection{Controlling the convergence rate $\rho(\mu)$ with $\mu\approx\frac{1}{m+n}$ (Proof of Theorem \ref{ALTthm})}
\label{ExactPT}
In this section we wish to bound the convergence rate $\rho(\mu)$ for convex functions $f$ as stated in Theorem \ref{ALTthm}. More specifically, let $m_0=\phi^{-1}(\omega+\eta)$ with $\phi(t)=\sqrt{2}\frac{\Gamma(\frac{t+1}{2})}{\Gamma(\frac{t}{2})}\approx \sqrt{t}$ be the minimum number of measurements required by the phase transition curve (please see Definition \ref{PTcurve} for a reminder). Then we will show that as long as
\begin{align}
\label{nummeaslinpf2}
m>m_0,
\end{align}
there is an event of probability at least $1-e^{-\frac{\eta^2}{2}}-e^{-\gamma n}$, such that on this event for $\mu=\frac{0.99}{\left(\sqrt{m}+\sqrt{n}\right)^2}$
\begin{align}
\label{rhoboundproof2}
\rho(\mu)\leq 1-\frac{0.3}{m+n}\left(\sqrt{m}-\sqrt{m_0}\right)^2.
\end{align}
To this aim first note that by standard results on concentration of spectral norm of random matrices\footnote{This also follows from Lemma \ref{GTescapelem} by using $\mathcal{C}=\R^n$ and $\eta=0.005\sqrt{n}$.} for all $\vct{u}\in\R^n$
\begin{align*}
\twonorm{\mtx{A}\vct{u}}^2\le \frac{100}{99}\left(\sqrt{m}+\sqrt{n}\right)^2\twonorm{\vct{u}}^2,
\end{align*}
holds with probability at least $1-e^{-\gamma n}$ with $\gamma$ a fixed numerical constant. The latter statement is equivalent to the matrix $\mtx{I}-\mu\mtx{A}^*\mtx{A}$ being Positive Semi-Definite (PSD) for $\mu=\frac{0.99}{\left(\sqrt{m}+\sqrt{n}\right)^2}$. Now applying the generalized Cauchy Schwarz inequality for the PSD matrix $\mtx{I}-\mu\mtx{A}^*\mtx{A}$ we have
\begin{align}
\label{tempproofPT}
\rho(\mu)=&\sup_{\vct{u},\vct{v}\in\mathcal{C}\cap\mathbb{S}^{n-1}}\vct{u}^*(\mtx{I}-\mu\mtx{A}^*\mtx{A})\vct{v}\nonumber\\
\le&\sup_{\vct{u},\vct{v}\in\mathcal{C}\cap\mathbb{S}^{n-1}}\sqrt{\left(\vct{u}^*(\mtx{I}-\mu\mtx{A}^*\mtx{A})\vct{u}\right)\cdot\left(\vct{v}^*(\mtx{I}-\mu\mtx{A}^*\mtx{A})\vct{v}\right)}\nonumber\\
\le&\sup_{\vct{u}\in\mathcal{C}\cap\mathbb{S}^{n-1}}\vct{u}^*(\mtx{I}-\mu\mtx{A}^*\mtx{A})\vct{u}\nonumber\\
=&1-\mu\left(\inf_{\vct{u}\in\mathcal{C}\cap\mathbb{S}^{n-1}}\twonorm{\mtx{A}\vct{u}}^2\right).
\end{align}
Furthermore, by Lemma \ref{GTescapelem}
\begin{align*}
\frac{1}{b_m}\left(\inf_{\vct{u}\in\mathcal{C}\cap\mathbb{S}^{n-1}}\twonorm{\mtx{A}\vct{u}}\right)\ge 1-\frac{(\omega+\eta)}{b_m},
\end{align*}
holds with probability at least $1-e^{-\frac{\eta^2}{2}}$. Plugging the latter into \eqref{tempproofPT} and using Lemma \ref{GammaComplem} we conclude that for $\mu=\frac{0.99}{\left(\sqrt{m}+\sqrt{n}\right)^2}$
\begin{align*}
\rho(\mu)\le&1-\mu b_m^2\left(1-\frac{(\omega+\eta)}{b_m}\right)^2\\
\le&1-\mu b_m^2\left(1-\sqrt{\frac{m_0}{m}}\right)^2\\
=&1-\frac{0.99b_m^2}{m\left(\sqrt{m}+\sqrt{n}\right)^2}\left(\sqrt{m}-\sqrt{m_0}\right)^2\\
=&1-\frac{\frac{0.99}{2}}{m+n}\cdot\frac{m}{m+1}\cdot\frac{b_m^2}{\frac{m^2}{m+1}}\cdot\frac{2(m+n)}{\left(\sqrt{m}+\sqrt{n}\right)^2}\left(\sqrt{m}-\sqrt{m_0}\right)^2\\
\le &1-\frac{0.3}{m+n}\left(\sqrt{m}-\sqrt{m_0}\right)^2.
\end{align*}
The last inequality follows by using the fact that for $m\ge 2$, $m/(m+1)\ge 2/3$ together with
\begin{align*}
b_m\ge \frac{m}{\sqrt{m+1}}\quad\text{and}\quad 2(m+n)\ge\left(\sqrt{m}+\sqrt{n}\right)^2.
\end{align*}

\subsubsection{Controlling the convergence rate $\rho(\mu)$ with a structure dependent choice of $\mu$ (Proof of Theorem \ref{PGthm})}
\label{pfstruct}
In this section we wish to bound the convergence rate $\rho(\mu)$ for convex functions as stated in Theorem \ref{PGthm}. More specifically, let $m_0=\phi^{-1}(\omega+\eta)$ with $\phi(t)=\sqrt{2}\frac{\Gamma(\frac{t+1}{2})}{\Gamma(\frac{t}{2})}\approx \sqrt{t}$ be the minimum number of measurements required by the phase transition curve (please see Definition \ref{PTcurve} for a reminder). Then we will show that as long as
\begin{align}
\label{PGDpf1}
m>4 m_0,
\end{align}
there is an event of probability at least $1-e^{-\frac{\eta^2}{2}}-4e^{-\frac{\eta^2}{8}}$, such that on this event for 
\begin{align*}
\mu=\frac{m}{b_m^2}\frac{2-\sqrt{2m_0}\sqrt[4]{m}\cdot\left(\sqrt{m}-\sqrt{m_0}\right)^{-\frac{3}{2}}}{\left(m_0-2\sqrt{mm_0}+2m\right)},
\end{align*}
we have
\begin{align}
\label{rhoboundproof}
\rho(\mu)\leq 1-\psi(\frac{m_0}{m}),
\end{align}
where 
\begin{align*}
\psi(\gamma)=2\frac{\left(\sqrt{2}(1-\sqrt{\gamma})^{1.5}-\sqrt{\gamma}\right)^2}{\left(\gamma-2\sqrt{\gamma}+2\right)^2}.
\end{align*}
Throughout this section we shall use the shorthand $\mtx{M}_\mu=\mtx{I}-\mu\mtx{A}^*\mtx{A}$.
First note that we have
\begin{align*}
\vct{u}^*\mtx{M}_\mu\vct{v}=\frac{1}{4}\left(\vct{u}+\vct{v}\right)^*\mtx{M}_\mu\left(\vct{u}+\vct{v}\right)-\frac{1}{4}\left(\vct{u}-\vct{v}\right)^*\mtx{M}_\mu\left(\vct{u}-\vct{v}\right)
\end{align*}
Also note that $\vct{u}+\vct{v}\in\mathcal{C}$. Therefore, by Gordon's lemma (\eqref{GordonL} in Lemma \ref{GTescapelem}) with probability at least $1-e^{-\frac{\eta^2}{2}}$
\begin{align*}
\frac{1}{4}\left(\vct{u}+\vct{v}\right)^*\mtx{M}_\mu\left(\vct{u}+\vct{v}\right)\le \frac{\twonorm{\vct{u}+\vct{v}}^2}{4}\left(1-\mu\left(b_m-\omega-\eta\right)^2\right).
\end{align*}
Now define
\begin{align*}
\mathcal{T}=\Big\{(\vct{u}-\vct{v}):\text{ }\vct{u},\vct{v}\in\mathcal{C}\cap\mathbb{S}^{n-1}\Big\}.
\end{align*}
Note that 
\begin{align*}
\sigma(\mathcal{T})=\underset{\vct{w}\in\mathcal{T}}{\max}\text{ }\twonorm{\vct{w}}=\underset{\vct{u},\vct{v}\in\mathcal{C}\cap\mathbb{S}^{n-1}}{\max}\text{ }\twonorm{\vct{u}-\vct{v}}\le 2.
\end{align*}
Applying Lemma \ref{GTtypelem} for $\vct{w}\in\mathcal{T}$ with probability at least $1-2e^{-\frac{\eta^2}{8}}$ have
\begin{align*}
\twonorm{\mtx{A}\vct{w}}^2\le (\twonorm{\vct{w}}b_m+\omega(\mathcal{T})+2\eta)^2\quad\Rightarrow\quad\twonorm{\mtx{A}(\vct{u}-\vct{v})}^2\le\left(\twonorm{\vct{u}-\vct{v}}b_m+\omega(\mathcal{T})+2\eta\right)^2.
\end{align*}
Note that
\begin{align*}
\omega(\mathcal{T})=\E[\underset{\vct{w}\in\mathcal{T}}{\sup}\text{ }\vct{g}^*\vct{w}]=\E\Big[\underset{\vct{u},\vct{v}\in\mathcal{C}\cap\mathbb{S}^{n-1}}{\sup}\text{ }\vct{g}^*\left(\vct{u}-\vct{v}\right)\Big]\le\E[\underset{\vct{u}\in\mathcal{C}\cap\mathbb{S}^{n-1}}{\sup}\text{ }\vct{g}^*\vct{u}]+\E[\underset{\vct{v}\in\mathcal{C}\cap\mathbb{S}^{n-1}}{\sup}\text{ }\vct{g}^*\vct{v}]=2\omega(\mathcal{C}\cap \mathbb{S}^{n-1}).
\end{align*}
Thus,
\begin{align*}
\twonorm{\mtx{A}(\vct{u}-\vct{v})}^2\le\left(b_m\twonorm{\vct{u}-\vct{v}}+2(\omega+\eta)\right)^2,
\end{align*}
holds with probability at least $1-4e^{-\frac{\eta^2}{8}}$.
Therefore,
\begin{align*}
(\vct{u}-\vct{v})^*\mtx{M}_\mu(\vct{u}-\vct{v})\ge&\twonorm{\vct{u}-\vct{v}}^2-\mu\left(\twonorm{\vct{u}-\vct{v}}b_m+2(\omega+\eta)\right)^2\\
=&\left(1-\mu b_m^2\right)\twonorm{\vct{u}-\vct{v}}^2-4\mu b_m(\omega+\eta)\twonorm{\vct{u}-\vct{v}}-4\mu(\omega+\eta)^2.
\end{align*}
As a result using the fact that $\twonorm{\vct{u}+\vct{v}}^2+\twonorm{\vct{u}-\vct{v}}^2=4$
\begin{align*}
\vct{u}^*\mtx{M}_\mu\vct{v}\le& \frac{\twonorm{\vct{u}+\vct{v}}^2}{4}\left(1-\mu\left(b_m-\omega-\eta\right)^2\right)-\frac{1}{4}(1-\mu b_m^2)\twonorm{\vct{u}-\vct{v}}^2+\mu b_m(\omega+\eta)\twonorm{\vct{u}-\vct{v}}+\mu(\omega+\eta)^2\\
=& \frac{4-\twonorm{\vct{u}-\vct{v}}^2}{4}\left(1-\mu\left(b_m-\omega-\eta\right)^2\right)-\frac{1}{4}(1-\mu b_m^2)\twonorm{\vct{u}-\vct{v}}^2+\mu b_m(\omega+\eta)\twonorm{\vct{u}-\vct{v}}+\mu(\omega+\eta)^2\\
=&-\frac{\twonorm{\vct{u}-\vct{v}}^2}{4}\left(2-\mu\left(b_m-\omega-\eta\right)^2-\mu b_m^2\right)+\mu b_m(\omega+\eta)\twonorm{\vct{u}-\vct{v}}+\mu(\omega+\eta)^2+\left(1-\mu\left(b_m-\omega-\eta\right)^2\right)\\
:=&-\alpha\twonorm{\vct{u}-\vct{v}}^2+\beta\twonorm{\vct{u}-\vct{v}}+\gamma\\
\le&\frac{\beta^2}{4\alpha}+\gamma\\
=&\frac{\mu^2b_m^2(\omega+\eta)^2}{\left(2-\mu\left(b_m-\omega-\eta\right)^2-\mu b_m^2\right)}+\mu(\omega+\eta)^2+\left(1-\mu\left(b_m-\omega-\eta\right)^2\right)\\
=&\frac{\mu^2b_m^2(\omega+\eta)^2}{\left(2-\mu[(b_m-\omega-\eta)^2+b_m^2]\right)}+\mu[(\omega+\eta)^2-\left(b_m-\omega-\eta\right)^2]+1.
\end{align*}
We note that there is an additional constraint which we need to impose here which is
\begin{align*}
\frac{\beta}{2\alpha}\le2\quad\Leftrightarrow\quad\frac{\mu b_m(\omega+\eta)}{\left(2-\mu\left(b_m-\omega-\eta\right)^2-\mu b_m^2\right)}\le 1.
\end{align*}
Using $\nu:=\frac{\omega+\eta}{b_m}$ we have
\begin{align}
\label{mutempbndt}
\vct{u}^*\mtx{M}_\mu\vct{v}\le\frac{\mu^2 b_m^4 \nu^2}{2-\mu b_m^2\left(\nu^2-2\nu+2\right)}+\mu b_m^2\left(2\nu-1\right)+1:=\tilde{g}(\mu b_m^2,\nu),
\end{align}
which holds as long as
\begin{align*}
\frac{\mu b_m^2 \nu}{2-\mu b_m^2\left(\nu^2-2\nu+2\right)}\le 1\quad\Leftrightarrow\quad \mu b_m^2\le \frac{2}{(\nu^2-\nu+2)}.
\end{align*}
Now note that
\begin{align*}
\frac{\partial }{\partial \nu}g(x,\nu)=2x^2\nu\frac{\left(2-x\left(\nu^2-2\nu+2\right)\right)+x\nu(\nu-1)}{\left(2-x\left(\nu^2-2\nu+2\right)\right)^2}+2x.
\end{align*}
The latter is positive when $x\ge 0$, $0\le\nu\le 1$, and $x\le\frac{2}{\nu^2-2\nu+2}$. This in turn implies that $g(\mu b_m^2,\nu)$ in $\nu$. Now define $s:=\sqrt{\frac{m_0}{m}}$ and note that by Lemma \ref{GammaComplem}, $\nu:=\frac{\omega+\eta}{b_m}\le \sqrt{\frac{m_0}{m}}:=s$. Since $g$ is increasing in its second argument we have $g(\mu b_m^2,\nu)\le g(\mu b_m^2,s)$. Using this in \eqref{mutempbnd} we arrive at
\begin{align}
\label{mutempbnd}
\vct{u}^*\mtx{M}_\mu\vct{v}\le g(\mu b_m^2,s),
\end{align}
which holds as long as
\begin{align*}
\mu b_m^2\le \frac{2}{(s^2-s+2)}.
\end{align*}
We again remind the reader that $s$ was defined as $s:=\sqrt{\frac{m_0}{m}}$. Further, note that
\begin{align*}
\frac{\partial }{\partial x}g(x,s)=\frac{\left(s^2-2s+2\right)s^2x^2}{\left(2-\left(s^2-2s+2\right)x\right)^2}+\frac{2s^2x}{2-\left(s^2-2s+2\right)x}+2s-1,
\end{align*}
and
\begin{align*}
\frac{\partial^2 }{\partial x^2}g(x,s)=-\frac{8s^2}{\left(\left(s^2-2s+2\right)x-2\right)^3}.
\end{align*}
Thus
\begin{align*}
\frac{\partial }{\partial x}g(x,s)=0\quad\Rightarrow\quad x=\frac{2\pm \sqrt{2}s(1-s)^{-\frac{3}{2}}}{\left(s^2-2s+2\right)}.
\end{align*}
Also note that at $x=\frac{2+ \sqrt{2}s(1-s)^{-\frac{3}{2}}}{\left(s^2-2s+2\right)}$, $\frac{\partial^2 }{\partial x^2}g(x,s)<0$ and at $x=\frac{2- \sqrt{2}s(1-s)^{-\frac{3}{2}}}{\left(s^2-2s+2\right)}$, $\frac{\partial^2 }{\partial x^2}g(x,s)>0$. Thus $g(\mu b_m^2,s)$ as a function of $s$, is minimized at either of the following two points
\begin{align*}
\mu b_m^2=\frac{2- \sqrt{2}s(1-s)^{-\frac{3}{2}}}{\left(s^2-2s+2\right)}\quad\text{or}\quad \mu b_m^2=\frac{2}{s^2-s+2}.
\end{align*}
The value of $g(\mu b_m^2,s)$ at these two points is given by
\begin{align*}
g_1(s)=&g\left(\frac{2-\sqrt{2}s(1-s)^{-\frac{3}{2}}}{\left(s^2-2s+2\right)},s\right)=1-2\frac{\left(\sqrt{2}(1-s)^{1.5}-s\right)^2}{\left(s^2-2s+2\right)^2},\\
g_2(s)=&g\left(\frac{2}{s^2-s+2},s\right)=\frac{s(s+5)}{s^2-s+2}.
\end{align*}

It is easy to check that for $0\le s\le 1$ we have $g_1(s)\le g_2(s)$. We note that $\frac{2-\sqrt{2}s(1-s)^{-\frac{3}{2}}}{\left(s^2-2s+2\right)}$ is non-negative for $s\le0.5$. Now using
 \begin{align*}
\mu=&\frac{m}{b_m^2}\frac{2-\sqrt{2m_0}\sqrt[4]{m}\cdot\left(\sqrt{m}-\sqrt{m_0}\right)^{-\frac{3}{2}}}{\left(m_0-2\sqrt{mm_0}+2m\right)}\quad\Leftrightarrow\quad\mu b_m^2=\frac{2-\sqrt{2}s(1-s)^{-\frac{3}{2}}}{\left(s^2-2s+2\right)},
\end{align*}
in \eqref{mutempbnd} we have
\begin{align*}
\vct{u}^*\mtx{M}_\mu\vct{v}\le1-2\frac{\left(\sqrt{2}(1-s)^{1.5}-s\right)^2}{\left(s^2-2s+2\right)^2}=1-2\frac{\left(\sqrt{2}(1-\sqrt{\frac{m_0}{m}})^{1.5}-\sqrt{\frac{m_0}{m}}\right)^2}{\left(\frac{m_0}{m}-2\sqrt{\frac{m_0}{m}}+2\right)^2}:=1-\psi(\frac{m_0}{m}).
\end{align*}
Finally, note that for $m\ge4m_0$ we have $\psi(\frac{m_0}{m})\le1$ and $\mu\ge 0$, concluding the proof.


\subsubsection{Lower bounds on convergence rate (Proof of Theorem \ref{thm:converse rate})}
\label{lowbnd}
To obtain lower bounds on the convergence rate, we make use of \cite[Lemma F.1]{OymProx} stated below, which relates the projection onto a set and the projection onto its conic approximation around the origin. 
\begin{lemma}\label{cone cool} Let $\mathcal{D}\in\R^n$ be a closed and convex set containing $\vct{0}$ and let $\mathcal{C}=$cone$(\mathcal{D})$ be its conic hull. Given any $0<\alpha,\epsilon\le 1$, there exists a positive constant $\delta:=\delta(\alpha,\epsilon,\mathcal{D})$ such that, for all vectors $\vct{w}$ obeying 
\begin{align}
\label{condlem}
\twonorm{\Pc_{\mathcal{C}}(\w)}\geq \alpha\twonorm{\vct{w}},
\end{align}
then
\begin{align*}
\label{main state}
\frac{\twonorm{\Pc_{\mathcal{D}}(\w)}}{\twonorm{\Pc_{\mathcal{C}}(\w)}}\geq 1-\eps,\end{align*}
holds for all $\twonorm{\vct{w}}\le\delta$.
\end{lemma}
This lemma suggests that for any set containing the origin, around a sufficiently small neighborhood of $\vct{0}$, the conic hull is an arbitrarily good approximation of the original set. We shall apply this lemma with $\mathcal{D}$ and $\mathcal{C}$ equal to the descent set and tangent cone of $f$ at $\vct{x}$. We remind the reader that from the third line of \eqref{interpfthm12} we know that the error $\vct{h}_\tau=\vct{z}_\tau-\vct{x}$ in our iterates obeys the update
\beq
\vct{h}_{\tau+1}=\mathcal{P}_{\mathcal{D}}\left((\mtx{I}-\mu\mtx{A}^*\mtx{A})\vct{h}_\tau\right).\label{iterh}
\eeq
Now note that we have
\begin{align}
\label{mytemp1}
\twonorm{(\mtx{I}-\mu\mtx{A}^*\mtx{A})\vct{h}_\tau}\le\opnorm{\mtx{I}-\mu\mtx{A}^*\mtx{A}}\twonorm{\vct{h}_\tau}\le\left(1+\mu\sigma_{\max}^2(\mtx{A})\right)\twonorm{\vct{h}_\tau}.
\end{align}
Applying Lemma \ref{GTescapelem} equation \eqref{GordonU} together with Lemma \ref{GammaComplem} and \eqref{mytemp1} we also have
\begin{align}
\label{mytemp2}
\twonorm{\mathcal{P}_{\mathcal{C}}\left((\mtx{I}-\mu\mtx{A}^*\mtx{A})\vct{h}_\tau\right)}=&\sup_{\vct{v}\in\mathcal{C}\cap\mathbb{S}^{n-1}}\vct{v}^*\left(\mtx{I}-\mu\mtx{A}^*\mtx{A}\right)\vct{h}_\tau\nonumber\\
\ge&\frac{1}{\twonorm{\vct{h}_\tau}}\vct{h}_\tau^*\left(\mtx{I}-\mu\mtx{A}^*\mtx{A}\right)\vct{h}_\tau\nonumber\\
=&\twonorm{\vct{h}_\tau}\left(1-\mu\frac{\twonorm{\mtx{A}\vct{h}_\tau}^2}{\twonorm{\vct{h}_\tau}^2}\right)\nonumber\\
\ge&\twonorm{\vct{h}_\tau}\left(1-\mu(b_m+\omega+\eta)^2\right)\nonumber\\
\ge&\twonorm{\vct{h}_\tau}\left(1-\mu\left(\sqrt{m}+\sqrt{m_0}\right)^2\right)\nonumber\\
\ge&\left(\frac{1-\mu\left(\sqrt{m}+\sqrt{m_0}\right)^2}{1+\mu\sigma_{\max}^2(\mtx{A})}\right)\twonorm{(\mtx{I}-\mu\mtx{A}^*\mtx{A})\vct{h}_\tau}.
\end{align}
As a result based on \eqref{mytemp1} and \eqref{mytemp2} the requirement \eqref{condlem} of Lemma \ref{cone cool} is satisfied for the vector $\vct{w}=(\mtx{I}-\mu\mtx{A}^*\mtx{A})\vct{h}_\tau$ with $\alpha=\left(\frac{1-\mu\left(\sqrt{m}+\sqrt{m_0}\right)^2}{1+\mu\sigma_{\max}^2(\mtx{A})}\right)$. Thus by Lemma \ref{cone cool} there exists a constant $\delta$ such that for all $\twonorm{(\mtx{I}-\mu\mtx{A}^*\mtx{A})\vct{h}_\tau}\le\delta$ 
\begin{align*}
\twonorm{\mathcal{P}_{\mathcal{D}}\left((\mtx{I}-\mu\mtx{A}^*\mtx{A})\vct{h}_\tau\right)}\ge(1-\epsilon_\tau)\twonorm{\mathcal{P}_{\mathcal{C}}\left((\mtx{I}-\mu\mtx{A}^*\mtx{A})\vct{h}_\tau\right)}
\end{align*}
holds with probability at least $1-e^{-\frac{\eta^2}{2}}$. Thus using $r=\frac{\delta}{1+\mu\sigma_{\max}^2(\mtx{A})}$ we can conclude that for all $\twonorm{\vct{h}_\tau}\le r$
\begin{align*}
\twonorm{\vct{h}_{\tau+1}}=\twonorm{\mathcal{P}_{\mathcal{D}}\left((\mtx{I}-\mu\mtx{A}^*\mtx{A})\vct{h}_\tau\right)}\ge&(1-\epsilon)\twonorm{\mathcal{P}_{\mathcal{C}}\left((\mtx{I}-\mu\mtx{A}^*\mtx{A})\vct{h}_\tau\right)}\\
\ge&(1-\epsilon)\left(1-\mu\left(\sqrt{m}+\sqrt{m_0}\right)^2\right)\twonorm{\vct{h}_\tau},
\end{align*}
holds with probability at least $1-e^{-\frac{\eta^2}{2}}$. By inductively applying the latter inequality for all starting points $\twonorm{\h_0}\le r$ we will show that for all $t\geq 0$
\beq
\tn{\h_t}\geq \left(1-\eps\right)^t\left(1-\mu\left(\sqrt{m}+\sqrt{m_0}\right)^2\right)^t\tn{\h_0}.\nn
\eeq
Suppose the proposed equation holds for $t=1,2,\ldots,\tau$. Define
\begin{align*}
\tilde{\vct{h}}_\tau=(1-\eps)^\tau\left(1-\mu\left(\sqrt{m}+\sqrt{m_0}\right)^2\right)^\tau\tn{\h_0}\frac{\h_\tau}{\tn{\h_\tau}}
\end{align*}
Using $\tn{\tilde{\h}_\tau}\leq \tn{\h_\tau}$, and applying Lemma \ref{incprojlem} for $t=\tau+1$ we have that
\begin{align*}
\tn{\h_{\tau+1}}&=\twonorm{\mathcal{P}_{\mathcal{D}}\left((\mtx{I}-\mu\mtx{A}^*\mtx{A})\h_\tau\right)},\\
&\geq \twonorm{\mathcal{P}_{\mathcal{D}}\left((\mtx{I}-\mu\mtx{A}^*\mtx{A})\tilde{\h}_\tau\right)},\\
&\geq (1-\eps)\left(1-\mu\left(\sqrt{m}+\sqrt{m_0}\right)^2\right)\tn{\tilde{\h}_\tau},\\
&\geq (1-\eps)^{\tau+1}\left(1-\mu\left(\sqrt{m}+\sqrt{m_0}\right)^2\right)^{\tau+1}\tn{\h_0},
\end{align*}
concluding the proof for all $t$ by induction.

\subsection{Sensitivity to the tuning parameter (Proof of Theorem \ref{sensitivitythm})}\label{app:sensitivity}
In all of our proofs so far we have assumed that the parameter $R$ is tuned perfectly and is set to $R=f(\vct{x})$. This is of course problematic because we do not know the signal $\x$, and therefore also don not know $f(\x)$ in advance. Crucial to our proofs so far has been the fact that different between our iterates $\vct{z}_\tau$ and the unknown signal $\vct{x}$ lie in the descent set $\mathcal{D}$. However, this is no longer true when $R\neq f(\vct{x})$. As a result it is important to ensure that our results are robust to perturbations of the feasible set $\Dc$. We will next illustrate how one can get bounds that are robust to the change of $\Dc$.
\begin{itemize}
\item To model $R> f(\x)$, define $\Dc_{sup}^R=\{\w\big|f(\x+\w)\leq R\}$.
\item To model $R< f(\x)$, define $\Dc_{sub}^R=\{\w\big|f(\x+\w)\leq R\}$.
\end{itemize}
Similarly, we define $\mathcal{K}_{sup}^R=\{\z\big|f(\z)\leq R\}$ and $\mathcal{K}_{sub}^R=\{\z\big|f(\z)\leq R\}$ for $R> f(\x)$ and $R< f(\x)$ and note that $\Dc_{sup}^R=\mathcal{K}_{sup}^R-\{\vct{x}\}$ and $\Dc_{sub}^R=\mathcal{K}_{sub}^R-\{\vct{x}\}$.

Before we begin our arguments we state a lemma which we utilize throughout the proof of this theorem. The proof of this theorem is essentially identical to bounding the convergence rate in the proof of Theorem \ref{first prop} in Section \ref{secprooffirstprop}. We skip the details.
\begin{lemma}Let $\vct{x}\in\R^n$ be an arbitrary vector in $\R^n$ and $f:\R^n\rightarrow\R$ be a proper function. Suppose $\mtx{A}\in\R^{m\times n}$ is a Gaussian map. Then
\begin{align*}
\underset{\vct{u},\vct{v}\in\mathcal{C}_f(\vct{x})\cap \mathbb{S}^{n-1}}{\sup}\text{ }-\vct{u}^*\left(\mtx{I}-\mu\mtx{A}^*\mtx{A}\right)\vct{v}\le\sqrt{8\frac{m_0}{m}},
\end{align*}
holds with probability at least $1-4e^{-\frac{\eta^2}{8}}$.
\end{lemma}
As a consequence of this lemma the upper bound on $\underset{\vct{u},\vct{v}\in\mathcal{C}_f(\vct{x})\cap \mathbb{S}^{n-1}}{\sup}\text{ }-\vct{u}^*\left(\mtx{I}-\mu\mtx{A}^*\mtx{A}\right)\vct{v}$ is the same as the upper bound on $\rho(\mu)$. So we shall assume throughout the proof of this theorem that
\begin{align*}
\underset{\vct{u},\vct{v}\in\mathcal{C}_f(\vct{x})\cap \mathbb{S}^{n-1}}{\sup}\text{ }-\vct{u}^*\left(\mtx{I}-\mu\mtx{A}^*\mtx{A}\right)\vct{v}\le\rho(\mu).
\end{align*}
\subsubsection{Under estimating the tunning parameter ($R< f(x)$)}
We first consider the case where $R< f(\x)$. In this case the difference between our iterates and the signal is in $\Dc_{sub}^R$, i.e.~$(\vct{z}_\tau-\vct{x})\in\Dc_{sub}^R$. It is important to note that in this case $0\not\in \Dc_{sub}^R$ since $f(\x)>R$. To handle this case it suffices to develop an analogue of Theorem \ref{master}, the rest of the proof follows along similar lines to the proof of Theorem \ref{first prop}. Please see Section \ref{secprooffirstprop} for further details.

\begin{lemma}Suppose, $f(\cdot)$ is a proper convex function and $\vct{y}=\mtx{A}\vct{x}$. Define $\mathcal{K}_{sub}^R=\{\z\big|f(\z)\leq R\}$. Consider the iterations
\begin{align*}
\z_{\tau+1}=\Pc_{\mathcal{K}_{sub}^R}(\z_\tau+\mu\A^*(\y-\A\z_\tau))
\end{align*}
The error vector $\vct{z}_\tau-\vct{x}$ satisfies the following recursion
\begin{align*}
\tn{\z_{\tau+1}-\x}\leq \kappa_f\rho(\mu)\tn{\z_\tau-\x}+(3-\kappa_f)\twonorm{\vct{x}-\mathcal{P}_{\mathcal{K}_{sub}^R}(\vct{x})}.
\end{align*}
\end{lemma}

\begin{proof} First note that by Lemma \ref{simplem} we have
\begin{align*}
\vct{z}_{\tau+1}-\x=&\Pc_{\mathcal{K}_{sub}^R-\{\x\}}(\z_\tau-x+\mu\A^*(\y-\A\z_\tau))\\
=&\Pc_{\Dc_{sub}^R}\left((\mtx{I}-\mu\mtx{A}^*\mtx{A})(\vct{z}_\tau-\vct{x})\right).
\end{align*}
Now defining $\vct{h}_\tau=\z_\tau-\x$ we arrive at
\begin{align}
\label{temp59}
\vct{h}_{\tau+1}=\Pc_{\Dc_{sub}^R}(\vct{h}_\tau-\mu\A^*\A\vct{h}_\tau).
\end{align}
We remind the reader that we use $\mathcal{C}:=\mathcal{C}_f(\vct{x})$ to be the tangent cone of $f$ at the point $\vct{x}$. In our calculations below it is convenient to use the short-hand notations $\tilde{\vct{h}}_\tau:=\vct{h}_\tau-\mu\mtx{A}^*\mtx{A}\vct{h}_\tau$ and $\vct{w}:=\Pc_{\Dc_{sub}^R}(\vct{0})=\vct{x}-\mathcal{P}_{\mathcal{K}_{sub}^R}(\vct{x})$. With this notation \eqref{temp59} can be rewritten as 
\begin{align*}
\vct{h}_{\tau+1}=\Pc_{\Dc_{sub}^R}(\tilde{\vct{h}}_\tau).
\end{align*}
We consider two cases. In the first case we assume the function $f$ is convex and $\kappa_f=1$. In the second case we do not assume convexity of $f$ and we prove the inequality with $\kappa_f=2$. 

\text{ }\\
\noindent\textbf{Proof for convex $f$}

When the function $f$ is convex which immediately implies that the sets $\mathcal{K}_{sub}^R$ and $\mathcal{D}_{sub}^R$ are convex. Noting that $\mathcal{D}_{sub}^R\subset \mathcal{C}$ we have
\begin{align}
\label{temp591}
\twonorm{\tilde{\vct{h}}_\tau-\Pc_{\mathcal{C}}(\tilde{\vct{h}}_\tau)}\le\twonorm{\tilde{\vct{h}}_\tau-\Pc_{\Dc_{sub}^R}(\tilde{\vct{h}}_\tau)}.
\end{align}
Also by Lemma \ref{projcvxthm} equation \eqref{projcvxthm1}
\begin{align}
\label{temp592}
\twonorm{\vct{w}-\Pc_{\Dc_{sub}^R}(\tilde{\vct{h}}_\tau)}^2+\twonorm{\Pc_{\Dc_{sub}^R}(\tilde{\vct{h}}_\tau)-\tilde{\vct{h}}_\tau}^2\le\twonorm{\vct{w}-\tilde{\vct{h}}_\tau}^2.
\end{align}
Furthermore, using the fact that $\vct{w}:=\vct{x}-\mathcal{P}_{\mathcal{K}_{sub}^R}(\vct{x})\in\mathcal{C}$ we have
\begin{align}
\label{temp593}
\twonorm{\vct{w}-\tilde{\vct{h}}_\tau}^2\le& \twonorm{\tilde{\vct{h}}_\tau}^2+\twonorm{\vct{w}}^2-2\vct{w}^*\tilde{\vct{h}}_\tau\nonumber\\
\le&\twonorm{\tilde{\vct{h}}_\tau}^2+\twonorm{\vct{w}}^2-2\vct{w}^*\left(\mtx{I}-\mu\mtx{A}^*\mtx{A}\right)\vct{h}_\tau\nonumber\\
\le&\twonorm{\tilde{\vct{h}}_\tau}^2+\twonorm{\vct{w}}^2+2\twonorm{\vct{w}}\twonorm{\vct{h}_\tau}\cdot\text{ }\sup_{\vct{u},\vct{v}\in\mathcal{C}\cap\mathbb{S}^{n-1}}\text{ }-\vct{u}^*\left(\mtx{I}-\mu\mtx{A}^*\mtx{A}\right)\vct{v}\nonumber\\
\le&\twonorm{\tilde{\vct{h}}_\tau}^2+\twonorm{\vct{w}}^2+2\rho(\mu)\twonorm{\vct{w}}\twonorm{\vct{h}_\tau}
\end{align}
Finally, we also note that
\begin{align}
\label{temp594}
\twonorm{\Pc_{\mathcal{C}}(\tilde{\vct{h}}_\tau)}=&\twonorm{\Pc_{\mathcal{C}}\left((\mtx{I}-\mu\mtx{A}^*\mtx{A})\vct{h}_\tau\right)}\nonumber\\
\le&\left(\sup_{\vct{u},\vct{v}\in\mathcal{C}\cap\mathbb{S}^{n-1}}\text{ }\vct{u}^*\left(\mtx{I}-\mu\mtx{A}^*\mtx{A}\right)\vct{v}\right)\cdot\twonorm{\vct{h}_\tau}\nonumber\\
\le& \rho(\mu)\twonorm{\vct{h}_\tau}.
\end{align}
Using \eqref{temp592} and \eqref{temp593} we have
\begin{align*}
\twonorm{\vct{w}-\Pc_{\Dc_{sub}^R}(\tilde{\vct{h}}_\tau)}^2+\twonorm{\Pc_{\Dc_{sub}^R}(\tilde{\vct{h}}_\tau)-\tilde{\vct{h}}_\tau}^2\le\twonorm{\tilde{\vct{h}}_\tau}^2+\twonorm{\vct{w}}^2+2\rho(\mu)\twonorm{\vct{w}}\twonorm{\vct{h}_\tau}.
\end{align*}
Combining this with \eqref{temp591}, applying Lemma \ref{firstconelem} equation \eqref{orthogonal}, and using \eqref{temp594} we conclude that
\begin{align*}
\twonorm{\vct{w}-\Pc_{\Dc_{sub}^R}(\tilde{\vct{h}}_\tau)}^2+\twonorm{\Pc_{\mathcal{C}}(\tilde{\vct{h}}_\tau)-\tilde{\vct{h}}_\tau}^2\le& \twonorm{\tilde{\vct{h}}_\tau}^2+\twonorm{\vct{w}}^2+2\rho(\mu)\twonorm{\vct{w}}\twonorm{\vct{h}_\tau}\\
=&\twonorm{\mathcal{P}_{\mathcal{C}}(\tilde{\vct{h}}_\tau)-\tilde{\vct{h}}_\tau}^2+\twonorm{\mathcal{P}_{\mathcal{C}}\left(\tilde{\vct{h}}_\tau\right)}^2+\twonorm{\vct{w}}^2+2\rho(\mu)\twonorm{\vct{w}}\twonorm{\vct{h}_\tau}\\
\le&\twonorm{\mathcal{P}_{\mathcal{C}}(\tilde{\vct{h}}_\tau)-\tilde{\vct{h}}_\tau}^2+\left(\rho(\mu)\right)^2\twonorm{\vct{h}_\tau}^2+\twonorm{\vct{w}}^2+2\rho(\mu)\twonorm{\vct{w}}\twonorm{\vct{h}_\tau}.
\end{align*}
Simplifying the above expression we conclude that
\begin{align*}
\twonorm{\vct{w}-\Pc_{\Dc_{sub}^R}(\tilde{\vct{h}}_\tau)}\le\rho(\mu)\twonorm{\vct{h}_\tau}+\twonorm{\vct{w}}.
\end{align*}
Finally, applying the triangular inequality we arrive at
\begin{align*}
\twonorm{\vct{h}_{\tau+1}}=\twonorm{\Pc_{\Dc_{sub}^R}(\tilde{\vct{h}}_\tau)}\le\twonorm{\vct{w}-\Pc_{\Dc_{sub}^R}(\tilde{\vct{h}}_\tau)}+\twonorm{\vct{w}}\le \rho(\mu)\twonorm{\vct{h}_\tau}+2\twonorm{\vct{w}},
\end{align*}
concluding the proof for convex functions $f$.

\text{ }\\
\noindent\textbf{Proof for general $f$}

Note that since $\vct{w}\in\mathcal{D}_{sub}^R$ by definition of projection onto $\mathcal{D}_{sub}^R$ we have
\begin{align*}
\twonorm{\tilde{\vct{h}}_\tau-\mathcal{P}_{\mathcal{D}_{sub}^R}(\tilde{\vct{h}}_\tau)}^2\le\twonorm{\tilde{\vct{h}}_\tau-\vct{w}}^2\quad\Rightarrow\twonorm{\mathcal{P}_{\mathcal{D}_{sub}^R}(\tilde{\vct{h}}_\tau)}^2\le&2\tilde{\vct{h}}_\tau^*\left(\mathcal{P}_{\mathcal{D}_{sub}^R}(\tilde{\vct{h}}_\tau)\right)-2\tilde{\vct{h}}_\tau^*\vct{w}+\twonorm{\vct{w}}^2.
\end{align*}
Noting that the points $\mathcal{P}_{\mathcal{D}_{sub}^R}(\tilde{\vct{h}}_\tau)$ and $\vct{w}$ belong to $\mathcal{C}_f(\vct{x})$ we have 
\begin{align*}
\twonorm{\mathcal{P}_{\mathcal{D}_{sub}^R}(\tilde{\vct{h}}_\tau)}^2\le&2\tilde{\vct{h}}_\tau^*\left(\mathcal{P}_{\mathcal{D}_{sub}^R}(\tilde{\vct{h}}_\tau)\right)-2\tilde{\vct{h}}_\tau^*\vct{w}+\twonorm{\vct{w}}^2\\
\le&2\vct{h}_\tau^*\left(\mtx{I}-\mu\mtx{A}^*\mtx{A}\right)\mathcal{P}_{\mathcal{D}_{sub}^R}(\tilde{\vct{h}}_\tau)-2\vct{h}_\tau^*\left(\mtx{I}-\mu\mtx{A}^*\mtx{A}\right)\vct{w}+\twonorm{\vct{w}}^2\\
\le&2\rho(\mu)\twonorm{\vct{h}_\tau}\twonorm{\mathcal{P}_{\mathcal{D}_{sub}^R}(\tilde{\vct{h}}_\tau)}+2\rho(\mu)\twonorm{\vct{h}_\tau}\twonorm{\vct{w}}+\twonorm{\vct{w}}^2.
\end{align*}

Completing the square we arrive at
\begin{align*}
\left(\twonorm{\mathcal{P}_{\mathcal{D}_{sub}^R}(\tilde{\vct{h}}_\tau)}-\rho(\mu)\twonorm{\vct{h}_\tau}\right)^2\le\left(\rho(\mu)\twonorm{\vct{h}_\tau}+\twonorm{\vct{w}}\right)^2,
\end{align*}
which implies that
\begin{align*}
\twonorm{\mathcal{P}_{\mathcal{D}_{sub}^R}(\tilde{\vct{h}}_\tau)}\le 2\rho(\mu)\twonorm{\vct{h}_\tau}+\twonorm{\vct{w}}.
\end{align*}
\end{proof}

\subsubsection{Over estimating the tunning parameter ($R> f(x)$)}
We now consider the case where $R> f(\x)$. In this case the difference between our iterates and the signal is in $\Dc_{sup}^R$, i.e.~$(\vct{z}_\tau-\vct{x})\in\Dc_{sup}^R$. It is important to note that in this case $0\in \Dc_{sup}^R$ since $f(\x)<R$. To handle this case it suffices to develop an analogue of Theorem \ref{master}, the rest of the proof follows along similar lines to the proof of Theorem \ref{first prop}. Please see Section \ref{secprooffirstprop} for further details.
%
\begin{lemma}Suppose, $f(\cdot)$ is a homogeneous function and $\vct{y}=\mtx{A}\vct{x}$. Define $\mathcal{K}_{sup}^R=\{\z\big|f(\z)\leq R\}$. Consider the iterations
\begin{align*}
\z_{\tau+1}=\Pc_{\mathcal{K}_{sup}^R}(\z_\tau+\mu\A^*(\y-\A\z_\tau))
\end{align*}
The error vector $\vct{z}_\tau-\vct{x}$ satisfies the following recursion
\begin{align*}
\tn{\z_{\tau+1}-\x}\leq \kappa_f\tilde{\rho}(\mu)\tn{\z_\tau-\x}+2\kappa_f\left(1+\tilde{\rho}(\mu)\right)\left(\frac{R}{f(\vct{x})}-1\right)\twonorm{\vct{x}},
\end{align*}
where the convergence rate $\tilde{\rho}(\mu)$ is defined as the restricted singular value of the descent cone of $f$ at $\frac{R}{f(\vct{x})}\vct{x}$. That is,
\begin{align*}
\tilde{\rho}(\mu)=\sup_{\vct{u},\vct{v}\in\mathcal{C}_f\left(\frac{R}{f(\vct{x})}\vct{x}\right)\cap\mathbb{S}^{n-1}}\text{ }\vct{u}^*\left(\mtx{I}-\mu\mtx{A}^*\mtx{A}\right)\vct{v}.
\end{align*}
Furthermore, if $f$ is a norm then $\tilde{\rho}(\mu)=\rho(\mu)$.
\end{lemma}
\begin{proof} Define $\tilde{\vct{x}}=\frac{R}{f(\x)}\x$ and note that since $f$ is a homogenous function $f(\tilde{\vct{x}})=\frac{R}{f(\vct{x})}f(\vct{x})=R$. Consequently, we may imagine that we are trying to estimate $\tilde{\vct{x}}$ where $\w=\A\x-\A\tilde{\vct{x}}$ is the additive noise on our measurements $\A\tilde{\vct{x}}$. Applying \eqref{interpfthm12} in the proof of Theorem \ref{master} in Section \ref{pffmaster} with $\mathcal{C}$ set to $\mathcal{C}_f(\tilde{\vct{x}})$ we can conclude that
\begin{align*}
\twonorm{\vct{z}_{\tau+1}-\tilde{\vct{x}}}\le&\kappa_f\tilde{\rho}(\mu)\twonorm{\vct{z}_\tau-\tilde{\vct{x}}}+\kappa_f\cdot\mu\cdot\sup_{\vct{u}\in\mathcal{C}_f(\tilde{\vct{x}})\cap\mathbb{S}^{n-1}}\text{ }\vct{u}^*\mtx{A}^*\vct{w}\\
=&\kappa_f\tilde{\rho}(\mu)\twonorm{\vct{z}_\tau-\tilde{\vct{x}}}+\kappa_f\cdot\mu\cdot\sup_{\vct{u}\in\mathcal{C}_f(\tilde{\vct{x}})\cap\mathbb{S}^{n-1}}\text{ }\vct{u}^*\mtx{A}^*\mtx{A}\left(\vct{x}-\tilde{\vct{x}}\right)\\
\le&\kappa_f\tilde{\rho}(\mu)\twonorm{\vct{z}_\tau-\tilde{\vct{x}}}+\kappa_f\cdot\sup_{\vct{u}\in\mathcal{C}_f(\tilde{\vct{x}})\cap\mathbb{S}^{n-1}}\text{ }\vct{u}^*(\vct{x}-\tilde{\vct{x}})+\kappa_f\cdot\sup_{\vct{u}\in\mathcal{C}_f(\tilde{\vct{x}})\cap\mathbb{S}^{n-1}}\text{ }-\vct{u}^*\left(\mtx{I}-\mu\mtx{A}^*\mtx{A}\right)\left(\vct{x}-\tilde{\vct{x}}\right)\\
=&\kappa_f\tilde{\rho}(\mu)\twonorm{\vct{z}_\tau-\tilde{\vct{x}}}+\kappa_f\twonorm{\vct{x}-\tilde{\vct{x}}}+\kappa_f\cdot\sup_{\vct{u}\in\mathcal{C}_f(\tilde{\vct{x}})\cap\mathbb{S}^{n-1}}\text{ }-\vct{u}^*\left(\mtx{I}-\mu\mtx{A}^*\mtx{A}\right)\left(\vct{x}-\tilde{\vct{x}}\right)\\
\le&\kappa_f\tilde{\rho}(\mu)\twonorm{\vct{z}_\tau-\tilde{\vct{x}}}+\kappa_f\twonorm{\vct{x}-\tilde{\vct{x}}}+\kappa_f\left(\sup_{\vct{u},\vct{v}\in\mathcal{C}_f(\tilde{\vct{x}})\cap\mathbb{S}^{n-1}}\text{ }-\vct{u}^*\left(\mtx{I}-\mu\mtx{A}^*\mtx{A}\right)\vct{v}\right)\twonorm{\vct{x}-\tilde{\vct{x}}}\\
=&\kappa_f\tilde{\rho}(\mu)\twonorm{\vct{z}_\tau-\tilde{\vct{x}}}+\kappa_f(1+\tilde{\rho}(\mu))\twonorm{\vct{x}-\tilde{\vct{x}}}\\
\le&\kappa_f\tilde{\rho}(\mu)\left(\twonorm{\vct{z}_\tau-\vct{x}}+\twonorm{\vct{x}-\tilde{\vct{x}}}\right)+\kappa_f(1+\tilde{\rho}(\mu))\twonorm{\vct{x}-\tilde{\vct{x}}}\\
=&\kappa_f\tilde{\rho}(\mu)\twonorm{\vct{z}_\tau-\vct{x}}+\kappa_f\left(\frac{R}{f(\vct{x})}-1\right)\left(1+2\tilde{\rho}(\mu)\right)\twonorm{\vct{x}}.
\end{align*}
The proof is complete by using the triangular inequality
\begin{align*}
\twonorm{\vct{z}_{\tau+1}-\vct{x}}\le \twonorm{\vct{z}_{\tau+1}-\tilde{\vct{x}}}+\twonorm{\vct{x}-\tilde{\vct{x}}}\le\kappa_f\tilde{\rho}(\mu)\twonorm{\vct{z}_\tau-\vct{x}}+\left(\frac{R}{f(\vct{x})}-1\right)\left(\kappa_f+1+2\kappa_f\tilde{\rho}(\mu)\right)\twonorm{\vct{x}}.
\end{align*}
Finally, we note that when $f$ is a norm $\mathcal{C}_f(\tilde{\vct{x}})=\mathcal{C}_f(\vct{x})$ so that $\tilde{\rho}(\mu)=\rho(\mu)$.
\end{proof}


\subsection{Proof of Theorem \ref{nonGaussThm} (ISG ensembles)}
\label{app:subgaussian}
When the measurement matrix is sub-Gaussian (recall Definition \ref{subgauss def}), we have equivalent results with possibly larger constants that depend on the sub-Gaussian parameter. This directly follows from the results of Dirksen \cite{dirksen2014dimensionality}. For earlier results with a similar flavor, the reader is referred to Mendelson et al. \cite{UUP2} as well as \cite{Mendel1}. In particular, we use the following restatement of Theorem $4.18$ of \cite{dirksen2014dimensionality}.

\begin{lemma}\label{subgauss cor} Suppose $\A$ is a matrix with isotropic and i.i.d. sub-Gaussian rows with parameter $\Delta$. For a set $\mathcal{T}\subseteq r\Bc^{n-1}$
\beq
\sup_{\vb\in \mathcal{T}} \left|\vb^*(\Iden-\frac{\A^*\A}{m})\vb\right|\leq C_\Delta\frac{r(\omega(\mathcal{T})+r\eta)}{\sqrt{m}},\label{desired subgauss line}
\eeq
holds with probability at least $1-e^{-\gamma \eta^2}$. Here, $C_\Delta$ only depending on the sub-Gaussian parameter $\Delta$ and $\gamma$ is a fixed numerical constant.
\end{lemma}
\begin{proof}
We briefly describe how this lemma is a consequence of results in \cite{dirksen2014dimensionality}. This lemma follows from taking the square root of line (14) of \cite{dirksen2014dimensionality}, completing the squares and using the fact that Talagrand's $\gamma_2$ functional is equal to Gaussian width up to a multiplicative absolute constant (e.g.~see \cite{talagrand2006generic}).
\end{proof}
Applying this Lemma with different sets $\mathcal{T}$ of the form $\Cc\cap \mathbb{S}^{n-1}-\Cc\cap \mathbb{S}^{n-1},\Cc\cap \mathbb{S}^{n-1}+\Cc\cap \mathbb{S}^{n-1}$ we can deduce that, for $\mathcal{T}\in\{\Cc\cap \mathbb{S}^{n-1}-\Cc\cap \mathbb{S}^{n-1},\Cc\cap \mathbb{S}^{n-1}+\Cc\cap \mathbb{S}^{n-1}\}$ 
\begin{align*}
\sup_{\vb\in \mathcal{T}} \left|\vb^*(\Iden-\frac{\mtx{A}^*\A}{m})\vb\right|\leq 4C_\Delta\frac{(\omega\left(\Cc\cap\mathbb{S}^{n-1}\right)+\eta)}{\sqrt{m}}.
\end{align*}
hold with high probability for a constant $\tilde{c}(\Delta)$ that depends only on the sub-Gaussian constant $\Delta$. Observe that for these sets $r=2$. In summary, for our purposes a sub-Gaussian matrix behaves like a Gaussian up to a constant that depends only on the sub-Gaussian parameter. This is good enough to ensure that all of our results can be extended to sub-Gaussians with essentially no modification in our analysis. In particular, combining the arguments of Section \ref{secprooffirstprop} with lemma above we arrive at the following bound on the convergence rate
\begin{align*}
\rho(\frac{1}{m})=\sup_{\vct{u},\vct{v}\in\mathcal{C}}\text{ }\vct{u}^*\left(\Iden-\frac{\A^*\A}{m}\right)\vb\leq 4C_\Delta\sqrt{\frac{m_0}{m}}.
\end{align*}

\subsection{Proof of Theorem \ref{SOSthm} (SORS ensembles)}
\label{proofSORS}
Our proof strategy for this theorem is exactly the same as the proof of Theorem \ref{nonGaussThm} on ISG ensembles described in Section \ref{app:subgaussian}. However, we need a variant of Lemma \ref{subgauss cor} for SOS ensembles. Below we state this lemma which is proven in our companion paper \cite{oymak2015isometric}.
\begin{lemma} Suppose $\A\in\R^{m\times n}$ is selected from the SORS distribution of Definition \ref{SORSdef}. For a set $\mathcal{T}\subseteq r\Bc^{n-1}$
\begin{align*}
\sup_{\vb\in \mathcal{T}} \left|\vb^*(\Iden-\frac{\A^*\A}{m})\vb\right|\leq C \Delta r^2\frac{\max\left(1,\frac{\omega(\mathcal{T})}{r}\right)}{\sqrt{m}},
\end{align*}
holds with probability at least $1-e^{-2\eta}$ as long as
\begin{align*}
m\ge C'\Delta^2(\eta+1)^2(\log n)^4 \max\left(1,\frac{\omega^2(\mathcal{T})}{r^2}\right).
\end{align*}
Here $\Delta$ is the bound in \eqref{BOS} and $C,C'$ are fixed numerical constants.
\end{lemma}

\small
{
\bibliography{Bibfiles}
\bibliographystyle{plain}
}
\appendix

\section{Proof of Gordon type lemma (Lemma \ref{GTtypelem})}\label{proofGTtypelem}
Our proof is related to the proof of Gordon's celebrated escape through the mesh \cite[Theorem A]{Gor} stated in Lemma \ref{GTescapelem}. We will first show the bound $\twonorm{\mtx{A}\vct{u}}\le b_m\twonorm{\vct{u}}+\omega(\mathcal{T})+\eta$. For $\vct{u}\in\mathcal{T}$ and $\vct{v}\in\mathbb{S}^{m-1}=\{\vct{v}\in\R^m;\text{ }\twonorm{\vct{v}}=1\}$, we define three Gaussian processes
\begin{align*}
X_{\vct{u},\vct{v}}=\vct{v}^*\mtx{A}\vct{u},\quad Y_{\vct{u},\vct{v}}=\twonorm{\vct{u}}\vct{a}^*\vct{v}+\vct{g}^*\vct{u}\quad\text{and}\quad Z_{\vct{u},\vct{v}}=\twonorm{\vct{u}}(\vct{a}^*\vct{v}-b_m)+\vct{g}^*\vct{u}.
\end{align*}
Here $\vct{a}\in\R^m$ is distributed as $\mathcal{N}(\vct{0},\mtx{I}_m)$ and $\vct{g}\in\R^n$ is distributed as $\mathcal{N}(\vct{0},\mtx{I}_n)$.

It follows that for all $\vct{u},\vct{u}'\in\mathcal{T}$ and $\vct{v},\vct{v}'\in\mathbb{S}^{m-1}$, we have
\begin{align*}
\mathbb{E}\abs{Y_{\vct{u},\vct{v}}-Y_{\vct{u}',\vct{v}'}}^2-\mathbb{E}\abs{X_{\vct{u},\vct{v}}-X_{\vct{u}',\vct{v}'}}^2=&\twonorm{\vct{u}}^2+\twonorm{\vct{u}'}^2-2\twonorm{\vct{u}}\twonorm{\vct{u}'}\langle\vct{v},\vct{v}'\rangle-2\langle\vct{u},\vct{u}'\rangle\left(1-\langle\vct{v},\vct{v}'\rangle\right)\\
\ge&\twonorm{\vct{u}}^2+\twonorm{\vct{u}'}^2-2\twonorm{\vct{u}}\twonorm{\vct{u}'}\langle\vct{v},\vct{v}'\rangle-2\twonorm{\vct{u}}\twonorm{\vct{u}'}\left(1-\langle\vct{v},\vct{v}'\rangle\right)\\
\ge& 0,
\end{align*}
with equality if $\vct{u}=\vct{u}'$ and $\vct{v}=\vct{v}'$.

We note that by standard concentration of measure for Gaussian random variables
\begin{align*}
\mathbb{P}\Big\{\twonorm{\vct{a}}\ge \E[\twonorm{\vct{a}}]+\eta\Big\}\le e^{-\frac{\eta^2}{2}},
\end{align*}
We also have
\begin{align*}
\Big\{\vct{a}: \quad\twonorm{\vct{u}}\twonorm{\vct{a}}\ge \twonorm{\vct{u}}b_m+\eta\Big\}\subset&\Big\{\vct{a}: \quad\twonorm{\vct{u}}\twonorm{\vct{a}}\ge \twonorm{\vct{u}}b_m+\twonorm{\vct{u}}\frac{\eta}{\sigma(\mathcal{T})}\Big\}\\
=&\Big\{\vct{a}: \quad\twonorm{\vct{a}}\ge b_m+\frac{\eta}{\sigma(\mathcal{T})}\Big\}\\
=&\Big\{\vct{a}: \quad\twonorm{\vct{a}}\ge \E[\twonorm{\vct{a}}]+\frac{\eta}{\sigma(\mathcal{T})}\Big\}.
\end{align*}
Thus, 
\begin{align*}
\mathbb{P}\Bigg\{\underset{\vct{u}\in\mathcal{T}}{\bigcup}\big\{\vct{a}:\quad\twonorm{\vct{u}}\twonorm{\vct{a}}>\twonorm{\vct{u}} b_m+\eta_1\big\}\Bigg\}\le\mathbb{P}\Big\{\vct{a}:\text{ }\twonorm{\vct{a}}\ge \E[\twonorm{\vct{a}}]+\frac{\eta_1}{\sigma(\mathcal{T})}\Big\}\le e^{-\frac{\eta_1^2}{2\sigma^2(\mathcal{T})}},
\end{align*}
which immediately implies
\begin{align}
\label{myfirstineqGT}
\mathbb{P}\Big\{\underset{\vct{u}\in\mathcal{T}}{\max}\text{ }\twonorm{\vct{u}}\left(\twonorm{\vct{a}}-b_m\right)>\frac{\eta}{2}\Big\}\le e^{-\frac{\eta^2}{8\sigma^2(\mathcal{T})}}.
\end{align}
Also 
\begin{align}
\label{mysecondineqGT}
\mathbb{P}\Big\{\text{ }\underset{\vct{u}\in\mathcal{T}}{\max} \text{ }\left(\vct{g}^*\vct{u}\right)>\omega(\mathcal{T})+\frac{\eta}{2}\Big\}=\mathbb{P}\Big\{\text{ }\underset{\vct{u}\in\mathcal{T}}{\max} \text{ }\left(\vct{g}^*\vct{u}\right)>\E\big[\underset{\vct{u}\in\mathcal{T}}{\max} \text{ }\left(\vct{g}^*\vct{u}\right)\big]+\frac{\eta}{2}\Big\}\le e^{-\frac{\eta^2}{8\sigma^2(\mathcal{T})}}.
\end{align}
Note that if
\begin{align*}
\underset{\vct{u}\in\mathcal{T}}{\max}\text{ }\left(\twonorm{\vct{u}}\left(\twonorm{\vct{a}}-b_m\right)+\vct{g}^*\vct{u}\right)>\eta,
\end{align*}
then either
\begin{align*}
\underset{\vct{u}\in\mathcal{T}}{\max}\text{ }\twonorm{\vct{u}}\left(\twonorm{\vct{a}}-b_m\right)>\frac{\eta}{2},
\end{align*}
or
\begin{align*}
\underset{\vct{u}\in\mathcal{T}}{\max} \text{ }\left(\vct{g}^*\vct{u}\right)>\omega(\mathcal{T})+\frac{\eta}{2}.
\end{align*}
This implies that
\begin{align*}
\big\{\vct{a},\vct{g}:\text{ }\underset{\vct{u}\in\mathcal{T}}{\max}\text{ }\left(\twonorm{\vct{u}}\twonorm{\vct{a}}+\vct{g}^*\vct{u}-b_m\twonorm{\vct{u}}\right)>\omega(\mathcal{T})+\eta\big\}
\end{align*}
is a subset of 
\begin{align*}
\big\{\vct{a},\vct{g}:\text{ }\underset{\vct{u}\in\mathcal{T}}{\max}\text{ }\twonorm{\vct{u}}\left(\twonorm{\vct{a}}-b_m\right)>\frac{\eta}{2}\big\}\bigcup\big\{\vct{a},\vct{g}:\text{ }\underset{\vct{u}\in\mathcal{T}}{\max} \text{ }\left(\vct{g}^*\vct{u}\right)>\omega(\mathcal{T})+\frac{\eta}{2}\big\}.
\end{align*}
Using the latter together with \eqref{myfirstineqGT} and \eqref{mysecondineqGT} and using the independence of $\vct{a}$ and $\vct{g}$ we have
\begin{align}
\label{mythirdineqGT}
\mathbb{P}\Big\{\underset{\vct{u}\in\mathcal{T}}{\max}\text{ }\left(\twonorm{\vct{u}}\left(\twonorm{\vct{a}}-b_m\right)+\vct{g}^*\vct{u}\right)>\eta\Big\}\le& \mathbb{P}\Big\{\underset{\vct{u}\in\mathcal{T}}{\max}\text{ }\twonorm{\vct{u}}\left(\twonorm{\vct{a}}-b_m\right)>\frac{\eta}{2}\Big\}+\mathbb{P}\Big\{\text{ }\underset{\vct{u}\in\mathcal{T}}{\max} \text{ }\left(\vct{g}^*\vct{u}\right)>\omega(\mathcal{T})+\frac{\eta}{2}\Big\}\nonumber\\
\le& 2e^{-\frac{\eta^2}{8\sigma^2(\mathcal{T})}}.
\end{align}
Now note that
\begin{align*}
\underset{\vct{u}\in\mathcal{T},\text{ }\vct{v}\in\mathbb{S}^{n-1}}{\max}\text{ }Z_{\vct{u},\vct{v}}=\underset{\vct{u}\in\mathcal{T}}{\max}\text{ }\left(\twonorm{\vct{u}}\left(\twonorm{\vct{a}}-b_m\right)+\vct{g}^*\vct{u}\right).
\end{align*}
This together with \eqref{mythirdineqGT} implies
\begin{align*}
\mathbb{P}\big\{\underset{\vct{u}\in\mathcal{T},\text{ }\vct{v}\in\mathbb{S}^{n-1}}{\max}\text{ }Z_{\vct{u},\vct{v}}>\eta\big\}\le& 2e^{-\frac{\eta^2}{8\sigma^2(\mathcal{T})}}\quad\Rightarrow\quad \mathbb{P}\Big\{\underset{\vct{u}\in\mathcal{T},\vct{v}\in\mathbb{S}^{n-1}}{\bigcup}[Y_{\vct{u},\vct{v}}>b_m\twonorm{\vct{u}}+\eta]\Big\}\le 2e^{-\frac{\eta^2}{8\sigma^2(\mathcal{T})}}.
\end{align*}
Now using Slepian's second inequality \eqref{secondSlep} with $\eta_{\vct{u},\vct{v}}=b_m\twonorm{\vct{u}}+\eta$ we have
\begin{align*}
\mathbb{P}\Big\{\underset{\vct{u}\in\mathcal{T},\vct{v}\in\mathbb{S}^{n-1}}{\bigcup}[X_{\vct{u},\vct{v}}>b_m\twonorm{\vct{u}}+\eta]\Big\}\le\mathbb{P}\Big\{\underset{\vct{u}\in\mathcal{T},\vct{v}\in\mathbb{S}^{n-1}}{\bigcup}[Y_{\vct{u},\vct{v}}>b_m\twonorm{\vct{u}}+\eta]\Big\}\le 2e^{-\frac{\eta^2}{8\sigma^2(\mathcal{T})}}.
\end{align*}
Noting that
\begin{align*}
\mathbb{P}\Big\{\underset{\vct{u}\in\mathcal{T},\text{ }\vct{v}\in\mathbb{S}^{n-1}}{\max}\text{ }\mtx{X}_{\vct{u},\vct{v}}>\twonorm{\vct{u}}b_m+\eta\Big\}=\mathbb{P}\Big\{\underset{\vct{u}\in\mathcal{T},\text{ }\vct{v}\in\mathbb{S}^{n-1}}{\bigcup}[X_{\vct{u},\vct{v}}>b_m\twonorm{\vct{u}}+\eta]\Big\},
\end{align*}
concludes the proof. Next, we show that $\twonorm{\mtx{A}\vct{u}}\ge b_m\twonorm{\vct{u}}-\omega(\mathcal{T})-\eta$. To accomplish this, we make use of Lemma $5.1$ of \cite{OymLAS}. The following is an immediate corollary.
\begin{corollary} Let $\A\in\R^{m\times n},\g\in\R^n,\h\in\R^m$ be independent vectors with independent $\Nn(0,1)$ entries. Then, for any $c\in\R$
\beq
\Pro(\min_{\ub\in \Tc}\max_{\vb\in\Sc^{m-1}}\vb^*\A\ub-b_m\tn{\ub}\geq c)\geq 2\Pro(\min_{\ub\in\Tc}\max_{\vb\in\Sc^{m-1}}\vb^*\h\tn{\ub}-\ub^*\g\tn{\vb}-b_m\tn{\ub}\geq c).\label{two times}
\eeq
\end{corollary}
The right-hand side can be simplified to $\min_{\ub\in\Tc}(\tn{\h}-b_m)\tn{\ub}-\ub^*\g$, which is $\sqrt{2}\sigma(\Tc)$-Lipschitz function of the vector $\begin{bmatrix}\g^*&\h^*\end{bmatrix}$. As a result we have
\beq
\Pro(\min_{\ub\in\Tc}(\tn{\h}-b_m)\tn{\ub}-\ub^*\g\geq -\omega(\Tc)-\eta)\geq 1-\exp(-\frac{\eta^2}{4\sigma^2(\Tc)}).
\eeq
The proof is complete by combining the latter with \eqref{two times}.

\section{Proof of Lemma \ref{GammaComplem}}
\label{GammaComplemP}
To prove this lemma it suffices to show that $g(t):=\log\left(\frac{\phi(t)}{\sqrt{t}}\right)$ is increasing when $t\ge0$. To this aim note that
\begin{align*}
g(t)=\frac{1}{2}\log 2+\log\left(\Gamma\left(\frac{t+1}{2}\right)\right)-\log\left(\Gamma\left(\frac{t}{2}\right)\right)-\frac{1}{2}\log t.
\end{align*}
Thus,
\begin{align*}
g'(t)=\frac{1}{2}\left(\frac{\Gamma'\left(\frac{t+1}{2}\right)}{\Gamma\left(\frac{t+1}{2}\right)}-\frac{\Gamma'\left(\frac{t}{2}\right)}{\Gamma\left(\frac{t}{2}\right)}-\frac{1}{t}\right).
\end{align*}
Now using the well known fact that\footnote{For example see the wikipedia entry on the Digamma function (more specifically the part on integral representations).} 
\begin{align*}
\psi(t):=\frac{\Gamma'(t)}{\Gamma(t)}=\int_0^\infty \left(\frac{e^{-x}}{x}-\frac{e^{-tx}}{1-e^{-x}}\right)dx,
\end{align*}
together with the change of variables $y=x/2$ we have
\begin{align*}
g'(t)=&\frac{1}{2}\left(\psi\left(\frac{t+1}{2}\right)-\psi\left(\frac{t}{2}\right)-\frac{1}{t}\right)\\
=&\frac{1}{2}\left(\int_0^\infty e^{-t\frac{x}{2}}\frac{1-e^{-\frac{x}{2}}}{1-e^{-x}}dx-\frac{1}{t}\right)\\
=&\int_0^\infty e^{-ty}\frac{1-e^{-y}}{1-e^{-2y}}dy-\frac{1}{2t}\\
=&\int_0^\infty e^{-ty}\frac{1}{1+e^{-y}}dy-\frac{1}{2t}\\
=&\int_0^\infty e^{-ty}\frac{1}{1+e^{-y}}dy-\int_0^\infty e^{-2ty}dy\\
=&\int_0^\infty e^{-ty}\frac{1-e^{-ty}}{1+e^{-y}}dy\\
\ge&0,
\end{align*}
where the last inequality holds since the integrand is non-negative when $t\ge 0$ and $y\ge 0$. This completes the proof as we have shown that $g'(t)$ is non-negative for $t\ge 0$.

\end{document}